\documentclass[a4paper,reqno,10pt]{amsart}

\date{\today}

\usepackage{amsmath}
\usepackage{amsfonts}
\usepackage{amssymb}
\usepackage{amsthm}
\usepackage{amstext}
\usepackage{hyperref}
\usepackage{enumerate}
\usepackage{url}

\newcommand{\bx}{\mathbf{x}}
\newcommand{\by}{\mathbf{y}}
\newcommand{\bz}{\mathbf{z}}
\newcommand{\bp}{\mathbf{p}}
\newcommand{\bq}{\mathbf{q}}

\newcommand{\bP}{\mathbf{P}}

\newcommand{\C}{\mathbb{C}}
\newcommand{\R}{\mathbb{R}}

\newcommand{\Tr}{\operatorname{Tr}}

\newcommand{\supp}{{\operatorname{supp}}}

\newcommand{\Real}{{\operatorname{Re}}}

\newtheorem{lemma}{Lemma}[section]
\newtheorem{corollary}[lemma]{Corollary}
\newtheorem{theorem}[lemma]{Theorem}

\newtheorem{proposition}[lemma]{Proposition}
\newtheorem{remark}[lemma]{Remark}
\newtheorem{definition}[lemma]{Definition}

\newtheorem*{acknowledgement}{Acknowledgement}

\newenvironment{pf*}[1]{\par\medskip\noindent\textit{#1}\,:}{\hspace*{\fill}\qed\medskip\par\noindent}

\title{Excess charge for pseudo-relativistic atoms in Hartree-Fock theory}
\author[Anna Dall'Acqua, Jan Philip Solovej]{Anna Dall'Acqua, Jan Philip Solovej}

\address[Anna Dall'Acqua]
{Institute for Analysis and Numerics,
Faculty of Mathematics,
Otto-von-Guericke Universit\"at,
Universit\"atsplatz 2,
D-39016 Magdeburg, Germany}
\email{anna.dallacqua@ovgu.de}

\address[Jan Philip Solovej]
{Department of Mathematics,
University of Copenhagen,
Universitetsparken 5,
DK-2100 Copenhagen, Denmark}
\email{solovej@math.ku.dk}

\begin{document}

\begin{abstract}
We prove within the Hartree-Fock theory of pseudo-relativistic atoms that the maximal negative ionization charge and the ionization energy of an
atom remain bounded independently of the nuclear charge $Z$ and the fine structure constant $\alpha$ as long as $Z\alpha$ is bounded.
\end{abstract}

\subjclass[2000]{81Q05, 81Q20}

\maketitle

\tableofcontents

\section{Introduction}

A long standing open problem in the mathematical physics literature is the Ionization conjecture. It can be formulated as follows. Consider atoms with arbitrarily large nuclear charge $Z$, is it true that the radius (see Definition \ref{defradius}) and the maximal negative ionization remain bounded? A positive answer to this question in the non-relativistic Hartree-Fock model has been given by the second author in \cite{Sol}. One of the aims of the present paper is to extend the result taking into account some relativistic effects. The ionization conjecture for the full Schr\"odinger theory is still open both in the non-relativistic and relativistic case. See \cite{L}, \cite{LSST1}, \cite{LSST2}, \cite{FS1}, \cite{FS2} and \cite{S} for some $Z$-dependent bounds on the maximal negative ionization. The best result is that $N(Z)=Z +O(Z^{a})$ with $a=47/56$ where $N(Z)$ denotes the maximal number of electrons a nucleus of charge $Z$ binds (see \cite{FS1}, \cite{FS2} and \cite{S}).

As a model for an atom with nuclear charge $Z$ and $N$ electrons we consider (in units where $\hbar = m= e=1$) the operator
\begin{equation}
H = \sum_{i=1}^{N}\alpha ^{-1}\big( \sqrt{-\Delta _{i}+\alpha ^{-2}}-\alpha ^{-1}-\frac{Z\alpha }{\vert \bx_{i}\vert}\big)+\sum_{1\leq i<j\leq N}\frac{1}{\vert \bx_{i}-\bx_{j}\vert },  \label{Hamiltonian}
\end{equation}
where $\alpha $ is Sommerfeld's fine structure constant. The operator $H$ acts on a dense subset of the $N$ body Hilbert space $\mathcal{H}_{F}:=\displaystyle{\wedge_{i=1}^{N}}L^{2}(\R^{3};\mathbb{C}^{q})$ of antisymmetric wave functions, where $q$ is the number of spin states. The operator $H$ is bounded from below on this subspace if $Z\alpha \leq 2/\pi$ (see \cite{H} for $N=1$, \cite{DL} and \cite{LY} for $N \geq 1$). In this paper we will consider the sub-critical case $Z \alpha <2/\pi$. Let us notice here that to define the operator $H$ there is an issue. Indeed for $Z \alpha <2/\pi$ the nuclear potential is only a small form perturbation of the kinetic energy and hence one needs to work with forms to define the operator $H$. This has been done in detail in \cite{DSS1}.

The quantum ground state energy is the infimum of the spectrum of $H$ considered as an operator acting on $\mathcal{H}_{F}$. In the Hartree-Fock approximation one restricts to wave-functions $\psi$ which are pure wedge products, also called Slater determinants:
\begin{equation}\label{slater}
\psi (\bx_{1},\sigma _{1},\bx_{2},\sigma _{2},\dots ,\bx_{N},\sigma _{N})=\tfrac{1}{\sqrt{N!}}\det (u_{i}(\bx_{j},\sigma_{j}))_{i,j=1}^{N},
\end{equation}
with $\{u_{i}\}_{i=1}^{N}$ orthonormal in $L^{2}(\R^{3};\mathbb{C}^{q})$. The $u_i$'s are also called orbitals. Notice that $\Vert \psi \Vert _{L^{2}(\R^{3N},\mathbb{C}^{qN})}=1$. The Hartree-Fock ground state energy is
\begin{equation*}
E^{\rm HF}(N,Z,\alpha):= \inf \{\mathrm{q}(\psi, \psi) | \psi \in \mathcal{Q}(H)  \mbox{ and } \psi \mbox{ a Slater determinant} \},
\end{equation*}
with $\mathtt{q}$ the quadratic form defined by $H$ and $\mathcal{Q}(H)$ the corresponding form domain.

One of the main result of the paper is the following.
\begin{theorem}\label{charge}
Let $Z\geq 1$ and $\alpha>0$. Let $Z \alpha =\kappa$ and assume that $0 \leq \kappa<2/\pi$. There is a constant $Q>0$ depending only on $\kappa$ such that if $N$ is such that a Hartree-Fock minimizer exists then $ N \leq Z+Q$.
\end{theorem}

The idea of the proof is the same as in \cite{Sol}. One shows that the Thomas-Fermi model is a good approximation of the Hartree-Fock model except in the region far away from the nucleus. We first introduce some notation in order to introduce the Hartree-Fock and Thomas-Fermi models.

\subsection{Notation}

Let $\mathrm{e}$ be the quadratic form with domain $H^{\frac{1}{2}}(\mathbb{R}^3,\mathbb{C}^q)$ such that
\begin{equation}\label{a}
\mathrm{e}(u,v)= (E(\bp)^{\frac12}u,E(\bp)^{\frac12}v) \mbox{ for all }u,v \in H^{\frac12}(\R^3,\C^q),
\end{equation}
where $E(\bp)$ denotes the operator $E(i \nabla)=\sqrt{-\Delta + \alpha^{-2}}$. As usual $(u,v)$ denotes the scalar product of $u$ and $v$ in $L^2(\mathbb{R}^3,\mathbb{C}^q)$. Let $V(\bx):= Z \alpha /|\bx|$ and $\mathrm{v}$ be the quadratic form with domain $H^{\frac{1}{2}}(\mathbb{R}^3,\mathbb{C}^q)$ defined by
\begin{equation}\label{b}
\mathrm{v}(u,v)=(V^{\frac12}u, V^{\frac12} v) \mbox{ for all }u,v \in H^{\frac12}(\mathbb{R}^3,\mathbb{C}^q).
\end{equation}
{F}rom \cite[5.33 p.307]{Kato} we have
\begin{equation}\label{kato}
\int_{\R^3} \frac{|f(\bx)|^2}{|\bx|} \; d\bx \leq \frac{2}{\pi} \int_{\R^3} |\bp| |\hat{f}(\bp)|^2 \; d\bp \mbox{ for } f\in H^{\frac12}(\R^3, \C)
\end{equation}
with $\hat{f}$ the Fourier transform of $f$. Thus since $Z \alpha \leq 2/\pi$ and $E(\bp) \geq |\bp|$ it follows that $\mathrm{v}(u,u) \leq \mathrm{e}(u,u)$ for all $u \in H^{\frac12}(\R^3, \C^q)$.

In the following $\mathrm{t}$ denotes the quadratic form associated to the kinetic energy; i.e. for all $u, v \in H^{\frac12}(\R^3,\C^q)$
\begin{equation}\label{quadkin}
\mathrm{t}(u,v):= \alpha^{-1} \mathrm{e}(u,v)-\alpha^{-2}(u,v)= \alpha^{-1}(T(\bp)^{\frac12}u,T(\bp)^{\frac12}v),
\end{equation}
with $T(\bp):=E(\bp)-\alpha^{-1}$.

A \textit{density matrix} $\gamma $ is a self-adjoint trace class operator that satisfies the operator inequality $0\leq \gamma \leq \mbox{\textit{Id}}$. A density matrix $\gamma :L^{2}(\mathbb{R}^{3};\mathbb{C}^{q})\rightarrow
L^{2}(\mathbb{R}^{3};\mathbb{C}^{q})$ has an integral kernel
\begin{equation}
\gamma \left( \mathbf{x},\sigma ,\mathbf{y},\tau \right) =\sum_{j}\lambda
_{j}u_{j}(\mathbf{x},\sigma )u_{j}(\mathbf{y},\tau )^{\ast },  \label{ker}\end{equation}
where $\lambda _{j},u_{j}$ are the eigenvalues and corresponding eigenfunctions of $\gamma $. We choose the $u_j$'s to be orthonormal in $L^2(\mathbb{R}^3,\mathbb{C}^q)$. Let $\rho_{\gamma} \in L^{1}(\R^3)$ denote the $1$-particle density associated to $\gamma$ given by
\begin{equation*}
\rho _{\gamma }(\mathbf{x}) =\sum_{\sigma =1}^{q}\sum_{j}\lambda_{j} \vert u_{j}(\bx ,\sigma )\vert ^{2} .
\end{equation*}

We define
\begin{equation}
\mathcal{A}:=\left\{ \gamma \text{ density matrix: } \Tr [ T(\bp) \gamma]<+\infty \right\} ,  \label{A}
\end{equation}
where for $\gamma \in \mathcal{A}$ written as in (\ref{ker}) $\Tr [T(\bp) \gamma]:=\Tr [E(\bp) \gamma]-\alpha^{-1}\Tr[\gamma]$ and
\begin{equation}\label{kintr}
\Tr [E(\bp) \gamma]: =\sum_{j}\lambda _{j} \mathrm{e}(u_j,u_j).
\end{equation}
Similarly we use the following notation $\Tr\left[ V \gamma \right] :=\sum_{j}\lambda _{j} \mathrm{v}(u_j,u_j)$.

\begin{remark}\label{rho43}
If $\gamma \in \mathcal{A}$ then $\rho_{\gamma} \in L^1(\mathbb{R}^3)$ since $\gamma$ is trace class and $\rho_{\gamma} \in L^{4/3}(\mathbb{R}^3)$. The second inclusion follows from Daubechies' inequality, a generalization of the Lieb-Thirring inequality (see Theorem~\ref{Daube}).
\end{remark}

\subsection{Hartree-Fock theory\label{HFtheory}}

In Hartree-Fock theory one considers wave functions that are pure wedge products and that satisfy the right statistic: determinantal wave functions as in \eqref{slater}. To define the HF-energy functional it is convenient to use the one to one correspondence between Slater determinants and projections onto finite dimensional subspaces of $L^2(\mathbb{R}^3, \mathbb{C}^q)$. Indeed if $\psi$ is given by (\ref{slater}) and $\gamma $ is the projection onto the space spanned by $u_{1}$, $\dots ,u_{N} $ the energy expectation depends only on $\gamma$: $\left( \psi ,H\psi \right) =\mathcal{E}^{\rm HF}(\gamma ).$ Here $\mathcal{E}^{\rm HF}$ defines the HF-energy functional
\begin{equation}\label{HFfunc}
\mathcal{E}^{\rm HF}(\gamma )=\alpha^{-1}\Tr[ (T(\bp)-V)\gamma ]+ \mathcal{D}(\gamma )- \mathcal{E}x\left( \gamma \right) ,
\end{equation}
where $\mathcal{D}(\gamma )$ is the direct Coulomb energy
\begin{equation*}
\mathcal{D}(\gamma)=\tfrac12 \int_{\R^3} \int_{\R^3} \frac{\rho_{\gamma}(\bx) \rho_{\gamma}(\by)}{|\bx -\by|} \; d\bx d\by,
\end{equation*}
and $\mathcal{E}x( \gamma) $ is the exchange Coulomb energy
\begin{equation*}
\mathcal{E}x( \gamma) =\tfrac{1}{2}\int_{\R^{3}}\int_{\R^{3}}\frac{\Tr_{\mathbb{C}^{q}} \big[\vert \gamma (\bx,\by)\vert ^{2}\big] }{\vert \bx -\by\vert } \; d\bx d\by ,
\end{equation*}
where we think of the integral kernel $\gamma(x,y)$ as a $q \times q$ matrix.

Using projections we can define as follows the HF-ground state.

\begin{definition}[The HF-ground state]
Let $Z>0$ be a real number and $N\geq 0$ be an integer. The HF-ground state energy is
\begin{equation*}
E^{\rm HF}(N,Z,\alpha):=\inf \left\{ \mathcal{E}^{\rm HF}(\gamma ):\gamma ^{2}=\gamma , ~ \gamma \in \mathcal{A} ,~\Tr[\gamma ]=N\right\}.
\end{equation*}
If a minimizer exists we say that the atom has a HF ground state described by $\gamma^{\rm HF}.$
\end{definition}

We may extend the definition of the HF-functional from projections to density matrices in $\mathcal{A}$. We first notice that if $\gamma \in \mathcal{A}$, then all the terms in $\mathcal{E}^{\rm HF}(\gamma)$ are finite. {F}rom \eqref{kato} it follows that
\begin{equation*}
\Tr[V \gamma] = \sum_{j} \lambda_{j} \mathrm{v}(u_{j}, u_{j}) \leq \sum_{j} \lambda_{j} \mathrm{e}(u_{j}, u_{j}) = \Tr[E(\bp) \gamma] .
\end{equation*}
On the other hand if $\gamma \in \mathcal{A}$ then $\rho_{\gamma} \in L^{1}(\R^3) \cap L^{\frac43}(\R^3)$ (see Remark~\ref{rho43}). By H\"older's inequality $\rho_{\gamma} \in L^{\frac65}(\R^3)$ and hence $\mathcal{D}(\gamma)$ is bounded by Hardy-Littlewood-Sobolev's inequality. The boundness of the exchange term follows from $0 \leq \mathcal{E}x (\gamma) \leq \mathcal{D}(\gamma)$. On the other hand if $\gamma$ is a density matrix with $\gamma \notin \mathcal{A}$ then $\mathcal{E}^{\rm HF}(\gamma)=\infty$. Here we use also that $Z \alpha <2/\pi$.

Extending the set where we minimize, we could have lowered the ground state energy and/or changed the minimizer. That this is not the case follows from Lieb's variational principle.

\begin{theorem}[Lieb's variational principle, \cite{lvar}]\label{varprin}
For all $N$ non-negative integers it holds that
\begin{equation*}
\inf \{ \mathcal{E}^{\rm HF}(\gamma ):\gamma  \in \mathcal{A},~\gamma ^{2}=\gamma ,~\Tr[\gamma ]=N\} =\inf \{ \mathcal{E}^{\rm HF}(\gamma ):\gamma \in \mathcal{A},~\Tr[\gamma ]=N\} ,
\end{equation*}
and if the infimum over all density matrices is attained so is the
infimum over projections.
\end{theorem}

The following existence theorem for the HF-minimizer in the pseudo-relativistic case has been recently proved in \cite{DSS1}.

\begin{theorem}\label{HF}
Let $Z\alpha <2/\pi $ and let $N\geq 2$ be a positive integer such that $N<Z+1$.

Then there exists an $N$-dimensional projection $\gamma ^{\rm HF}=\gamma^{\rm HF}(N,Z,\alpha)$ minimizing the HF-energy functional $\mathcal{E}^{\rm HF}$ given by \eqref{HFfunc}, that is, $E^{\rm HF}(N,Z,\alpha)$ is attained. Moreover, one can write
\begin{equation*}
\gamma ^{\rm HF}(\bx,\sigma ,\by,\tau )=\sum_{i=1}^{N}u_{i}(\bx,\sigma) u_{i}( \by ,\tau ) ^{\ast },
\end{equation*}
with $u_{i} \in L^{2}(\R^3, \C^q)$, $i=1, \dots,N$, orthonormal, such that the HF-orbitals $\{ u_{i}\}_{i=1}^{N}$ satisfy:
\begin{enumerate}
\item  $h_{\gamma^{\rm HF}} u_{i} =\varepsilon _{i}u_{i},$ with $0 > \varepsilon _{N} \geq \varepsilon_{N-1} \geq \dots \geq \varepsilon_{1} > -\alpha^{-1}$ and
\begin{equation}
h_{\gamma^{\rm HF}}:= T(\bp) -\frac{Z\alpha}{\vert \bx\vert }+ \rho ^{\rm HF}\ast \vert \bx\vert ^{-1}-\mathcal{K}_{\gamma ^{\rm HF}},  \label{euler}
\end{equation}
where $\rho^{\rm HF}$ denotes the density of the HF-minimizer and for $f \in H^{\frac12}(\R^3)$
\begin{equation*}
( \mathcal{K}_{\gamma ^{\rm HF}}f) (\bx,\sigma)=\sum_{i=1}^{N} u_{i}( \bx,\sigma) \sum_{\tau =1}^{q} \int_{\R^{3}}u_{i}( \by,\tau) ^{\ast }f(\by,\tau )\vert \bx-\by\vert ^{-1}d\by .
\end{equation*}
\item $u_i \in C^{\infty}(\R^3 \setminus \{0\}, \C^q)$ for $i=1, \dots, N$;
\item $u_i \in H^1(\R^3 \setminus B_R(0))$ for all $R>0$ and $i=1, \dots , N$.
\end{enumerate}
\end{theorem}

In the opposite direction the following result gives an upper bound on the excess charge.

\begin{theorem}
\label{Lieb}Let $\alpha Z<\frac{2}{\pi}$. If $N$ is a positive integer such that $N>2Z+1$ there are no minimizers for the HF-energy functional.
\end{theorem}

This theorem for $Z \alpha <1/2$ was proved by Lieb in \cite{L}. With an improved approximation argument the proof can be extended to $Z\alpha < 2/\pi$ (see \cite{DSS2}). Notice that both proofs work not only in the Hartree-Fock approximation but for the minimization problem on $\wedge^{N} L^2(\R^3)$.

\begin{definition}\label{screened}
Let $\gamma ^{\rm HF}$ be the HF-minimizer. The function
\begin{equation*}
\varphi ^{\rm HF}( \bx) :=\frac{Z}{\vert \bx\vert }-\int_{\R^3} \frac{\rho ^{\rm HF}(\by) }{\vert\bx-\by\vert }d\by \mbox{ for } \bx \in \R^3,
\end{equation*}
is called the HF-mean field potential and
\begin{equation*}
\Phi _{R}^{\rm HF}(\bx) :=\frac{Z}{\vert \bx\vert }-\int_{\vert \by\vert <R}\frac{\rho^{\rm HF}(\by) }{\vert \bx-\by\vert }d\by \mbox{ for } \bx \in \R^3,
\end{equation*}
is the HF-screened nuclear potential.
\end{definition}

\begin{definition}\label{defradius}
We define the HF-radius $R^{\rm HF}_{Z,N}(\nu)$ to the $\nu$ last electrons by
\begin{equation*}
\int_{|\bx| \geq R^{\rm HF}_{Z,N}(\nu)} \rho^{\rm HF}(\bx) \, d\bx = \nu.
\end{equation*}
\end{definition}

\subsection{A bit of Thomas-Fermi theory}

In this subsection we present briefly the Thomas-Fermi theory and especially the result that will be used in the rest of the paper. We refer the interested reader to \cite{Ltf}.

Let $U$ be a potential in $L^{5/2}(\R^{3})+L^{\infty }(\R^{3})$ with
\begin{equation*}
\inf\{ \Vert W\Vert _{\infty }: U-W\in L^{\frac{5}{2}}(\R^{3})\} =0.
\end{equation*}
Then the TF-energy functional is defined by
\begin{equation*}
\mathcal{E}_{U}^{\rm TF}( \rho) =\tfrac{3}{10} (\tfrac{6\pi^{2}}{q})^{\frac{2}{3}}\int_{\R^{3}}\rho( \bx)^{\frac{5}{3}}d\bx-\int_{\R^{3}}U( \bx) \rho (\bx) d\bx + \tfrac12 \int_{\R^3} \int_{\R^3} \frac{\rho(\bx) \rho(\by)}{|\bx -\by|} \; d\bx d\by ,
\end{equation*}
on non-negative functions $\rho \in L^{5/3}(\R^{3})\cap L^{1}(\R^{3})$. As before, $q$ denotes the number of spin states.

We recall some properties of the TF-model, see \cite{LSTF}.

\begin{theorem}\label{ExTF}
Let $U$ be as above. For all $N^{\prime }\geq 0$ there exists a unique non-negative $\rho _{U}^{\rm TF}\in L^{5/3}( \R^{3}) $ such that $\int \rho _{U}^{\rm TF}\leq N^{\prime }$ and
\begin{equation*}
\mathcal{E}_{U}^{\rm TF}(\rho _{V}^{\rm TF}) =\inf \{ \mathcal{E}_{U}^{\rm TF}(\rho) :\rho \in L^{5/3}(\R^{3}),\text{ }\int_{\R^3}\rho(\bx) \, d\bx \leq N^{\prime }\} .
\end{equation*}
There exists a unique chemical potential $\mu_{U}^{\rm TF}(N^{\prime }),$ with $0\leq \mu _{U}^{\rm TF}(N^{\prime })\leq \sup U,$ such that $\rho _{U}^{\rm TF}$ is uniquely characterized by
\begin{eqnarray*}
& & \mathcal{E}_{U}^{\rm TF}(\rho _{U}^{\rm TF}) +\mu_{U}^{\rm TF}(N^{\prime }) \, \int_{\R^3} \rho_{U}^{\rm TF}(\bx) \, d\bx\\
&=&\inf \{ \mathcal{E}_{U}^{\rm TF}( \rho) +\mu _{U}^{\rm TF}(N^{\prime }) \; \int_{\R^3} \rho(\bx) \, d\bx :0\leq \rho \in L^{5/3}(\R^{3})\cap L^{1}(\R^{3})\} .
\end{eqnarray*}
Moreover $\rho _{U}^{\rm TF}$ is the unique solution in $L^{5/3}(\R^{3})\cap L^{1}(\R^{3})$ to the TF-equation
\begin{equation*}
\tfrac{1}{2} (\tfrac{6\pi^{2}}{q})^{\frac{2}{3}}( \rho _{U}^{\rm TF}(\bx)) ^{\frac{2}{3}}=\left[ U(\bx) -\rho_{U}^{\rm TF}\ast \vert \bx\vert ^{-1}-\mu _{U}^{\rm TF}(N^{\prime})\right] _{+}.
\end{equation*}
If $\mu _{U}^{\rm TF}(N^{\prime })>0$ then $\int \rho _{U}^{\rm TF}=N^{\prime }.$ For all $\mu >0$ there is a unique minimizer $0\leq \rho \in L^{5/3}(\R^{3})\cap L^{1}(\R^{3})$ to $\mathcal{E}_{U}^{\rm TF}(\rho) +\mu \int \rho$.
\end{theorem}

One defines the \textit{TF-mean field potential} $\varphi _{U}^{\rm TF}$, the \textit{TF-screened nuclear potential} $\Phi _{U,R}^{\rm TF}$ and the TF-radius $R^{\rm TF}_{N,Z}(\nu)$ to the $\nu$ last-electron similarly as in Definitions~\ref{screened} and \ref{defradius} replacing the HF-density with the TF-density.

\begin{theorem}\label{TfCou}
If $U(\bx) =Z/\vert \bx\vert $ (the Coulomb potential), then the minimizer of $\mathcal{E}_{U}^{\rm TF},$ under the condition $\int \rho \leq N,$ exists for every $N$. Moreover, $\mu_{U}^{\rm TF}(N)=0$ if and only if $N\geq Z$.
\end{theorem}

When $U(\bx) =Z/\vert \bx\vert $ we denote the minimizer of the TF-functional, under the condition $\int \rho\leq Z,$ simply by $\rho ^{\rm TF}$ and $\int \rho^{\rm TF}=Z$. Correspondingly $\varphi ^{\rm TF}$ and $\Phi _{R}^{\rm TF}$ denote, respectively, its mean field and screened nuclear potential. With this notation
\begin{equation}\label{Tfe}
\mathcal{E}^{\rm TF}(\rho^{\rm TF}) = -  e_0 Z^{\frac73},
\end{equation}
where $e_0$ is the total binding energy of a neutral TF-atom of unit nuclear charge.

We recall here a result due to Sommerfeld on the asymptotic behavior of the TF-mean field potential, see \cite[Th. 4.6]{Sol}.

\begin{theorem}[Sommerfeld asymptotics]\label{Somm}
Assume that the potential $U$ is continuous and harmonic for $|\bx|>R$ and that it satisfies $ \lim_{|\bx|\rightarrow \infty} U(\bx)=0$.

Consider the corresponding TF-mean field potential $\varphi _{U}^{\rm TF}$ and assume that $\mu_{U}^{\rm TF} < \displaystyle{\liminf_{r \searrow R} \inf_{|\bx|=r}}\varphi _{U}^{\rm TF}(\bx)$. With $\zeta = (-7+\sqrt{73})/2$ define
\begin{eqnarray*}
a(R) &:= &\liminf_{r \searrow R} \sup_{|\bx|=r} \Big[\Big(\frac{\varphi_{U}^{\rm TF}(\bx)}{3^4 2^{-1} q^{-2} \pi^2 r^{-4}} \Big)^{-\frac12} -1 \Big] r^{\zeta}\\
A(R, \mu_{U}^{\rm TF} ) & := & \liminf_{r \searrow R} \sup_{|\bx|=r} \Big[\frac{\varphi_U^{\rm TF}(\bx) -\mu_{U}^{\rm TF}}{3^4 2^{-1} q^{-2} \pi^2 r^{-4}} -1 \Big] r^{\zeta} .
\end{eqnarray*}
Then we find for all $|\bx|> R$
\begin{eqnarray*}
\varphi_U^{\rm TF}( \bx) & \leq & \tfrac{3^4 \pi^2}{ 2 q^2} (1+A(R,\mu_U^{\rm TF}) |\bx|^{-\zeta }) |\bx|^{-4}+ \mu_{U}^{\rm TF} \; \; \mbox{ and }\\
\varphi_U^{\rm TF}( \bx) & \geq & \max\Big\{\tfrac{3^4 \pi^2}{ 2 q^2} (1+a(R) |\bx|^{-\zeta })^{-2} |\bx|^{-4}, \nu(\mu_U^{\rm TF}) |\bx|^{-1} \Big\},
\end{eqnarray*}
where
\begin{equation*}
 \nu(\mu_U^{\rm TF}) := \inf_{|\bx|\geq R} \max\big\{\tfrac{3^4 \pi^2}{ 2 q^2} (1+a(R) |\bx|^{-\zeta })^{-2}|\bx|^{-3}, \mu_U^{\rm TF} |\bx|  \big\}.
\end{equation*}
\end{theorem}

For easy reference we give here the estimate on the TF-mean field potential corresponding to the Coulomb potential.

\begin{theorem}[Atomic Sommerfeld estimate, \protect{\cite[Thm 5.2-5.4]{Sol}}]\label{Slb}
The atomic TF-mean field potential satisfies the bound
\begin{equation}\label{Tfpou}
\frac{Z}{|\bx|} -\min \Big\{ \frac{Z}{|\bx|}, \frac{Z^{\frac43}}{2 \beta_0} \Big\} \leq \varphi ^{\rm TF}( \bx) \leq \min \Big\{ \tfrac{3^{4}\pi ^{2}}{2 q^{2}}\frac{1}{\vert \bx\vert^{4}},  \frac{Z}{\vert \bx\vert} \Big\} ,
\end{equation}
with $2 \beta_0 = \pi^{\frac23} 3^{-\frac53} 2^{-\frac13} q^{-\frac23}$, and for $|\bx| \geq R >0$
\begin{equation*}
\varphi ^{\rm TF}( \bx) \geq \tfrac{3^4 \pi^2}{ 2 q^2} (1+a(R) |\bx|^{-\zeta })^{-2} |\bx|^{-4},
\end{equation*}
where $\zeta$ and $a(R)$ are defined in Theorem~\ref{Somm}.
\end{theorem}

\begin{corollary}\label{TFsotto}
Let $\zeta$ and $\beta_0$ be defined as in Theorem~\ref{Somm} and \ref{Slb} respectively. Then the TF-mean field potential satisfies the bound
\begin{equation*}
\varphi^{\rm TF}(\bx) \geq \left\{ \begin{array}{ll}
\displaystyle{\frac{Z}{|\bx|} - \frac{Z^{\frac43}}{2 \beta_0}} & \mbox{ if }|\bx| \leq \beta_0 Z^{-\frac13} \vspace{.2cm}\\
\displaystyle{\frac{3^4 \pi^2}{2 q^2} (1+a Z^{-\frac{\zeta}{3}}|\bx|^{-\zeta})^{-2} |\bx|^{-4}} & \mbox{ if }|\bx| > \beta_0 Z^{-\frac13} ,
\end{array} \right.
\end{equation*}
with $a=\beta_0^{\zeta} (3^2 \pi/(q \beta_0^{\frac32})-1)$.
\end{corollary}

\begin{corollary}\label{TFscr}
The TF-screened nuclear potential satisfies
\begin{equation*}
\Phi^{\rm TF}_{|\bx|}(\bx) \leq \tfrac{3^4 2 \pi^2}{q^2} |\bx|^{-4} \; \mbox{ for all } \bx \in \R^3.
\end{equation*}
\end{corollary}

\begin{corollary}\label{rTF53}
The following estimate holds
\begin{equation*}
\int_{\R^{3}}(\rho ^{\rm TF}(\bx)) ^{\frac{5}{3}}d\bx\leq 4 \tfrac{2^{\frac23}}{\pi ^{2}}\tfrac{5}{7} q^{\frac43} Z^{\frac{7}{3}}.
\end{equation*}
\end{corollary}

\begin{proof}
By the TF-equation and since $\mu ^{\rm TF}=0$ we find
\begin{equation*}
\int_{\R^{3}}( \rho ^{\rm TF}(\bx)) ^{\frac{5}{3}}d\bx=2^{\frac52}(\tfrac{q}{6 \pi^2})^{\frac53} \int_{\R^{3}}( \varphi^{\rm TF}(\bx)) ^{\frac{5}{2}}d\bx.
\end{equation*}
The estimate follows from the atomic Sommerfeld upper bound.
\end{proof}

\subsection{Construction and main results}

We present the basic idea for the proof of Theorem~\ref{charge}. Let us consider an atomic system with $N \geq 2$ fermionic particles and a nucleus of charge $Z\geq 1$ with $Z \alpha =\kappa$ and $0 \leq \kappa <2/\pi$. We assume that $N\geq Z$ and that $N$ is such that a HF-minimizer exists. That is: there exists a density matrix $\gamma ^{\rm HF}\in \mathcal{A}$ such that $\Tr[\gamma ^{\rm HF}]=N$ and
\begin{equation*}
\mathcal{E}^{\rm HF}(\gamma ^{\rm HF})=\inf \left\{ \mathcal{E}^{\rm HF}(\gamma ):\gamma =\gamma ^{\ast },0\leq \gamma \leq I,\Tr[\gamma ]=N\right\} .
\end{equation*}
Let $\rho ^{\rm TF}$ be the TF-minimizer with potential $U(\bx)=Z/\vert \bx \vert $ and under the condition $\int \rho^{\rm TF}=Z.$ We know that such a minimizer exist and that the corresponding chemical potential is zero (see Theorem~\ref{TfCou}).

Denoting by $\rho ^{\rm HF}$ the density of the minimizer $\gamma ^{\rm HF},$ we find for all $r>0$
\begin{eqnarray*}
N &=&\int_{\mathbb{R}^{3}}\rho ^{\rm HF}(\bx)d\bx \\
&=&\int_{\vert \bx\vert <r}\left[ \rho^{\rm HF}(\bx)-\rho^{\rm TF}(\bx)\right] d\bx+\int_{\vert \bx\vert <r}\rho ^{\rm TF}(\bx) \; d\bx+\int_{\vert \bx\vert >r}\rho ^{\rm HF}(\bx) \; d\bx .
\end{eqnarray*}

By the equalities above and since $\int_{\vert \bx\vert <r}\rho ^{\rm TF}(\bx)d\bx\leq Z$, Theorem~\ref{charge} follows from the following result.

\begin{theorem}
There exist $r>0$ and positive constants $c_{1}$ and $c_{2}$ independent of $N$ and $Z$ but possibly depending on $\kappa$ such that
\begin{equation*}
\int_{\vert \bx \vert <r}\left[ \rho ^{\rm HF}(\bx)-\rho^{\rm TF}(\bx)\right] d\bx \leq c_{1}\text{ and }\int_{\vert \bx\vert >r}\rho ^{\rm HF}(\bx)d\bx\leq c_{2}.
\end{equation*}
\end{theorem}

The following theorem is the principal ingredient in the proof of the previous one and is the main technical estimate in the paper.

\begin{theorem}\label{mainesti}
Let $Z\alpha =\kappa$, $0 \leq \kappa< 2/\pi$. Assume $N \geq Z \geq 1$.

Then there exist universal constants $\alpha_{0}>0$, $0<\varepsilon <4$ and $C_M$ and $C_{\Phi}$ depending on $\kappa$ such that for all $\alpha \leq \alpha_{0}$
\begin{equation*}
\left|\Phi^{\rm HF}_{|\bx|}(\bx)- \Phi^{\rm TF}_{|\bx|}(\bx)\right| \leq C_{\Phi} |\bx|^{-4+\varepsilon}+C_{M}.
\end{equation*}
\end{theorem}
This main estimate is proven by an iterative procedure. We first prove the estimate for small $\bx$ (i.e. $|\bx| \leq \beta_{0} Z^{-\frac13}$), then for intermediate $\bx$ (i.e. up to a fixed distance independent of $Z$) and finally for big $\bx$.

By proving Theorem~\ref{mainesti} we also get the following interesting results. The proofs of those are given in Section~\ref{secproofs}.

\begin{theorem}[Asymptotic formula for the radius]\label{radius}
Let $Z\alpha=\kappa$, $0 \leq \kappa< 2/\pi$. Both $\liminf_{Z \rightarrow \infty} R^{\rm HF}_{Z,Z}(\nu)$ and $\limsup_{Z \rightarrow \infty} R^{\rm HF}_{Z,Z}(\nu)$ are bounded and behave asymptotically as
\begin{equation*}
3^{\frac43} \frac{2^{\frac12}\pi^{\frac23}}{q^{\frac23}} \nu^{-\frac13} + o(\nu^{-\frac13}) \mbox{ as }\nu \rightarrow \infty .
\end{equation*}
\end{theorem}

\begin{theorem}[Bound on the ionization energy of a neutral atom]\label{Ionenergy}
Let $Z\alpha=\kappa$, $0 \leq \kappa< 2/\pi$ and $Z \geq 1$. The ionization energy of a neutral atom $E^{\rm HF}(Z-1,Z) - E^{\rm HF}(Z,Z)$ is bounded by a universal constant.
\end{theorem}

\begin{theorem}[Potential estimate]\label{thpot}
Let $Z\alpha=\kappa$, $0 \leq \kappa< 2/\pi$. For all $Z \geq 1$ and $N$ with $N \geq Z$ for which a HF minimizer exists with $\int \rho^{\rm HF}=N$, we have
\begin{equation*}
|\varphi^{\rm TF}(\bx)-\varphi^{\rm HF}(\bx)| \leq A_{\varphi} |\bx|^{-4+\varepsilon_{0}} + A_{1},
\end{equation*}
with $A_{0}, A_{1}$ and $\varepsilon_{0}$ universal constants.
\end{theorem}

\section{Prerequisites}

In this section we recall some results that will be used in the rest of the paper.

\textit{Localization of the kinetic energy}. The following is the IMS formula corresponding to the operator $T(\bp)$.

\begin{theorem}[\cite{LY}]\label{IMSrel}
Let $\chi _{i},$ $i=0,\dots ,K,$ be real valued Lipschitz continuous functions on $\mathbb{R}^{3}$ such that $\sum_{i=0}^{K}\chi_{i}^{2}\left( \mathbf{x}\right) =1$ for all $\mathbf{x}\in \mathbb{R}^{3}$. Then for every $f\in H^{1/2}(\mathbb{R}^{3})$
\begin{equation*}
\mathrm{t}(f,f)=\sum_{i=0}^{K} \mathrm{t}(\chi _if, \chi _{i}f) -\alpha^{-1} \sum_{i=0}^{K}( f,L_{i} f) ,
\end{equation*}
where $L_{i}$ is a bounded operator with kernel
\begin{equation}
L_{i}(\bx ,\by)=\tfrac{\alpha^{-2}}{4\pi ^{2}} \frac{\vert \chi _{i}( \bx) -\chi_{i}( \by)\vert^{2}}{\vert \bx -\by\vert^{2}}K_{2}(\alpha ^{-1}\vert \bx-\by\vert),  \label{Lalpha}
\end{equation}
where $K_2$ is a modified Bessel function of the second kind.
\end{theorem}

\begin{remark}
As in \cite[App.A, pages 94--98]{SSS} we use the following integral formula for the modified Bessel function
\begin{equation*}
K_2(t)= t \int_{0}^{\infty} e^{-t\sqrt{s^2+1}} s^2 \; ds \, , \quad t>0.
\end{equation*}
We recall that this function is decreasing and smooth in $\R^{+}$. Moreover,
\begin{equation}
\int_{0}^{+\infty }t^{2}K_{2}(t)~dt=\tfrac{3\pi }{2} \;  \mbox{ and } \;
K_{2}\left( t\right) \leq 16~t^{-2}e^{-\frac{1}{2}t} \mbox{ for }t>0.  \label{estk2}
\end{equation}
The integral is computed in \cite[(A6)]{TS} while the estimate follows directly from the integral formula for $K_2$ by estimating $\sqrt{s^2+1}\geq \frac12 + \frac12 s $.
\end{remark}

\textit{Generalization of the Lieb-Thirring inequality.} This result due to Daubechies generalizes the Lieb-Thirring inequality to the pseudo-relativistic case.

\begin{theorem}[Daubechies' inequality, \cite{D}]\label{Daube}
For $\gamma \in \mathcal{A}$
\begin{equation*}
\Tr[T(\bp) \gamma] \geq \int_{\R^{3}}G_{\alpha}(\rho_{\gamma}(\bx)) d\bx,
\end{equation*}
where $G_{\alpha}( \rho) =\tfrac{3}{8}\alpha ^{-4}Cg( \alpha (\rho/C) ^{\frac{1}{3}}) -\alpha ^{-1}\rho $ with $C=.163q$, $q$ the number of spin states and $g(t) =t( 1+t^{2}) ^{\frac{1}{2}}(1+2t^{2}) -\ln ( t+( 1+t^{2}) ^{\frac{1}{2}})$.
\end{theorem}

\begin{remark}\label{Daure}
The function $G_{\alpha}$ defined in the previous theorem is convex and it has the following behavior:
\begin{equation}
\tfrac{9}{20}\min \left\{ \tfrac{1}{5}\alpha C^{-\frac{2}{3}}\rho ^{\frac{5}{3}},\tfrac{1}{2}C^{-\frac{1}{3}}\rho ^{\frac{4}{3}}\right\} \leq G_{\alpha}\left( \rho \right) \leq \tfrac{3}{2}\min \left\{ \tfrac{1}{5}\alpha C^{-\frac{2}{3}}\rho ^{\frac{5}{3}},\tfrac{1}{2}C^{-\frac{1}{3}}\rho ^{\frac{4}{3}}\right\} .  \label{BehG}
\end{equation}
(The proof of the estimate above is in Appendix A.) Notice that when $\alpha \searrow 0$ then $\alpha^{-1} G_{\alpha}(\rho)$ tends to a constant times $\rho^{5/3}$.
\end{remark}

\begin{theorem}[Generalization of the Lieb-Thirring inequality, \cite{D}]\label{DaubeV}
Let $f^{-1}$ be the inverse of the function $f(t) :=\sqrt{t^{2}+\alpha^{-2}}-\alpha ^{-1}$, $t\geq 0$, and define $F(s)=\int_{0}^{s}dt~ [ f^{-1}(t)]^{3}$. Then for any density matrix $\gamma $ it holds
\begin{equation*}
\Tr[(T(\bp)-U) \gamma] \geq -Cq\int_{\R^3}F(\vert U(\bx)\vert) d\bx,
\end{equation*}
with $C\leq 0.163$.
\end{theorem}

\begin{remark}\label{DaubeR}
Since $f^{-1}(t) =( t^{2}+2\alpha ^{-1}t)^{1/2}$ we find for $F$
\begin{equation}\label{Fs}
F(s)=2^{\frac32} \alpha^{-3/2}\int_{0}^{s}t^{3/2}\big( 1+\tfrac{1}{2} \alpha t \big)^{3/2} \; dt \; \; \mbox{ for } s \geq 0,
\end{equation}
and since by convexity $(1+\frac12 \alpha t)^{\frac32} \leq \sqrt{2}+\frac12 (\alpha t)^{\frac32}$ we have
\begin{equation*}
F(s)\leq \tfrac{2^3}{5} \alpha^{-\frac{3}{2}} s^{\frac{5}{2}}+\tfrac{1}{2\sqrt{2}}s^{4} \mbox{ for } s \geq 0.
\end{equation*}
Hence for any density matrix $\gamma $ and potential $U\in L^{\frac{5}{2}}\left( \mathbb{R}^{3}\right) \cap L^{4}\left( \mathbb{R}^{3}\right) $
\begin{equation}
\Tr[ (T(\bp)-U) \gamma] \geq -Cq\int_{\R^{3}} \Big(\tfrac{2^{3}}{5} \alpha ^{-\frac{3}{2}} \vert U(\bx)\vert^{\frac{5}{2}}+\tfrac{1}{2\sqrt{2}} \vert U(\bx)\vert^{4}) d\bx.  \label{Dau}
\end{equation}
\end{remark}

\textit{Coulomb norm estimate}. We present here only the definition of Coulomb norm and the result we need. For a more complete presentation we refer to \cite[Sec.9]{Sol}.

\begin{definition}\label{Cnorm}
For $f,g\in L^{\frac{6}{5}}(\R^{3})$ we define the Coulomb inner product
\begin{equation*}
D(f,g) :=\tfrac{1}{2}\int_{\R^{3}}\int_{\R^{3}}\frac{f( \bx) \overline{g(\by)}}{\vert\bx-\by\vert }d\bx d\by ,
\end{equation*}
and the corresponding norm $\Vert g\Vert _{C}:=D( g,g) ^{\frac{1}{2}}$.
\end{definition}

In the following we write the direct term in the HF-energy functional using the Coulomb scalar product: i.e. $\mathcal{D}(\gamma)=D(\rho_{\gamma}, \rho_{\gamma})=D(\rho_{\gamma})$. Similarly, for $\rho \in L^{1}(\R^3) \cap L^{\frac53}(\R^3)$ the term $D(\rho)$ denotes $D(\rho , \rho) $.

The next proposition follows as Corollary 9.3 in \cite{Sol}.

\begin{proposition}\label{Counorm}
For $s>0$, $\bx \in \R^{3}$ and $f \in L^{\frac65}(\R^3)$ it holds
\begin{equation*}
f *\vert \bx\vert^{-1} \leq \int_{\vert \bx -\by\vert<s} [ f(\by)] _{+}\Big(\frac{1}{\vert \bx-\by\vert}-\frac{1}{s}\Big)d\by+ \sqrt{2} \, s^{-\frac12} \Vert f\Vert _{C}.
\end{equation*}
Moreover, for $k>0$
\begin{equation*}
\int_{\vert \by\vert <\vert \bx\vert }\frac{f(\by)}{\vert \bx-\by\vert}d\by\leq \int_{A(\vert \bx\vert,k)}\frac{[ f(\by)] _{+}}{\vert \bx-\by\vert }d\by+ 2^{\frac{3}{2}}k^{-1} \vert \bx\vert ^{-\frac{1}{2}}\Vert f\Vert _{C},
\end{equation*}
where $A(\vert \bx\vert ,k)$ denotes the annulus
\begin{equation*}
A(\vert \bx\vert ,k) :=\left\{ \by \in \R^{3}:( 1-2k) \vert \bx\vert \leq \vert \by\vert \leq \vert \bx\vert \right\} .
\end{equation*}
\end{proposition}

\subsection{Improved relativistic Lieb-Thirring inequalities}

A major difference between the pseudo-relativistic HF-model and the non-relativistic one studied in \cite{Sol} is that the HF-density $\rho^{\rm HF}$ in the pseudorelativistic case is not in $L^{\frac53}(\R^3)$. By Theorem~\ref{Daube} and Remark~\ref{Daure} we see that $\rho^{\rm HF}$ is only in $L^{\frac43}(\R^3)$. Therefore one cannot estimate the term $\rho^{\rm HF}*|\bx|^{-1}$ in $L^{1}$-norm simply by H\"older's inequality with $p=5/2$ and $q=5/3$. To estimate it we are going to use a combined Daubechies-Lieb-Yau inequality.

The following lemma can be found in \cite[pages 98--99]{SSS}\footnote{The result of the lemma and the proof given in \cite{SSS} are actually due to us, but we communicated the result to the authors of \cite{SSS}, where it is referred to as a a private communication.}.
\begin{lemma}\label{Lemnice}
For $f \in \mathcal{S}(\R^3)$,
\begin{equation*}
\int_{\mathbb{R}^{3}}\frac{e^{-\mu \vert \bx\vert ^{2}}}{\vert \bx\vert }\vert f(\bx)\vert ^{2}d\bx \leq \frac{\pi}{2} \frac{1}{\sqrt{2}-1} \;  (f, T(\bp )f),
\end{equation*}
with $\mu=\pi^{-1} \alpha^{-2}$.
\end{lemma}

The following is a slight generalization of the
Daubechies-Lieb-Yau inequality formulated in Theorem 2.8 in
\cite{SSS}.
\begin{theorem}[Daubechies-Lieb-Yau inequality]
\label{dayau} Assume that the potential $U \in L^1_{loc}(\R^3)$ satisfies
\begin{equation}\label{eq:Uassumption}
0 \geq - U(\bx) \geq -\kappa |\bx|^{-1} \; \; \mbox{ for }
|\bx| <  \max\{\alpha,R\} \, ,
\end{equation}
for $\alpha,R>0$ and $0\leq\kappa\leq2/\pi$. Then we have
\begin{equation*}
\Tr[T(\bp)-U]_- \geq -C\kappa^{5/2}\alpha^{-3/2}R^{1/2} -C\kappa^4\alpha^{-1}
 - C \int_{|\bx|> R} \big( \alpha^{-\frac32} |U(\bx)|^{\frac52} + |U(\bx)|^{4}\big)
\; d\bx .
\end{equation*}
\end{theorem}
\begin{proof}
If $(\sqrt{2}-1)/\pi\leq \kappa\leq 2/\pi$ then
$\kappa^{5/2}\alpha^{-3/2}R^{1/2} +\kappa^4\alpha^{-1}\geq
C\kappa^{5/2}\alpha^{-1} $
and the result follows immediately from Theorem 2.8 in \cite{SSS} observing that for $R>\alpha$ the two integrals of the potential on $\{\alpha <|\bx|<R\}$ are bounded by the constants.

If $0\leq \kappa < (\sqrt{2}-1)/\pi$ we write
$$
U(x)=e^{-\mu|x|^2}U(x)\chi_{|\bx|<R}
+(1-e^{-\mu|x|^2})U(x)\chi_{|\bx|<R}+U(x)\chi_{|\bx|>R}
$$
with $\mu=\alpha^{-2}\pi^{-1}$.
Using (\ref{eq:Uassumption}) and Lemma~\ref{Lemnice} we find that
$$
T(\bp)-U(\bx)\geq \frac12 T(\bp)
-\kappa(1-e^{-\mu|\bx|^2})|\bx|^{-1}\chi_{|\bx|<R}
-U(\bx)\chi_{|\bx|>R}.
$$
Hence from the generalization of the Lieb-Thirring inquality Theorem~\ref{DaubeV} (see \eqref{Dau}) we obtain
\begin{eqnarray*}
  \Tr[T(\bp)-U]_- &\geq&
  -C \int_{|\bx|< R}
  \alpha^{-\frac32}
  \big(\kappa(1-e^{-\mu|\bx|^2})|\bx|^{-1})\big)^{\frac52}  \; d\bx \\&&
  -C\int_{|\bx|< R} \big(\kappa(1-e^{-\mu|\bx|^2})|\bx|^{-1}\big)^{4}
  \; d\bx \\&&
  - C \int_{|\bx|> R}
  \big( \alpha^{-\frac32} |U(\bx)|^{\frac52} + |U(\bx)|^{4}\big)
  \; d\bx .
\end{eqnarray*}
Since the two first integrals above are estimated below by
$-C\kappa^{5/2}\alpha^{-3/2}R^{1/2} -C\kappa^4\alpha^{-1}$ we get the
result in the theorem.
\end{proof}

By Theorem~\ref{dayau} we find
\begin{equation}\label{inema}
\kappa \int_{|\bx-\by|<R} \frac{\rho^{\rm HF}(\by)}{|\bx -\by|} \, d\by  \leq \Tr[T(\bp) \gamma^{\rm HF}]+ C_1 \kappa Z^{\frac32}R^{\frac{1}{2}}+ C_2 \kappa^{3} Z,
\end{equation}
with $\kappa \in [0,2/\pi]$, $\kappa=Z\alpha$ and $R>0$ parameters to be chosen. This is the inequality that we use to estimate $\rho^{\rm HF}*|\bx|^{-1}$ (see proof of Lemma~\ref{smallx} below).

\subsubsection{Bound on the Hartree-Fock energy}

As a first application of Theorem~\ref{dayau} we can give a lower bound to the HF-energy.
\begin{theorem}[Bound on the HF-energy]\label{HFen}
Let $N>0$, $Z>0$ and such that $Z \alpha = \kappa$ with $0\leq \kappa\leq 2/\pi$. Then
\begin{equation*}
E^{\rm HF}(N,Z) \geq - 2C^{\frac23} Z^2 N^{\frac13} -C \kappa^{2} Z^2,
\end{equation*}
with $C$ the constant in Theorem~\ref{dayau}.
\end{theorem}
\begin{proof}
Let $\gamma$ be a $N$-dimensional projection. Since the electron-electron iteraction is positive we see that
\begin{eqnarray*}
\mathcal{E}^{\rm HF}(\gamma) & \geq & \alpha^{-1} \Tr[(T(\bp)-\frac{Z\alpha}{|\cdot|})\gamma]\\
& = & \alpha^{-1} \Tr[(T(\bp)-\frac{\kappa}{|\cdot|}\chi_{|\bx|<R})\gamma] -\alpha^{-1} \Tr[\frac{\kappa}{|\cdot|}(1-\chi_{|\bx|<R})\gamma]
\end{eqnarray*}
with $R>0$ a parameter to be choosen. By Theorem~\ref{dayau} we find
\begin{eqnarray*}
\mathcal{E}^{\rm HF}(\gamma) & \geq & - 2C^{\frac23} Z^2 N^{\frac13} -C \kappa^{2} Z^2,
\end{eqnarray*}
using that $\kappa = Z \alpha$ and by choosing $R=C^{-\frac23} Z^{-1} N^{\frac23}$.
\end{proof}

\section{Near the nucleus}

In this section we prove the estimate in Theorem~\ref{mainesti} in the region near the nucleus (i.e. at distance of $Z^{-\frac13}$).

We again assume that $N\geq Z$ and that an HF-minimizer $\gamma^{\rm
HF}$ exists for this $N$ and $Z$. We denote the density of $\gamma^{\rm
HF}$ by $\rho^{\rm HF}$. We assume throughout that $\alpha Z=\kappa$
is fixed with $0 \leq \kappa<2/\pi$ and $Z \geq 1$.

\begin{lemma}\label{rrcou}
Let $Z\alpha=\kappa$ be fixed with $0 \leq \kappa < 2/\pi$ and $Z \geq 1$. Let $G_{\alpha}$ be the function defined in Theorem~\ref{Daube}. Then, there exists $\alpha_0 >0$ such that for all $\alpha \leq \alpha_0$
\begin{eqnarray}\label{G1}
\left. \begin{array}{l}
\displaystyle{\alpha ^{-1}\int_{\R^3} G_{\alpha}( \rho ^{\rm HF}(\bx)) d\bx \leq C Z^{7/3}, \; \; \;  \alpha ^{-1} \Tr[T(\bp)\gamma^{\rm HF}] \leq C Z^{7/3}} \vspace{.3cm} \\
\hspace{2cm} \mbox{ and }\; \; \; \displaystyle{\Vert \rho ^{\rm TF}-\rho ^{\rm HF}\Vert^2_{C} \leq C  Z^{2+\frac{3}{11}}},
\end{array} \right.
\end{eqnarray}
with $C$ a universal constant depending only on $\kappa$.
\end{lemma}

\begin{proof} Let $\mu \in (0,1)$ be such that $\mu^{-1} \kappa < 2/\pi$. Notice that here we need $\kappa < 2/\pi$. Splitting the kinetic energy into two parts we find
\begin{eqnarray*}
\mathcal{E}^{\rm HF}( \gamma ^{\rm HF}) &=&  (1-\mu) \alpha^{-1}\Tr[T(\bp)\gamma^{\rm HF}] +\mathcal{D}(\gamma^{\rm HF})-\mathcal{E}x (\gamma^{\rm HF})\\
&&  + \mu  \Tr[(\alpha^{-1}T(\bp)-\frac{Z}{\mu \vert \bx \vert }) \gamma ^{\rm HF}]=\dots ,
\end{eqnarray*}
 and introducing $\rho \in L^{\frac53}(\R^3)\cap L^1(\R^3),$ $\rho \geq 0$, to be chosen
\begin{eqnarray} \label{pr2}
\dots  &=&  (1-\mu) \alpha^{-1}\Tr[T(\bp) \gamma^{\rm HF}] + \mu \| \rho-\rho^{\rm HF}\|_{C}^2+(1-\mu) \mathcal{D}(\gamma^{\rm HF})\\ \notag
&& -\mathcal{E}x (\gamma^{\rm HF}) -\mu D(\rho)+ \mu  \Tr[(\alpha^{-1}T(\bp)-\big(\frac{Z}{\mu \vert \bx \vert }-\rho*\frac{1}{|\bx|}\big)) \gamma ^{\rm HF}].
\end{eqnarray}
Here $\| \cdot \|_{C}$ denotes the Coulomb norm defined in Definition~\ref{Cnorm} and we used that
\begin{equation*}
 \| \rho-\rho^{\rm HF}\|_{C}^2=D(\rho)-\iint \frac{\rho^{\rm HF}(\bx) \rho(\by)}{|\bx-\by|} \; d\bx d\by +\mathcal{D}(\gamma^{\rm HF}).
\end{equation*}

The estimates in the claim will follow from \eqref{pr2} with different choices of $\mu$ and $\rho$. The main idea is to relate, up to lower order term,
the last term on the right hand side of \eqref{pr2} to the TF-energy of a neutral atom of nuclear charge $Z \mu^{-1}$. This has been done in \cite{TS}.
For completeness and easy reference we repeat the reasoning in Propositions~\ref{prthomas} and \ref{prthomas3} in Appendix~\ref{ZZ}.

To prove the first inequality in \eqref{G1} we choose $\rho$ as the minimizer of the TF-energy functional of a neutral atom with charge $\mu^{-1}Z$.
Since the corresponding TF-mean field potential is $Z/(\mu|\bx|)-\rho*1/|\bx|$ by Proposition~\ref{prthomas3} in Appendix~\ref{ZZ} we find
\begin{eqnarray}\label{pr3}
\Tr[(\alpha^{-1}T(\bp)-(\frac{Z}{\mu \vert \bx \vert }-\rho*\frac{1}{|\bx|})) \gamma ^{\rm HF}] \geq -C_{1} Z^{\frac73} + D(\rho).
\end{eqnarray}
Here we use \eqref{Tfe}. Since $\mathcal{E}^{\rm HF}(\gamma^{\rm HF}) \leq 0$ from \eqref{pr2} and \eqref{pr3} leaving out the positive terms we find
\begin{eqnarray} \label{pri10}
0 &\geq &  (1-\mu) \alpha^{-1}\Tr[T(\bp) \gamma^{\rm HF}] -\mathcal{E}x (\gamma^{\rm HF}) -C_1 Z^{\frac73}.
\end{eqnarray}
{F}rom \eqref{pri10} and Theorem~\ref{Daube} we get
\begin{equation} \label{pri1}
(1-\mu) \alpha^{-1}\int_{\R^3} G_{\alpha}( \rho ^{\rm HF}(\bx)) \, d\bx \leq (1-\mu) \alpha^{-1}\Tr[T(\bp) \gamma^{\rm HF}] \leq \mathcal{E}x (\gamma^{\rm HF}) +C_1 Z^{\frac73}.
\end{equation}
It remains to estimate the exchange term. By the exchange inequality (see \cite{LO})
\begin{eqnarray*}
\mathcal{E}x(\gamma^{\rm HF}) &\leq &1.68\int_{\R^3} \big(\rho^{\rm HF}(\bx)\big)^{\frac43} \, d\bx .
\end{eqnarray*}
To proceed we separate $\R^3$ into two regions. Let us define
\begin{equation}\label{Sigma}
\Sigma = \big\{ \bx \in \R^3: \alpha \big(C^{-1}\rho^{\rm HF} (\bx)\big)^{\frac13} \geq \tfrac52 \big\} ,
\end{equation}
with the same notation as in \eqref{BehG}. By Remark~\ref{Daure}, $G_{\alpha}(\rho^{\rm HF}(\bx)) \geq C_2 (\rho^{\rm HF}(\bx))^{\frac43}$ in $\Sigma$ and $\alpha^{-1} G_{\alpha}(\rho^{\rm HF}(\bx)) \geq C_3 (\rho^{\rm HF}(\bx))^{\frac53}$ in $\R^3 \setminus \Sigma$. Hence by H\"older's inequality we find
\begin{eqnarray}\notag
\mathcal{E}x( \gamma ^{\rm HF}) &\leq &1.68\int_{\Sigma} ( \rho^{\rm HF}(\bx)) ^{\frac43} \; d\bx \\ \notag
&& + 1.68 \Big(\int_{\R^3 \setminus \Sigma} ( \rho^{\rm HF}(\bx))^{\frac53} \, d\bx \Big)^{\frac12} \Big(\int_{\R^3 \setminus \Sigma}  \rho^{\rm HF} (\bx) \, d\bx \Big)^{\frac12}\\ \label{pri4}
& \leq & C_4 \int_{\R^3} G_{\alpha}(\rho^{\rm HF}(\bx)) \, d\bx + C_5 \Big(\int_{\R^3} \alpha^{-1} G_{\alpha}(\rho^{\rm HF}(\bx)) \, d\bx \Big)^{\frac12} N^{\frac12}.
\end{eqnarray}
Choosing $\alpha_0$ such that $1-\mu >2 C_4 \alpha$ for $\alpha \leq \alpha_0$, from \eqref{pri1} and \eqref{pri4} we find
\begin{eqnarray*}
\tfrac{1-\mu}{2} \alpha^{-1}\int_{\R^3} G_{\alpha}( \rho ^{\rm HF}(\bx)) \, d\bx \leq C_1 Z^{\frac73} + C_5 \Big(\int_{\R^3} \alpha^{-1} G_{\alpha}(\rho^{\rm HF}(\bx)) \, d\bx \Big)^{\frac12} N^{\frac12} .
\end{eqnarray*}
The first estimate in \eqref{G1} follows from the estimate above using that $x^2 -bx-c \leq 0$ implies $x^2 \leq b^2 + 2c$ and that $N \leq 2Z+1$ (Theorem~\ref{Lieb}). The second inequality in \eqref{G1} follows then from \eqref{pri1} and the bound on the exchange term.

To prove the third inequality in \eqref{G1} we estimate from above and from below $\mathcal{E}^{\rm HF}(\gamma^{\rm HF})$. For the one from below we choose in \eqref{pr2} $\mu=1$ and $\rho=\rho^{\rm TF}$ the TF-minimizer of a neutral atom with nucleus of charge $Z$. We find
\begin{equation}\label{r1}
\mathcal{E}^{\rm HF}( \gamma ^{\rm HF})= \sum_{i=1}^N (u_i , ( \alpha^{-1}T(\bp) -\varphi ^{\rm TF}) u_i )+\Vert \rho^{\rm HF}-\rho ^{\rm TF}\Vert _{C}^{2}- D(\rho^{\rm TF}) - \mathcal{E}x(\gamma ^{\rm HF}) .
\end{equation}
{F}rom \eqref{r1} and the proof of Proposition~\ref{prthomas3} (see \eqref{p4}), we find
\begin{eqnarray}\label{z2}
\mathcal{E}^{\rm HF}(\gamma ^{\rm HF}) & \geq & -\tfrac{2^{\frac32}}{15 \pi^2}  q \int d\bq (\varphi^{\rm TF}(\bq))^{\frac{5}{2}} - C Z^{2+1/5}\\ \notag
&& -D(\rho^{\rm TF}) +\Vert \rho^{\rm HF}-\rho ^{\rm TF}\Vert _{C}^{2}- \mathcal{E}x(\gamma ^{\rm HF}).
\end{eqnarray}

To estimate from above $\mathcal{E}^{\rm HF}(\gamma^{\rm HF})$ we may proceed exactly as in \cite[page 543]{Sol} using that $\alpha^{-1}T(\bp)\leq \frac12 |\bp|^2$. For completeness we repeat the main ideas. We consider $\gamma$ the density matrix that acts identically on each of the spin components as
\begin{equation*}
\gamma^{j} =\tfrac{1}{( 2\pi) ^{3}} \iint_{\frac{1}{2}\vert \bp\vert ^{2} \leq \varphi ^{\rm TF}( \bq)}\Pi_{\bp,\bq} \; d\bq d\bp \mbox{ for }j=1, \dots,q.
\end{equation*}
Here $\Pi_{\bp,\bq}$ is the projection onto the space spanned by $h_{s}^{\bp,\bq}(\bx):=h_{s}(\bx-\bq) e^{i \bp.\bx}$ where $h_{s}$ is the ground state (normalized in $L^2(\R^3)$) for the Dirichlet Laplacian on the ball of radius $Z^{-s}$ with $s\in (1/3,2/3)$ to be chosen. One sees that $\Tr[\gamma]=Z \leq N$ since
\begin{equation*}
\rho _{\gamma }(\bx) =\tfrac{2^{3/2}q}{6 \pi^2}( \varphi ^{\rm TF}) ^{3/2}\ast  h^2_{s}(\bx) = \rho ^{\rm TF}\ast h^2_{s}(\bx) ,
\end{equation*}
where we have used the TF-equation. Hence $\mathcal{E}^{\rm HF}(\gamma) \geq \mathcal{E}^{\rm HF}( \gamma ^{\rm HF})$. Now we estimate from above $\mathcal{E}^{\rm HF}(\gamma)$. Since $\alpha^{-1} T(\bp) \leq \frac12 |\bp|^2$ and $\mathcal{E}x(\gamma) \geq 0$ we find
\begin{equation*}
\mathcal{E}^{\rm HF}( \gamma) \leq \Tr[(-\tfrac12 \Delta-\frac{Z}{|\cdot |}) \gamma] +D(\rho _{\gamma }) =\dots ,
\end{equation*}
and proceeding as in \cite[page 543]{Sol})
\begin{equation*}
\dots = \tfrac{q}{(2 \pi)^3} \iint_{\frac12 |\bp|^2 \leq \varphi^{\rm TF}(\bq)} \tfrac12 |\bp|^2 \; d\bp d\bq - \tfrac{\pi^2}{2} Z^{2s} N - \int_{\R^3} \frac{Z}{|\bx|} \rho_{\gamma} (\bx) \, d\bx +D(\rho _{\gamma }) .
\end{equation*}
Computing the integral and summing and subtracting the term $\int \rho^{\rm TF} \varphi^{\rm TF}$ we get
\begin{eqnarray}\notag
\mathcal{E}^{\rm HF}( \gamma) & \leq & \tfrac{q 2^{\frac12}}{5\pi^2} \int_{\R^3} (\varphi^{\rm TF}(\bq))^{\frac52} \, d\bq - \tfrac{\pi^2}{2} Z^{2s} N -\int_{\R^3} \varphi^{\rm TF}(\bx) \rho^{\rm TF}(\bx) \, d\bx \\ \label{pa}
&&- \int_{\R^3} \frac{Z}{|\bx|} (\rho_{\gamma}(\bx)-\rho^{\rm TF}(\bx)) d \bx -2 D(\rho^{\rm TF})+D(\rho_{\gamma}).
\end{eqnarray}
By Newton's theorem one sees that $D(\rho_{\gamma}) \leq D(\rho^{\rm TF})$ and that\begin{equation*}
Z \int_{\R^3} \frac{\rho^{\rm TF}(\bx)-\rho_{\gamma}(\bx)}{|\bx|} \, d\bx \leq Z \int_{|\bx | \leq Z^{-s}} \frac{\rho^{\rm TF}(\bx)}{|\bx|} \, d\bx \leq C Z^{\frac15(12-s)}.
\end{equation*}
In the last step we use H\"older's inequality and Corollary~\ref{rTF53}. {F}rom \eqref{pa} using the TF-equation, that $N \leq 2Z+1$ (Theorem~\ref{Lieb}) and optimizing in $s$ we find
\begin{equation}\label{z1}
\mathcal{E}^{\rm HF}(\gamma) \leq -\tfrac{2^{\frac32}}{15\pi^2} q\int_{\R^3} (\varphi^{\rm TF}(\bq))^{\frac52} \, d\bq + C Z^{\frac15(12-\frac{7}{11})}-D(\rho^{\rm TF}).
\end{equation}
Hence from \eqref{z2} and \eqref{z1} we obtain
\begin{equation*}
\left\Vert \rho ^{\rm HF}-\rho ^{\rm TF}\right\Vert _{C}^{2}\leq C Z^{2+\frac{3}{11}}+\mathcal{E}x(\gamma^{\rm HF}) .
\end{equation*}
The last estimate in \eqref{G1} follows from the estimate above since $\mathcal{E}x(\gamma^{\rm HF}) \leq C Z^{\frac53}$ using \eqref{pri4} and the estimate just proved on $\alpha^{-1} \int G_{\alpha}(\rho^{\rm HF}(\bx)) \; d \bx$.
\end{proof}

\begin{lemma}\label{smallx}
Let  $Z \alpha = \kappa$ be fixed with $0 \leq \kappa < 2/\pi$ and $Z \geq 1$. Then, there exists an $\alpha_{0}>0$ such that for all $\alpha \leq \alpha_{0}$, $\mu >0$ and $\bx \in \R^3$ with $|\bx | \leq \beta Z^{-\frac{1+\mu}{3}}$ we have
\begin{equation*}
|\Phi_{|\bx|}^{\rm TF}(\bx)-\Phi^{\rm HF}_{|\bx|}(\bx) | \leq  C \beta^{\frac{4}{1+\mu}}(1+ \beta^{\frac{9}{22(1+\mu)}}|\bx|^{\frac{2+11\mu}{22(1+\mu)}}) |\bx|^{-4 +\frac{4\mu}{1+\mu}}.
\end{equation*}
\end{lemma}

\begin{proof}
By the definition of screened nuclear potential we have
\begin{equation*}
\left\vert \Phi _{\vert \bx\vert }^{\rm HF}( \bx) -\Phi _{\vert \bx\vert }^{\rm TF}( \bx) \right\vert \leq \int_{\vert \by\vert<\vert \bx\vert }\frac{\vert \rho^{\rm HF}( \by) -\rho^{\rm TF}( \by) \vert }{\vert \bx-\by\vert }d\by=\dots
\end{equation*}
and for all $k>0$ by Proposition~\ref{Counorm}
\begin{eqnarray}\label{f4}
\dots &\leq & 2^{\frac{3}{2}}k^{-1} \vert \bx\vert ^{-\frac{1}{2}}\left\Vert \rho ^{\rm HF}-\rho ^{\rm TF}\right\Vert _{C}+\int_{A\left(\vert \bx\vert ,k\right) }\frac{\rho ^{\rm HF}(\by)+\rho ^{\rm TF}(\by)}{\vert \bx-\by\vert}d\by.
\end{eqnarray}

Since $\Vert \rho ^{\rm TF}\Vert _{L^{\frac{5}{3}}( \R^{3}) }\leq C Z^{\frac{7}{5}}$ (Corollary~\ref{rTF53}) and
\begin{equation}\label{52ann}
\int_{A(\vert \bx\vert ,k) }\frac{1}{\vert \bx-\by\vert ^{\frac52}}d\by\leq8\pi\vert \bx \vert ^{\frac12} (2k) ^{\frac{1}{2}}.
\end{equation}
(see \cite{Sol} page 549) one finds
\begin{eqnarray}\label{f5}
\int_{A(\vert \bx \vert ,k) }\frac{\rho^{\rm TF}(\by) }{\vert \bx-\by \vert }d\by &\leq & C Z^{\frac{7}{5}}\vert \bx\vert ^{\frac15}k^{\frac15}.
\end{eqnarray}

The term with the HF-density has to be treated differently since we do not have a bound for the $L^{\frac{5}{3}}$-norm of $\rho ^{\rm HF}$. For a $R\in \mathbb{R}^{+}$ to be chosen later we consider the splitting
\begin{equation}\label{f2}
\int_{A( \vert \bx\vert ,k) }\frac{\rho^{\rm HF}( \by) }{\vert \bx-\by\vert }d\by=\int_{\substack{ A(\vert \bx\vert,k)  \\ \vert \bx-\by\vert >R}}\frac{\rho ^{\rm HF}( \by) }{\vert \bx-\by\vert }d\by +\int _{\substack{ A(\vert \bx\vert ,k)  \\
\vert\bx-\by\vert < R}}\frac{\rho^{\rm HF}(\by) }{\vert \bx-\by\vert }d\by.
\end{equation}
We consider these two terms separately. Let $\Sigma$ be defined as in \eqref{Sigma}; i.e. the region where $G_{\alpha}(\rho^{\rm HF})$ behaves like $(\rho^{\rm HF})^{\frac43}$ (Remark~\ref{Daure}). By H\"older's inequality we find
\begin{eqnarray*}
\int_{\substack{ A(\vert \bx\vert,k)  \\ \vert \bx-\by\vert >R}}\frac{\rho ^{\rm HF}( \by) }{\vert \bx-\by\vert }d\by & \leq & \Big( \int_{\substack{ A(\vert \bx\vert,k)  \\ \vert \bx-\by\vert >R}} \frac{1}{|\bx -\by|^4} d\by \Big)^{\frac14} \Big(\int_{\by \in \Sigma} \big(\rho ^{\rm HF}( \by))^{\frac43}d\by \Big)^{\frac34}\\
&& +  \Big( \int_{\substack{ A(\vert \bx\vert,k)}}\frac{1}{\vert \bx-\by\vert^{\frac52} }d\by\Big)^{\frac25} \Big( \int_{\substack{\by \in \R^3 \setminus \Sigma}} \big(\rho ^{\rm HF}( \by)\big)^{\frac53}d\by\Big)^{\frac35}.
\end{eqnarray*}
{F}rom the inequality above, Remark~\ref{Daure} and estimate \eqref{G1} we get
\begin{eqnarray}\label{1a}
\int_{\substack{ A(\vert \bx\vert,k)  \\ \vert \bx-\by\vert >R}}\frac{\rho ^{\rm HF}( \by) }{\vert \bx-\by\vert }d\by & \leq & C R^{-\frac38}|\bx|^{\frac18} k^{\frac18} Z +  C |\bx|^{\frac15} k^{\frac15}Z^{\frac75}.
\end{eqnarray}
On the other hand for the second term on the right hand side of \eqref{f2} by \eqref{inema} and Lemma~\ref{rrcou} we find
\begin{eqnarray}\label{2a}
\int_{\vert \bx-\by\vert <R}\frac{\rho^{\rm HF}(\by)}{\vert \bx-\by\vert}d\by \leq  C (Z^{\frac43}+ R^{\frac12} Z^{\frac32}).
\end{eqnarray}
Hence from \eqref{f4}, Lemma~\ref{rrcou}, \eqref{f5}, \eqref{1a} and \eqref{2a}, we get
\begin{equation}\label{s1}
\vert \Phi _{|\bx|}^{\rm HF}( \bx) -\Phi _{|\bx|}^{\rm TF}( \bx) \vert \leq C\big(\frac{Z^{1+\frac{3}{22}}}{|\bx|^{1/2} k}+ Z^{\frac75} |\bx|^{\frac15} k^{\frac15}+ R^{-\frac38}|\bx|^{\frac18} k^{\frac18} Z+  R^{\frac12} Z^{\frac32}+Z^{\frac43} \big) .
\end{equation}
Choosing $k$ such that $Z^{\frac43}= Z^{\frac75} |\bx|^{\frac15} k^{\frac15}$, i.e. $k=|\bx|^{-1}Z^{-\frac13}$ and $R$ such that $R^{-\frac38}Z^{1-\frac{1}{24}}=Z^{\frac43}$, i.e. $R=Z^{-1}$ we find
\begin{equation*}
\vert \Phi _{|\bx|}^{\rm HF}( \bx) -\Phi _{|\bx|}^{\rm TF}( \bx) \vert \leq C(|\bx|^{\frac12} Z^{\frac43+\frac{3}{22}}+Z^{\frac43} ).
\end{equation*}
The claim follows using that $|\bx|\leq \beta Z^{-\frac{1+\mu}{3}}$ .
\end{proof}

\begin{theorem}\label{Lmnu}
Let $Z \alpha = \kappa$ be fixed with $0 \leq \kappa < 2/\pi$ and $Z \geq 1$. Then there exists an $\alpha_0>0$ such that for all $\alpha \leq \alpha_0$ and $\bx \in \R^3$ with $|\bx | \leq \beta Z^{-\frac{1}{3}}$ we have
\begin{equation}\label{pa1}
|\Phi _{|\bx|}^{\rm HF}( \bx) -\Phi _{|\bx|}^{\rm TF}( \bx)|\leq C \beta^{2-\frac{1}{66}}(1+\beta^2+\beta^{\frac52}+ \beta^{2+\frac{789}{1936}}|\bx|^{\frac{179}{1936}}) |\bx|^{-4+\frac{1}{66}}.
\end{equation}
Moreover if $|\bx | \leq \beta Z^{-\frac{1-\mu}{3}}$ for $\mu < \frac{2}{11}\frac{1}{49}$, then
\begin{equation}\label{pa0}
|\Phi_{|\bx|}^{\rm TF}(\bx)-\Phi^{\rm HF}_{|\bx|}(\bx) | \leq   C \beta^{2-a(\mu)}(1+\beta^2+ \beta^{\frac52}+\beta^{b(\mu)}|\bx|^{c(\mu)}) |\bx|^{-4 +a(\mu)},
\end{equation}
with $a(\mu)=\frac{1}{66(1-\mu)}-\frac{49 \mu}{12(1-\mu)}$, $b(\mu)=2+\frac{3}{176}\frac{24-24\mu-\frac{1}{11}+\frac{49}{2}\mu}{1-\mu}$ and $c(\mu)=\frac{1}{11}-\frac{\frac{3}{11}-\frac32 49 \mu}{22(8-8\mu)}$ strictly positive constants.
\end{theorem}

\begin{proof}
Proceeding as in the proof of Lemma~\ref{smallx} up to \eqref{1a} we get
\begin{eqnarray}\notag
\vert \Phi _{|\bx|}^{\rm HF}( \bx) -\Phi _{|\bx|}^{\rm TF}( \bx) \vert &\leq& C(k^{-1}|\bx|^{-\frac12} Z^{1+\frac{3}{22}}+ Z^{\frac75} |\bx|^{\frac15} k^{\frac15}+ R^{-\frac38}|\bx|^{\frac18} k^{\frac18} Z) \\  \label{a3}
& +&  \int_{|\bx -\by|\leq R}\frac{\rho ^{\rm HF}(\by)}{\vert \bx-\by\vert}d\by,
\end{eqnarray}
for $R\in \mathbb{R}^{+}$ to be chosen. It remains to estimate the last term on the right hand side of \eqref{a3}. For \lq small\rq \ $R$ which
is relevant for small $\bx$ we already did it in Lemma~\ref{smallx}, for \lq big\rq \ $R$ which is relevant for big $\bx$ we use
Proposition~\ref{prthomas} in Appendix~\ref{ZZ}.

Take $\gamma \leq 1/263$ to be chosen. If $|\bx| \leq \beta Z^{-\frac{1+\gamma}{3}}$ then by Lemma~\ref{smallx}
\begin{equation}\label{pa2}
|\Phi_{|\bx|}^{\rm TF}(\bx)-\Phi^{\rm HF}_{|\bx|}(\bx) | \leq   C \beta^{\frac{4}{1+\gamma}}(1+ \beta^{\frac{9}{22(1+\gamma)}}|\bx|^{\frac{2+11\gamma}{22(1+\gamma)}}) |\bx|^{-4 +\frac{4\gamma}{1+\gamma}}.
\end{equation}
If instead  $|\bx| > \beta Z^{-\frac{1+\gamma}{3}}$, let $H_{\bx}$ be the Hamiltonian defined in \eqref{HamiltonianP} with $\bP=\bx$ and $\nu=Z$. Then by the definition of $H_{\bx}$ and taking the HF-minimizer as a trial wave function we have
\begin{eqnarray*}
\inf_{\substack{\psi \in \wedge_{i=1}^N L^2(\R^3)\\ \|\psi\|_2=1}} \langle \psi, H_{\bx} \psi \rangle & \leq & \mathcal{E}^{\rm HF}(\gamma^{\rm HF}) - Z  \int_{|\bx -\by|< R}\frac{\rho ^{\rm HF}(\by)}{\vert \bx-\by\vert}d\by\\
& = & \inf_{\gamma \in \mathcal{A}} \mathcal{E}^{\rm HF}(\gamma) - Z  \int_{|\bx -\by|< R}\frac{\rho ^{\rm HF}(\by)}{\vert \bx-\by\vert}d\by =\dots.
\end{eqnarray*}
Since $\tfrac12 |\bp|^2 \geq \alpha^{-1}T(\bp)$, $ \inf_{\gamma \in \mathcal{A}} \mathcal{E}^{\rm HF}(\gamma)$ is estimated from above by the HF-ground state energy of the non-relativistic model (i.e. when the kinetic energy is given by $-\tfrac12 \Delta$). Moreover, this last one can be estimated from above by $\mathcal{E}^{\rm TF}(\rho^{\rm TF})+CN^{\frac15}Z^2$ (see \cite{LSTF} and \cite{Ltf}). Hence we find
\begin{equation*}
\dots \leq \mathcal{E}^{\rm TF}(\rho^{\rm TF}) +C N^{\frac15} Z^2 - Z  \int_{|\bx -\by|\leq R}\frac{\rho ^{\rm HF}(\by)}{\vert \bx-\by\vert}d\by.
\end{equation*}
On the other hand since $|\bx|> \beta Z^{-\frac{1+\gamma}{3}}$
choosing for some $l>\frac{1+\gamma}{3}$, $R<\beta Z^{-l} /4$ from Proposition~\ref{prthomas}
it follows that there exists a constant depending only on $\kappa$ such that for $t \in ((1+\gamma)/3, \min\{l,3/5\})$, and
for every $\psi \in \wedge_{i=1}^N L^2(\R^3)$ with $\| \psi \|_2 =1$ we have
\begin{eqnarray*}
\langle \psi, H_{\bx} \psi \rangle & \geq  &  \mathcal{E}^{\rm TF}(\rho^{\rm TF}) -C (\beta^{1/2} +\beta^{-2}) Z^{\frac52-\frac{t}{2}} ,
\end{eqnarray*}
Hence combining the two inequalities above we find
\begin{equation}\label{R2}
 \int_{|\bx -\by|\leq R}\frac{\rho ^{\rm HF}(\by)}{\vert \bx-\by\vert}d\by \leq C  (\beta^{1/2} +\beta^{-2}) Z^{\frac12(3-t)}.
\end{equation}
{F}rom \eqref{a3} and the inequality above we get
\begin{eqnarray*}
\vert \Phi _{|\bx|}^{\rm HF}( \bx) -\Phi _{|\bx|}^{\rm TF}( \bx) \vert & \leq &  Ck^{-1}|\bx|^{-\frac12} Z^{1+\frac{3}{22}}+ C Z^{\frac75} |\bx|^{\frac15} k^{\frac15}\\
&& +C R^{-\frac38}|\bx|^{\frac18} k^{\frac18} Z +C  (\beta^{1/2} +\beta^{-2}) Z^{\frac12(3-t)} .
\end{eqnarray*}
Choosing $k$ such that $Z^{\frac12(3-t)}= Z^{\frac75} |\bx|^{\frac15} k^{\frac15}$, i.e $k= |\bx|^{-1} Z^{\frac12(1-5t)}$ and  $R$ such that $Z^{\frac12(3-t)} \sim R^{-\frac38} Z^{1+\frac{1}{16}(1-5t)} $, i.e $R=\beta Z^{-\frac76+\frac12t}/4$ we find
\begin{equation}\label{sudata}
\vert \Phi _{|\bx|}^{\rm HF}( \bx) -\Phi _{|\bx|}^{\rm TF}( \bx) \vert \leq C(|\bx|^{\frac12} Z^{\frac{7}{11}+\frac{5}{2}t}+(\beta^{1/2} +\beta^{-2})Z^{\frac12(3-t)} ).
\end{equation}
Notice that $R< \beta Z^{-l} /4$ is satisfied choosing $l=4 t/3$. Then for $\bx$ such that $\beta Z^{-\frac{1+\gamma}{3}} \leq|\bx| \leq \beta Z^{-\frac13}$ we find
\begin{equation*}
|\Phi _{|\bx|}^{\rm HF}( \bx) -\Phi _{|\bx|}^{\rm TF}( \bx)|\leq C (|\bx|^{-\frac{31}{22}-\frac{15}{2}t} \beta^{\frac{21}{11}+\frac{15}{2}t}+(\beta^{1/2} +\beta^{-2})\beta^{\frac32(3-t)}|\bx|^{-\frac32(3-t)} ) .
\end{equation*}
Optimizing in $t$ gives $t=1/3+1/99$. For this value of $t$ we get
\begin{equation}\label{pa3}
|\Phi _{|\bx|}^{\rm HF}( \bx) -\Phi _{|\bx|}^{\rm TF}( \bx)|\leq C (1+\beta^{\frac52}) \beta^{2-\frac{1}{66}}|\bx|^{-4+\frac{1}{66}}.
\end{equation}
Inequality \eqref{pa1} follows from \eqref{pa2} and \eqref{pa3} choosing $\gamma$ such that $4\gamma/(1+\gamma)=1/66$, i.e. $\gamma=1/263$.

On the other hand from \eqref{sudata} for $\bx$ such that $\beta Z^{-\frac{1+\gamma}{3}} \leq|\bx| \leq \beta Z^{-\frac{1-\mu}{3}}$ we find
\begin{eqnarray*}
|\Phi _{|\bx|}^{\rm HF}( \bx) -\Phi _{|\bx|}^{\rm TF}( \bx)| & \leq &  C |\bx|^{\frac12-\frac{3}{1-\mu}(\frac{7}{11}+\frac52t)} \beta^{\frac{3}{1-\mu}(\frac{7}{11}+\frac52t)}\\
&&+C(\beta^{1/2} +\beta^{-2}) \beta^{\frac{3}{2(1-\mu)}(3-t)}|\bx|^{-\frac{3}{2(1-\mu)}(3-t)}  .
\end{eqnarray*}
Optimizing in $t$ gives $t=1/3+1/99-\frac{1}{18}\mu$. For this value of $t$ we get
\begin{equation*}
|\Phi _{|\bx|}^{\rm HF}( \bx) -\Phi _{|\bx|}^{\rm TF}( \bx)|\leq C (1+\beta^{\frac52}) \beta^{2-\frac{1}{66(1-\mu)}+\frac{49 \mu}{12(1-\mu)}} |\bx|^{-4+\frac{1}{66(1-\mu)}-\frac{49 \mu}{12(1-\mu)}} .
\end{equation*}
Inequality \eqref{pa0} follows from the one above and \eqref{pa2} choosing $\gamma$ such that $4\gamma/(1+\gamma)=\frac{1}{66(1-\mu)}-\frac{49 \mu}{12(1-\mu)}$.
\end{proof}


\section{The exterior part}

In this section we complete the proof of Theorem~\ref{mainesti}. We first estimate the exterior integral of the density and study the minimization problem that the exterior part of the minimizer satisfies. Then we prove the main estimate in Theorem~\ref{mainesti} in an intermediate zone, i.e. far from the nucleus but not further than a fixed distance independent of $Z$. To study this area we need first to construct a TF-model that gives a good approximation of the HF-density in this intermediate zone. By the estimate on the exterior integral of the density we can then also prove Theorem~\ref{mainesti} in the region far away from the nucleus.

\subsection{The exterior integral of the density}

The main result of this section is the following lemma.

\begin{lemma}[The exterior integral of the density]\label{extde} Assume that for some $R, \sigma, \varepsilon'>0$
\begin{equation}\label{19}
| \Phi^{\rm HF}_{|\bx|}(\bx) -\Phi^{\rm TF}_{|\bx|}(\bx)| \leq \sigma |\bx|^{-4+\varepsilon '} \, ,
\end{equation}
holds for $|\bx| \leq R$. Then for $0<r\leq R$
\begin{equation}\label{rr1}
\Big|\int_{|\bx|<r} (\rho^{\rm HF}(\bx)-\rho^{\rm TF}(\bx)) \; d\bx \Big| \leq \sigma r^{-3+\varepsilon '}
\end{equation}
and
\begin{equation}\label{rr1bis}
\int_{|\bx|>r} \rho^{\rm HF}(\bx) d\bx \leq C (1+\sigma r^{\varepsilon '})(1+r^{-3}),
\end{equation}
with $C$ a universal constant.
\end{lemma}

We proceed similarly as in the proof of Lemma 10.5 in \cite{Sol}. Since we need to localize we first present some technical lemmas that will take care of the error terms due to the localization. The localization error that will appear in the argument below (see \eqref{j10}) will be in the form of an
operator $L$ similar to the error \eqref{Lalpha} in the IMS formula. We estimate
this error in Lemma \ref{ErrorA}.

\begin{remark}\label{K2est}
Let $0\leq \beta _{1}<..<\beta _{4}$ be real numbers with possibly $\beta _{4}=\infty $. Let us denote $\Sigma_r(\beta_i, \beta_j)=\{ \bx \in \R^3: \beta_i r \leq |\bx| \leq \beta_j r\}$. Then we have
\begin{equation*}
\iint_{\begin{array}{l}\bx \in \Sigma_r(\beta_1,\beta_2)\\
\by \in \Sigma_r(\beta_3,\beta_4) \end{array}} K_{2}( \alpha ^{-1}\vert \bx-\by\vert)^{2} \, d\bx d\by \leq \frac{\left( 4^3\pi \right) ^{2}}{3}\frac{\beta _{2}^{3}-\beta _{1}^{3}}{\beta _{3}-\beta _{2}} \; \alpha ^{4} r^{2}  e^{-\alpha ^{-1}r( \beta _{3}-\beta _{2})} \, .
\end{equation*}
The proof of this estimate is given in Appendix~\ref{tech}.
\end{remark}

\begin{lemma}\label{ErrorA}
Let $r>0$ and $\lambda ,\nu \in ( 0,1) $.  Let $\chi_{-}$ be the characteristic function of $B_{r(1-\nu)}(0)$ and $\chi _{0}$ be the characteristic function of the sector $\{ \bx \in \R^3: r(1-\nu)<|\bx| <r (1+\nu)/(1-\lambda)\}$. Let $\eta $ be a Lipschitz function such that $0 \leq \eta(\bx) \leq 1$ for all $\bx \in \R^3$,  $\eta (\bx)\equiv 0$ if $\vert \bx\vert \leq r$, $\eta (\bx) \equiv 1$ if $\vert \bx\vert \geq r(1-\lambda)^{-1}$ and $\Vert \nabla \eta\Vert _{\infty }$ is bounded. Let $L$ denote the operator with integral kernel
\begin{equation}\label{LL}
L(\bx,\by)=\frac{\alpha^{-2}}{4\pi ^{2}}\frac{( \eta(\bx) -\eta(\by))(\eta( \bx) \vert \bx\vert -\eta(\by) \vert \by\vert) }{\vert\bx-\by\vert^{2}}K_{2}( \alpha ^{-1}\vert \bx-\by\vert) .
\end{equation}

Then for every function $f\in L^{2}( \R^{3}) $ we have
\begin{equation*}
\alpha^{-1} |( f,Lf )| \leq  3 \, D(\eta, \lambda, r) \, \| \chi_{0}f\|_2^2 + D(\eta, \lambda, r) e^{-\tfrac12\alpha ^{-1}r\nu}  \| \chi_{-}f\|_2^2 + \alpha^{-1}|( f, Q f )| ,
\end{equation*}
with $D(\eta, \lambda, r):= \Vert \nabla \eta \Vert _{\infty }\Big( \frac{\Vert \nabla \eta \Vert _{\infty }r}{1-\lambda }+1\Big) $ and $Q$ a positive semi-definite operator such that
\begin{equation*}
\Tr[Q] \leq  C D(\eta, \lambda, r) \alpha^{-1} r^2 e^{-\tfrac12 \alpha^{-1}r \nu},
\end{equation*}
with $C$ depending only on $\lambda$ and $\nu$.
\end{lemma}

\begin{proof}
As a first step we decompose the operator $L$. We introduce a third cut-off function $\chi _{+}$ such that $1= \chi_{-}(\bx)+\chi_{0}(\bx)+\chi_{+}(\bx)$ for all $\bx \in \R^3$. We decompose the operator $L$ with respect to these characteristic functions
as follows:
\begin{equation*}
L =\chi _{-}L (\chi _{0}+ \chi_{+})+(\chi _{0}+\chi_{+})L\chi _{-}+\chi _{0}L\chi _{+}+\chi _{+}L\chi _{0}+\chi _{0}L\chi _{0}.
\end{equation*}
We proceed similarly as in \cite[Proof of Theorem 2.6 (Localization error)]{SSS}. For $\Gamma_1, \Gamma_2$ bounded operators from $(\Gamma_1 -\Gamma_2)(\Gamma_1 -\Gamma_2)^* \geq 0$ it follows
\begin{equation}\label{G12}
\Gamma_1 \Gamma_2^* +\Gamma_2  \Gamma_1^* \leq \Gamma_1 \Gamma_1^* + \Gamma_2 \Gamma_2^* .
\end{equation}
We are going to use several times this inequality with different choices of $\Gamma_1$ and $\Gamma_2$.

As a first choice we consider $\Gamma_1= \sqrt{\varepsilon_1} \chi_{-}$ and $\Gamma_2 = 1/\sqrt{\varepsilon_1} (\chi_{0}+\chi_{+}) L \chi_{-}$ with $\varepsilon_1>0$ to be chosen. Using \eqref{G12} we get
\begin{equation}\label{Ober1}
|(f, (\chi _{-}L(\chi _{0}+\chi_{+})+(\chi _{0}+\chi_{+})L\chi _{-}) f) | \leq \varepsilon_1 \| \chi_{-} f\|_{2}^2 + \frac{1}{\varepsilon_1} ( f, Q_{1}f ) ,
\end{equation}
with $Q_{1}= (\chi_{0}+\chi_{+})L \chi_{-}^2 L (\chi_{0}+\chi_{+})$. We estimate now the trace of $Q_{1}$. By the definition of $\eta, \chi_{-},\chi_{0}$ and $\chi_{+}$ it follows that
\begin{equation*}
\Tr[Q_1] = \int_{|\bx| \leq r(1-\nu)} \int_{|\by| \geq r} L^2(\bx, \by) \, d\bx d\by \leq \tfrac{(16)^2}{3\pi^{2}}  \tfrac{(1-\nu)^3}{\nu} D(\eta, \lambda,r)^2 r^2 e^{-\alpha^{-1}r \nu}.
\end{equation*}
In the last step we use the definition of $L$, Remark~\ref{K2est} and the definition of the constant $D(\eta, \lambda,r)$ given in the statement of the lemma.

Now we choose  $\Gamma_1= \sqrt{\varepsilon_2} \chi_0$ and $\Gamma_2 = 1/\sqrt{\varepsilon_2} \chi_{+} L \chi_{0}$ with $\varepsilon_2>0$ to be chosen. Proceeding as above we get
\begin{equation}\label{Ober2}
| (f, (\chi _{+}L\chi _{0}+\chi _{0}L\chi _{+}) f) | \leq \varepsilon_2 \| \chi_0 f\|_{2}^2 + \frac{1}{\varepsilon_2} (f, Q_{2} f ) ,
\end{equation}
with $Q_{2}= \chi_{+} L \chi^2_{0} L \chi_{+} $ and such that
\begin{equation*}
\Tr[Q_{2}] \leq \tfrac{(16)^2}{3 \pi ^{2}} \tfrac{1-(1-\nu)^3(1-\lambda)^3}{\nu(1-\lambda)^2}  D(\eta, \lambda,r)^2 \; r^2 e^{-\alpha^{-1}r \frac{\nu}{1-\lambda}}.
\end{equation*}

It remains to study the term $\chi _{0}L\chi _{0}$. This one has to be treated differently. By Schwartz's inequality one gets
\begin{equation}\label{Ober3}
|( f,\chi _{0}L\chi _{0}f)| \leq \tfrac{3 \alpha}{2}   D(\eta, \lambda,r) \int_{\R^{3}}\chi _{0}(\bx) \vert f(\bx) \vert^{2} ,
\end{equation}
since $\int_{\R^3} |L(\bx,\by)| d\bx d\by \leq \tfrac{3 \alpha}{2}  D(\eta, \lambda,r)$.

The claim follows from \eqref{Ober1}, \eqref{Ober2} and \eqref{Ober3} choosing $\varepsilon _{1} =   D(\eta, \lambda,r) \alpha e^{-\tfrac12 \alpha ^{-1}r\nu }$, $\varepsilon _{2} =\tfrac{3\alpha}{2}  D(\eta, \lambda,r)$ and with $Q:= \frac{1}{\varepsilon_1} Q_{1}+\frac{1}{\varepsilon_2}Q_{2}$.
\end{proof}

\begin{definition}[The localization function]\label{defloc}
Fix $0<\lambda <1$ and let $G :\R^{3}\rightarrow  \R$ be given by
\begin{equation*}
G( \bx ) :=\left\{
\begin{array}{ll}
0 & \text{if }\vert \bx\vert \leq 1, \\
\frac{\pi }{2}(\vert \bx\vert -1) \frac{1}{( 1-\lambda)^{-1}-1} & \text{if }1\leq \vert \bx\vert \leq ( 1-\lambda)^{-1}, \\
\frac{\pi }{2} & \text{if }( 1-\lambda )^{-1}\leq \vert \bx\vert .
\end{array}
\right.
\end{equation*}
Let $r>0$ and define the outside localization function $\theta_{r}(\bx) :=\sin(G( \frac{\vert \bx\vert }{r})).$
\end{definition}

\begin{remark}
{F}rom the definition it follows that $\Vert \nabla \theta _{r}\Vert _{\infty }\leq \frac{\pi}{2}\frac{1-\lambda }{\lambda }r^{-1}$.
\end{remark}

\begin{lemma}\label{extterm}
For all $r>0$ and $\lambda ,\nu \in( 0,1) $ the density $\rho ^{\rm HF}$ of the minimizer satisfies
\begin{eqnarray*}
\int_{\vert \bx\vert >r( 1-\lambda)^{-1}}\rho^{\rm HF}( \bx) d\bx & \leq & 1 +\tfrac{2}{\lambda}+ 2 \sup_{|\bx|=r(1-\lambda)}|\bx|\Phi_{r(1-\lambda)}^{\rm HF}(\bx) + \mathcal{R}^{\frac12}
\end{eqnarray*}
with
\begin{equation*}
\mathcal{R} =  6 D(\lambda) r^{-1} \int_{r(1-\nu)<|\bx| < r\frac{1+\nu}{1-\lambda}} \rho^{\rm HF}(\bx) \; d\bx +2 D(\lambda)( r^{-1}N + C r \alpha^{-2}) e^{-\tfrac12 \alpha^{-1}r \nu},
\end{equation*}
with $D(\lambda):=(1+ \pi /(2 \lambda (1-\lambda))) \pi/ (2 \lambda)$ and $C=C(\lambda, \nu)$.
\end{lemma}

\begin{proof}
Let $\gamma ^{\rm HF}$ be the minimizer. By the variational principle, $\gamma^{\rm HF}$ is a projection onto the subspace spanned by $u_{1},\dots ,u_{N}$. These functions $u_{i}$ satisfy the Euler Lagrange equations $h_{\gamma^{\rm HF}} u_{i}=\varepsilon _{i}u_{i}$, $\varepsilon _{i}<  0,$ for $i=1,\dots ,N$, with $h_{\gamma^{\rm HF}}$ defined in \eqref{euler}.

Given $\eta $ a function in $C^{1}( \R^{3}) $ with support away from zero, we find
\begin{equation*}
0 \geq \sum_{i=1}^{N}\varepsilon _{i}\int_{\R^{3}}\vert u_{i}( \bx)\vert^{2}\vert \bx\vert \eta ^{2}( \bx) d\bx = \sum_{i=1}^{N}\int_{\R^{3}}u_{i}(\bx) ^{\ast}\vert \bx\vert \eta^{2}( \bx) \, h_{\gamma^{\rm HF}}u_{i}( \bx) d\bx .
\end{equation*}
Since $ \eta T(\bp)u_i \in L^{2}(\R^3)$ (Theorem~\ref{HF}, (3)), using the Euler-Lagrange equations and treating all the terms, except the kinetic energy, as in \cite[Formula (63)]{Sol} we get
\begin{eqnarray}
0 &\geq &\alpha ^{-1}\sum_{i=1}^{N} (u_i \eta |\cdot|, \eta T(\bp) u_i ) - Z \int_{\R^{3}}\rho^{\rm HF}(\bx) \eta^{2}(\bx) d\bx \notag \\
&&+\int_{\R^{3}}\int_{\R^{3}}\left[ \rho ^{\rm HF}( \bx) \rho ^{\rm HF}(\by) -\Tr_{\C^{q}}\vert \gamma ^{\rm HF}( \bx,\by )\vert ^{2} \right] \frac{\vert \by\vert( 1-\eta^{2}(\bx)) \eta^{2}( \by) }{\vert\bx-\by\vert} \, d\bx d\by \notag \\
&&+\tfrac{1}{2}\Big( \int_{\R^{3}}\rho ^{\rm HF}(\bx) \eta^{2}(\bx) d\bx\Big) ^{2}-\tfrac{1}{2}\int_{\R^{3}}\rho^{\rm HF}(\bx) \eta^{2}(\bx) d\bx .  \label{for2}
\end{eqnarray}

Now we look at the kinetic energy term. For each $i \in \{1, \dots ,N \}$ we may write
\begin{equation}
\Real (u_i \eta |\cdot|, \eta T(\bp) u_i ) = \Real (u_i \eta |\cdot|, T(\bp) (\eta u_i) ) + \Real (u_i \eta |\cdot|,[ \eta,  T(\bp)] u_i ) , \label{comm1}
\end{equation}
where $[A,B]$ denotes the commutator of the operators $A$ and $B$. The first term on the right hand side of \eqref{comm1} is non-negative by the result of Lieb in \cite{L}. Notice that here we may use that $\eta u_i \in H^{1}(\R^3)$ (see Theorem~\ref{HF}, (3)).

Hence, from  \eqref{for2} and \eqref{comm1} we find
\begin{eqnarray}
0 &\geq &\alpha ^{-1}\sum_{i=1}^{N} \Real (u_i \eta |\cdot|, [\eta ,T(\bp)] u_i ) - Z \int_{\R^{3}}\rho^{\rm HF}(\bx) \eta^{2}(\bx) d\bx \notag \\
&&+\int_{\R^{3}}\int_{\R^{3}}\left[ \rho ^{\rm HF}( \bx) \rho ^{\rm HF}(\by) -\Tr_{\C^{q}}\vert \gamma ^{\rm HF}( \bx,\by )\vert ^{2} \right] \frac{\vert \by\vert( 1-\eta^{2}(\bx)) \eta^{2}( \by) }{\vert\bx-\by\vert} \, d\bx d\by \notag \\
&&+\tfrac{1}{2}\Big( \int_{\R^{3}}\rho ^{\rm HF}(\bx) \eta^{2}(\bx) d\bx\Big) ^{2}-\tfrac{1}{2}\int_{\R^{3}}\rho^{\rm HF}(\bx) \eta^{2}(\bx) d\bx .  \label{for3}
\end{eqnarray}

By a density argument we may choose $\eta =\theta _{r}$ the localization function defined in Definition~\ref{defloc}. Reasoning as on page 541 of \cite{Sol},
we get
\begin{eqnarray}
0 &\geq &\alpha ^{-1}\sum_{i=1}^{N} \Real (u_i \eta |\cdot|, [\eta ,T(\bp)] u_i ) \notag +\tfrac{1}{2}\Big( \int_{\R^{3}}\rho ^{\rm HF}(\bx) \eta^{2}(\bx) d\bx\Big) ^{2}\\
&& -\Big( \tfrac{1}{2} +\tfrac{1}{\lambda}+ \sup_{|\bx|=r(1-\lambda)}|\bx|\Phi_{r(1-\lambda)}^{\rm HF}(\bx) \Big)\int_{\R^{3}}\rho^{\rm HF}(\bx) \eta^{2}(\bx) d\bx .  \label{for4}
\end{eqnarray}

It remains to estimate the first term on the right hand side of \eqref{for4}. With the same arguments used in the proof of the IMS formula, it can be rewritten as
\begin{equation} \label{j10}
\alpha^{-1} \sum_{i=1}^{N} \Real (u_i \eta |\cdot|, [\eta ,T(\bp)] u_i ) = - \alpha^{-1} \sum_{i=1}^{N} (u_i, L u_i ),
\end{equation}
where $L$ is the operator defined in \eqref{LL}. Using Lemma~\ref{ErrorA} and since $\Vert \nabla \eta \Vert_{\infty }=\Vert \nabla \theta_{r} \Vert_{\infty }\leq \pi /\left( 2\lambda r\right) $ we find, with $D(\lambda)$ defined as in the statement,
\begin{eqnarray}
\alpha^{-1} \Big|\sum_{i=1}^{N} (u_i, L u_i ) \Big| & \leq & 3 D(\lambda) r^{-1} \|\chi_0 \rho^{\rm HF} \|_{1} + D(\lambda) r^{-1} e^{-\tfrac12 \alpha^{-1}r \nu}  \| \chi_{-}\rho^{\rm HF} \|_{1} \notag \\
&& + C D(\lambda) r \alpha^{-2}e^{-\tfrac12 \alpha^{-1}r \nu} ,
\label{sopra}
\end{eqnarray}
where $\chi_{0}, \chi_{-}$ and $C$ are as defined in the statement of Lemma~\ref{ErrorA}. Hence combining \eqref{for4} with \eqref{sopra}, using the definition of $\chi_{0}$ and that $\| \chi_{-} \rho^{\rm HF} \|_{1} \leq N$ we have
\begin{eqnarray*}
0 &\geq & - 3 D(\lambda) r^{-1} \int_{r(1-\nu)<|\bx| < r\frac{1+\nu}{1-\lambda}} \rho^{\rm HF}(\bx) \; d\bx - D(\lambda) r^{-1} e^{-\tfrac12 \alpha^{-1}r \nu}  N  \\
&& - C D(\lambda) r \alpha^{-2}e^{-\tfrac12 \alpha^{-1}r \nu} +\tfrac{1}{2}\Big(\int_{\R^{3}}\rho ^{\rm HF}(\bx) \eta^{2}(\bx) d\bx\Big) ^{2}\\
&& -\Big( \tfrac{1}{2} +\tfrac{1}{\lambda}+ \sup_{|\bx|=r(1-\lambda)}|\bx|\Phi_{r(1-\lambda)}^{\rm HF}(\bx) \Big)\int_{\R^{3}}\rho^{\rm HF}(\bx) \eta^{2}(\bx) d\bx .
\end{eqnarray*}
The claim follows using  that $x^{2}-Bx-C\leq 0$ implies $x\leq B+\sqrt{C}$.
\end{proof}

\begin{proof}[Proof of Lemma~\ref{extde}] We proceed as in \cite[page 551]{Sol}. The first estimate follows directly from the equality
\begin{equation*}
\int_{|\bx|<r} (\rho^{\rm HF}(\bx)-\rho^{\rm TF}(\bx)) \; d\bx =\tfrac{1}{4 \pi}r \int_{S^2} (\Phi^{\rm HF}_{r}(r\omega) -\Phi^{\rm TF}_{r}(r\omega) ) \; d\omega,
\end{equation*}
and \eqref{19}. To prove \eqref{rr1bis} we use Lemma~\ref{extterm}. We first notice that for $0<\beta<\gamma$ and $\gamma$ such that $r \gamma \leq R$
\begin{eqnarray} \notag
\int_{r\beta < |\by|< r\gamma} \rho^{\rm HF}(\by) \; d\by & \leq &  \Big| \int_{|\by|< r\gamma} (\rho^{\rm HF}(\by) -\rho^{\rm TF}(\by))\; d\by\Big| \\
&&\hspace{-1cm} +\Big| \int_{|\by|< r\beta} ( \rho^{\rm HF}(\by) -\rho^{\rm TF}(\by)) \; d\by \Big| + \int_{|\by|> r\beta} \rho^{\rm TF}(\by) \; d\by  \notag \\
& \leq &  C r^{-3} \beta^{-3} (1+\sigma r^{\varepsilon'}). \label{stanca}
\end{eqnarray}
Here we used  \eqref{rr1} and that by the TF-equation and \eqref{Tfpou}
\begin{equation*}
\int_{|\by|> r\beta} \rho^{\rm TF}(\by) \; d\by \leq \tfrac{3^4 2 \pi^2}{q^2} \beta^{-3} r^{-3}.
\end{equation*}
Since $\int_{\vert \bx\vert >r}\rho^{\rm HF}\leq \int_{\vert \bx\vert > 2r/3}\rho^{\rm HF}$ to prove the claim we estimate this second integral. By Lemma~\ref{extterm} with $r$ replaced by $r/2$, $\lambda=\tfrac14$ and $\nu=\tfrac12$ we get
\begin{equation*}
\int_{\vert \bx\vert > 2r/3}\rho^{\rm HF}( \bx) d\bx \leq 9+ \tfrac{3}{4}r \sup_{|\bx|=3r/8}\Phi_{3r/8}^{\rm HF}(\bx) + \mathcal{R}^{\frac12},
\end{equation*}
with $\mathcal{R}$ defined as in the statement of Lemma~\ref{extterm}. By \eqref{19} and Corollary~\ref{TFscr} we find
\begin{equation*}
 \sup_{|\bx|=3r/8}\Phi_{3r/8}^{\rm HF}(\bx) \leq C \sigma r^{-4+\varepsilon'} + \sup_{|\bx|=3r/8}\Phi_{3r/8}^{\rm TF}(\bx) \leq C(1+\sigma r^{\varepsilon'}) r^{-4}.
\end{equation*}
Moreover, from \eqref{stanca} with $\beta=1/4$ and $\gamma=1$, since $N < 2Z+1$ and the boundness of $\mathbb{R}^{+} \ni x \mapsto x^p e^{-x}$ for all $p>0$, we find
\begin{equation*}
\mathcal{R} \leq C (r^{-4}(1+\sigma r^{\varepsilon'}) +r^{-1}).
\end{equation*}
The claim follows directly.
\end{proof}


\subsection{Separating the inside from the outside}

We consider the exterior part of the minimizer, i.e. the density matrix
\begin{equation}\label{exteriordensity}
\gamma^{\rm HF}_r:= \theta_r \gamma^{\rm HF} \theta_r,
\end{equation}
with $\theta_r$ as defined in Definition~\ref{defloc}. This density matrix almost minimizes a new energy functional where there is no exchange term. Indeed sufficiently far away from the nucleus the electrons are far apart and hence their mutual interaction is small.

We define an auxiliary energy functional on $\mathcal{A}$ (see \eqref{A}) given by
\begin{equation}\label{ea}
\mathcal{E}^{A}(\gamma ):=\Tr[ (\alpha^{-1} T(\bp) -\Phi_{r}^{\rm HF}) \gamma ]+ D(\rho_{\gamma}).
\end{equation}

\begin{theorem}\label{EA}
Let $r>0$ and $\lambda, \nu \in (0,1)$. Let $\chi_{r}^{+}$ denote the characteristic function of $\R^3 \setminus B_{r}(0)$. The density matrix $\gamma^{\rm HF}_r$ defined in \eqref{exteriordensity} satisfies
\begin{equation*}
\mathcal{E}^{A}(\gamma^{\rm HF}_r ) \leq \Big\{\mathcal{E}^{A}(\gamma): \gamma \in \mathcal{A}, \supp(\rho_{\gamma}) \subset \R^3 \setminus B_r(0), \|\rho_{\gamma}\|_1 \leq \|\rho^{\rm HF} \chi_r\|_1\Big\} + \mathcal{R},
\end{equation*}
where
\begin{eqnarray*}
\mathcal{R} & = &  (\tfrac{\pi}{2 \lambda} +\tfrac{C}{\lambda^2}r^{-1})r^{-1} \int_{r(1-\lambda)(1-\nu)\leq |\bx|} \rho^{\rm HF}(\bx) \; d\bx+c' \alpha^{-2}(1+\alpha r^{-2}) e^{-\tfrac12 \alpha^{-1}rd}\\
&& + \mathcal{E}x(\gamma^{\rm HF}_r) + C \int_{r(1-\lambda)\leq |\bx| \leq \frac{r}{1-\lambda}} \Big[\big( \Phi_{r(1-\lambda)}^{\rm HF}(\bx)\big)^{\frac52} +\alpha^3 \big( \Phi_{r(1-\lambda)}^{\rm HF}(\bx)\big)^{4}\Big]\; d\bx ,
\end{eqnarray*}
and $c', d$ are positive constants depending only on $\nu$ and $\lambda$.
\end{theorem}

\begin{proof}
We proceed as in \cite[pages 532-6]{Sol}. The first step of the proof is a localization. Once again we have to treat carefully the localization error coming from the kinetic energy. This is the main difference with \cite{Sol}. For completeness we repeat the main ideas of the reasoning.

We consider the following partition of unity of $\R^3$: $1=\theta_r^2(\bx)+\theta_{0}^2(\bx)+\theta_{-}^2(\bx) $ with $\theta_r$ defined as in Definition~\ref{defloc} and
\begin{equation*}
\theta_{0}(\bx):= \big(\theta_{r(1-\lambda)}^2(\bx) - \theta_r^2(\bx)\big)^{\frac12} \mbox{ and }\theta_{-}(\bx):= \big(1 - \theta_{r(1-\lambda)}^2(\bx)\big)^{\frac12}.
\end{equation*}
Associated to this partition of unity we define
\begin{equation*}
\gamma^{\rm HF}_{0}:= \theta_{0} \gamma^{\rm HF} \theta_{0} \mbox{ and }\gamma^{\rm HF}_{-}:= \theta_{-} \gamma^{\rm HF} \theta_{-}.
\end{equation*}

We prove the claim by showing that for all density matrices $\gamma \in \mathcal{A}$ such that $\supp(\rho_{\gamma}) \subset \R^3 \setminus B_r(0)$ and $\|\rho_{\gamma}\|_1 \leq \|\rho^{\rm HF} \chi_r^{+}\|_1 $ it holds that
\begin{equation}\label{Ea}
\mathcal{E}^{A}(\gamma^{\rm HF}_r ) + \mathcal{E}^{\rm HF}(\gamma^{\rm HF}_{-} ) - \mathcal{R} \leq \mathcal{E}^{\rm HF}(\gamma^{\rm HF}) \leq \mathcal{E}^{A}(\gamma )+ \mathcal{E}^{\rm HF}(\gamma^{\rm HF}_{-} ).
\end{equation}

The proof of the upper bound in \eqref{Ea} is as in \cite[page 533]{Sol}.

To prove the lower bound as a first step we localize. By Theorem~\ref{IMSrel} we find
\begin{equation}\label{guancia}
\alpha^{-1} \Tr[T(\bp)\gamma^{\rm HF}]=  \alpha^{-1} \Tr[T(\bp) (\gamma_{r}^{\rm HF}+\gamma_{0}^{\rm HF}+\gamma_{-}^{\rm HF} )]-\alpha^{-1}\sum_{i=1}^N (u_i, (L_{r}+L_{0}+L_{-}) u_i),
\end{equation}
where $L_{r}, L_{0}$ and $L_{-}$ are defined as the $L_{i}$'s in \eqref{Lalpha}.

We first estimate the error term. The procedure is similar to the one used in the proof of Lemma~\ref{ErrorA}. We introduce three cut-off functions: $\chi_{-}$ be the characteristic function of $B_{r(1-\lambda)(1-\nu)}(0)$, $\chi_{r}$ the characteristic function of $\R^3 \setminus B_{r\frac{1+\nu}{1-\lambda}}(0)$ and $\chi_{0}$ defined by $\chi_{0}(\bx)=1-\chi_{r}(\bx)-\chi_{-}(\bx)$ for all $\bx \in \R^3$. Notice that $\chi_{-}$ and $\chi_{r}$ are the characteristic functions of sets where $\theta_{-},\theta_{0}$ and $\theta_{r}$ are constants. For $k \in \{-,0,r\}$ we have the following splitting
\begin{equation*}
L_{k} = \chi_{-} L_{k} (\chi_{0}+\chi_{r})+ (\chi_0+\chi_{r}) L_{k} \chi_{-} + \chi_{r} L_{k} \chi_{0}+ \chi_0 L_{k} \chi_{r} + \chi_{0} L_{k} \chi_{0},
\end{equation*}
and proceeding as in the proof of Lemma~\ref{ErrorA} with $\varepsilon_{1,k}, \varepsilon_{2,k}$ to be chosen we find
\begin{eqnarray*}
(f,L_{k} f) & \leq  & \varepsilon_{1,k} \|\chi_{-} f\|_2^2 + \varepsilon_{1,k}^{-1} (f, Q_{1}f)+  \varepsilon_{2,k} \|\chi_0 f\|_2^2 + \varepsilon_{2,k}^{-1} (f,Q_{2} f)\\
&& + \tfrac{3\alpha}{2} \|\nabla \theta_{k}\|_{\infty}^2 \|\chi_0 f\|_2^2.
\end{eqnarray*}
with operators $Q_1$ and $Q_2$ being positive semi-definite operators with
\begin{eqnarray*}
\Tr[Q_1] & \leq  & \tfrac{(16)^2}{3\pi^2} \tfrac{(1-\lambda)^2(1-\nu)^3}{\nu} \|\nabla \theta_{k}\|_{\infty}^4 r^2 e^{-\alpha^{-1}r\nu(1-\lambda)} \\
\Tr[Q_2] & \leq  & \tfrac{(16)^2}{3\pi^2} \tfrac{1}{\nu(1-\lambda)^2} \|\nabla \theta_{k}\|_{\infty}^4 r^2 e^{-\alpha^{-1}r\frac{\nu}{1-\lambda}}.
\end{eqnarray*}
Choosing then
\begin{eqnarray*}
\varepsilon_{2,k} = \tfrac{3\alpha}{2} \|\nabla \theta_{k}\|_{\infty}^2  \mbox{ and }
\varepsilon_{1,k} = \alpha  \|\nabla \theta_{k}\|_{\infty}^2  e^{-\tfrac12 \alpha^{-1}r\nu(1-\lambda)},
\end{eqnarray*}
since $ (\|\nabla \theta_r\|_{\infty}^2+\|\nabla \theta_{0}\|_{\infty}^2+\|\nabla \theta_{-}\|_{\infty}^2 ) \leq 3\pi^2/(4\lambda^2) r^{-2}$ and $\| \rho^{\rm HF}\chi_{-}\|_{1} \leq N$ we get
\begin{eqnarray*}
\alpha^{-1} \sum_{i=1}^{N}(u_{i},(L_{r}+L_{0}+L_{-}) u_{i}) & \leq  & \tfrac{3\pi^2}{4\lambda^2} \, r^{-2} \| \rho^{\rm HF}\chi_0\|_{1} + \tfrac{3\pi^2}{4 \lambda^2} r^{-2} e^{-\tfrac12 \alpha^{-1}r\nu (1-\lambda)} N \\
&+& c \alpha^{-2} e^{-\tfrac12 \alpha^{-1}r\nu(1-\lambda)} .
\end{eqnarray*}
Here $c$ is a constant that depends only on $\nu$ and $\lambda$.

Hence from \eqref{guancia}, the inequality above and since $N \leq 2Z+1$ we find
\begin{eqnarray*}
\mathcal{E}^{\rm HF}(\gamma^{\rm HF})& \geq & \Tr\Big[\Big(\alpha^{-1}T(\bp)-\frac{Z}{|\cdot|}\Big)(\gamma_{r}^{\rm HF}+\gamma_{0}^{\rm HF}+\gamma_{-}^{\rm HF})\Big]+\mathcal{D}(\gamma^{\rm HF})\\
&&- \mathcal{E}x(\gamma^{\rm HF}) - \tfrac{3\pi^2}{4\lambda^2} \, r^{-2} \| \rho^{\rm HF}\chi_0\|_{1} - c' \alpha^{-2}(1+\alpha r^{-2}) e^{-\tfrac12 \alpha^{-1}rd} .
\end{eqnarray*}
The constants $c',d$ depend only on $\lambda$ and $\nu$. Proceeding as in \cite{Sol} we get
\begin{eqnarray*}
\mathcal{E}^{\rm HF}(\gamma^{\rm HF})& \geq & \mathcal{E}^{\rm HF}(\gamma^{\rm HF}_{-})+\mathcal{E}^{A}(\gamma^{\rm HF}_{r})- \mathcal{E}x(\gamma^{\rm HF}_{r})- c' \alpha^{-2}(1+\alpha r^{-2}) e^{-\tfrac12 \alpha^{-1}rd} \\
&& + \Tr\Big[\Big(\alpha^{-1}T(\bp)-\Phi^{\rm HF}_{r(1-\lambda)}(\cdot)\Big)\gamma_{0}^{\rm HF}\Big]\\
&& -( \tfrac{\pi}{2\lambda}+ \tfrac{3\pi^2}{4\lambda^2}r^{-1})r^{-1} \int_{|\bx|\geq r(1-\lambda)(1-\nu)}\rho^{\rm HF}(\bx) \; d\bx .
\end{eqnarray*}
The claim follows using Theorem~\ref{DaubeV}.
\end{proof}


\subsection{Comparing with an Outside Thomas Fermi}\label{Sotf}

At this point we introduce an \lq\lq Outside Thomas Fermi\rq\rq : a TF-energy functional whose minimizer approximates the HF-density at a certain distance from the nucleus.

Let $r>0$ such that
\begin{equation}\label{inter}
| \Phi_{|\bx|}^{\rm HF}(\bx) -\Phi_{|\bx|}^{\rm TF}(\bx)|  \leq \sigma |\bx|^{-4 + \varepsilon'},
\end{equation}
for all $|\bx| \leq r$ for some $\sigma >0$ and $\varepsilon ' >0$. Let $V_r$ be the potential defined by
\begin{equation}\label{Vr}
V_r(\bx) = \chi_r^{+}(\bx) \Phi^{\rm HF}_{r}(\bx) = \left\{
\begin{array}{ll}
0 & \mbox{ if }|\bx | <r ,\\
\Phi^{\rm HF}_{r}(\bx) & \mbox{ if } |\bx| \geq r.
\end{array} \right.
\end{equation}
Here and in the following $\chi_r^{+}(\bx):=1-\chi_r(\bx)$, $\bx \in \R^3$, where $\chi_r$ is the characteristic function of the ball of radius $r$ centered at $0$. Notice that $V_r \in L^{\frac52}(\R^3) + L^{\infty}(\R^3)$ with
\begin{equation*}
\inf \{ \| W\|_{\infty} : V_r-W \in L^{\frac52}(\R^3) \} =0 .
\end{equation*}
Let $\mathcal{E}^{\rm OTF}_{r}$ be the TF-functional $\mathcal{E}^{\rm TF}_{V_r}$ corresponding to the potential $V_r$ defined in \eqref{Vr}. Let $\rho_r^{\rm OTF}$ be the unique minimizer of $\mathcal{E}^{\rm OTF}_r$ under the condition
\begin{equation*}
\int_{\R^3} \rho(\bx) d\bx \leq \int_{|\by| \geq r} \rho^{\rm HF}(\by) d\by ,
\end{equation*}
(see Theorem~\ref{ExTF}). Then $\rho_r^{\rm OTF}$ is solution to the OTF-equation
\begin{equation}\label{OTFequation}
\tfrac12 \left(\tfrac{6\pi^2}{q}\right)^{\frac23} (\rho_r^{\rm OTF})^{\frac23}= [\varphi_r^{\rm OTF}-\mu_r^{\rm OTF} ]_{+} \, ,
\end{equation}
where
\begin{equation*}
\varphi_r^{\rm OTF}(\bx)= V_r(\bx) -\int_{\R^3} \frac{\rho^{\rm OTF}_{r}(\by)}{|\bx -\by|} d\by ,
\end{equation*}
is the OTF-mean field potential and $\mu_r^{\rm OTF}$ is the corresponding chemical potential. {F}rom \eqref{OTFequation} (and $\mu_r^{\rm OTF}\geq 0$) we see that the support of $\rho^{\rm OTF}_r$ is contained in $\R^3 \setminus B_r(0)$.

In the intermediary zone instead of comparing directly $\Phi^{\rm HF}_{|\bx|}$ and $\Phi_{|\bx|}^{\rm TF}$ we compare first the HF-density with the OTF-density and then the OTF-density with the TF-density. When comparing the TF and OTF there is no difference with the non-relativistic case and for brevity we refer for the proofs to \cite{Sol}.

We start by studying the behavior of the minimizer and mean field potential of the OTF. The proof of the following bounds is in \cite[page 557-558]{Sol} in the case $q=2$ and it can be directly generalised to the other values of $q$.
\begin{lemma}[\protect{\cite[Lem.12.1]{Sol}}]\label{est0ext}
For all $\by \in \R^3$ we have
\begin{equation*}
\varphi^{\rm TF}(\by) \leq 3^4 2^{-1}q^{-2} \pi^2 |\by|^{-4} \mbox{ and } \rho^{\rm TF}(\by) \leq 3^5 2^{-1} q^{-2} \pi |\by|^{-6}.
\end{equation*}
Let $\beta_0$ be as defined in Theorem~\ref{Slb}, then for all $|\by| \geq \beta_0 Z^{-\frac13}$ we have
\begin{equation*}
\varphi^{\rm TF}(\by) \geq C |\by|^{-4} \mbox{ and } \rho^{\rm TF}(\by) \geq C |\by|^{-6}.
\end{equation*}
With $r, \sigma, \varepsilon ' $ such that \eqref{inter} holds and $\sigma r^{\varepsilon '} \leq 1$  we have for all $|\by | \geq r$
\begin{equation*}
\rho_{r}^{\rm OTF} (\by) \leq C r^{-6} \mbox{ and } \varphi_r^{\rm OTF}(\by ) \leq |V_r(\by )| \leq C r^{-4}.
\end{equation*}
\end{lemma}

\begin{lemma}[\protect{\cite[Lem.12.2]{Sol}}]
With $r, \sigma, \varepsilon ' $ such that \eqref{inter} holds for all $|\bx| \leq r$ we have
\begin{equation*}
\int_{|\by| \geq r} (\rho^{\rm TF}(\by) -\rho^{\rm HF}(\by)) d\by \leq \sigma r^{-3+\varepsilon'}.
\end{equation*}
\end{lemma}

For $\bx \in \R^3$ with $|\bx| > r$ we may write
\begin{equation}\label{A123}
\Phi_{|\bx|}^{\rm HF}(\bx) - \Phi_{|\bx|}^{\rm TF}(\bx) = \mathcal{A}_1(r,\bx)+ \mathcal{A}_2(r,\bx)+ \mathcal{A}_3(r,\bx),
\end{equation}
where
\begin{eqnarray*}
\mathcal{A}_1(r,\bx)&=& \varphi_r^{\rm OTF}(\bx) - \varphi^{\rm TF}(\bx) , \\
\mathcal{A}_2(r,\bx)&=& \int_{|\by|>|\bx|} \frac{\rho_r^{\rm OTF}(\by)-\rho^{\rm TF}(\by)}{|\bx-\by|} \, d\by \,
\end{eqnarray*}
and
\begin{equation*}
\mathcal{A}_3(r,\bx)= \int_{r<|\by|<|\bx|} \frac{\rho_r^{\rm OTF}(\by)-\rho^{\rm HF}(\by)}{|\bx-\by|} \, d\by .
\end{equation*}

\subsubsection{Estimate on $\mathcal{A}_1$ and $\mathcal{A}_2$}

\begin{lemma}[\protect{\cite[Lem.12.4]{Sol}}]\label{mu=0}
Let $N \geq Z$. Given $\varepsilon ', \sigma >0$ there exists a constant $D>0$ such that for all $r$ with $\beta_0 Z^{-\frac13} \leq r \leq D$ for which \eqref{inter} holds for all $|\bx| \leq r$, then $\mu_r^{\rm OTF} =0 $ and
\begin{equation*}
\tfrac{3^4 \pi^2}{2q^2} |\bx|^{-4}(1+ a r^{\zeta} |\bx|^{-\zeta} )^{-2} \leq \varphi_r^{\rm OTF}(\bx) \leq \tfrac{3^4 \pi^2}{2q^2} |\bx|^{-4}(1+ A r^{\zeta} |\bx|^{-\zeta} ) \; \mbox{ for }|\bx|>r ,
\end{equation*}
where $a, A$ are universal constants and $\zeta= (-7+\sqrt{73})/2$.
\end{lemma}

\begin{lemma}[\protect{\cite[Lem.12.5]{Sol}}]\label{A1A2}
Let $N \geq Z$. Given $\varepsilon ', \sigma >0$ there exists a constant $D>0$ depending only on $\varepsilon'$, $\sigma$ such that for all $r$ with $\beta_0 Z^{-\frac13} \leq r \leq D$ for which \eqref{inter} holds for $|\bx| \leq r$, then for all $|\bx| \geq r$
\begin{equation*}
|\mathcal{A}_1(r,\bx)| \leq C |\bx|^{-4-\zeta} r^{\zeta} \; \; \mbox{ and } \;
|\mathcal{A}_2(r,\bx)| \leq C |\bx|^{-4-\zeta} r^{\zeta},
\end{equation*}
with $\zeta = (-7 + \sqrt{73})/2$ and $C$ a universal constant.
\end{lemma}

The proof of the previous lemmas is in \cite[p. 558-564]{Sol}.

\subsubsection{Estimate on $\| \chi_{r}^{+} \rho^{\rm HF} -\rho_{r}^{\rm OTF}\|_{C}$ and $\Tr[T(\bp) \gamma_{r}^{\rm HF}]$}

\begin{lemma}\label{extGest}
Let $G_{\alpha}$ be the function defined in Theorem~\ref{Daube} and $\rho_r^{\rm HF}(\bx)$ be the one-particle density of the density matrix $\gamma_r^{\rm HF}$ defined in \eqref{exteriordensity}. Let $Z \alpha= \kappa $ fixed, $0 \leq \kappa < 2/\pi$ and $Z \geq 1$.

Given constants $\varepsilon', \sigma >0$ there exists $D< \frac45$ such that for all $r$ with $\beta_0 Z^{-\frac13} \leq r \leq D$ for which \eqref{inter} holds for $|\bx|\leq r$, it follows that
\begin{equation*}
\alpha^{-1} \int_{\R^3} G_{\alpha}(\rho_r^{\rm HF}(\bx)) \; d\bx \leq \alpha^{-1} \Tr[T(\bp)\gamma_{r}^{\rm HF}] \leq 2 \mathcal{R}+ Cr^{-7} + C r^{-4} \int_{\R^3} \rho^{\rm HF}_r(\bx) \; d\bx ,
\end{equation*}
with $C$ a universal positive constant and $\mathcal{R}$ as defined in Theorem~\ref{EA}.
\end{lemma}

\begin{proof}The first inequality follows directly from Theorem~\ref{Daube}. To prove the second inequality we proceed as in Lemma~\ref{rrcou}. In this case we are interested only in the exterior part of the minimizer. Hence, instead of considering the HF-energy functional we consider the auxiliary functional $\mathcal{E}^{A}$, defined in \eqref{ea}, applied to the \lq\lq exterior part of the minimizer\rq\rq  $\gamma_r^{\rm HF}$.

Splitting the kinetic energy in two terms we find
\begin{equation}\label{neapol}
\mathcal{E}^A(\gamma^{\rm HF}_r) \geq \tfrac12 \alpha^{-1} \Tr[T(\bp) \gamma^{\rm HF}_{r}] + D(\rho_r^{\rm HF}) + \tfrac12 \Tr[ (\alpha^{-1} T(\bp) -2 \Phi_{r}^{\rm HF}) \gamma_r^{\rm HF} ].
\end{equation}
Since $\Phi^{\rm HF}_r(\bx)$ is harmonic for $|\bx|>r$ and going to zero at infinity
\begin{equation*}
\Phi^{\rm HF}_r(\bx) \leq \frac{r}{|\bx|} \sup_{|\by|=r} \Phi_r^{\rm HF}(\by) \; \mbox{ for }|\bx| >r.
\end{equation*}
Hence, since $\supp(\rho^{\rm HF}_r) \subset \R^3 \setminus B_r(0)$ we find
\begin{equation*}
\Tr[ (\alpha^{-1} T(\bp) -2 \Phi_{r}^{\rm HF}) \gamma_r^{\rm HF} ]\geq \Tr[ (\alpha^{-1} T(\bp)
-\frac{2 r}{|\cdot|} \sup_{|\by|=r}\Phi^{\rm HF}_r(\by)) \gamma_r^{\rm HF} ] =\dots .
\end{equation*}
Adding and subtracting $2 D(\rho, \rho_{r}^{\rm HF})$ for $\rho \in L^1(\R^3) \cap L^{\frac53}(\R^3)$, $\rho \geq 0$, to be chosen
\begin{eqnarray}\label{rit0}
\dots &=& \Tr[ (\alpha^{-1} T(\bp) - V_{\rho}) \gamma_r^{\rm HF} ] - \int_{\R^3} \int_{\R^3}\frac{\rho_r^{\rm HF}(\bx)\rho(\by)}{|\bx-\by|} \; d\bx d\by .
\end{eqnarray}
where for simplicity of notation here and in the following $V_{\rho}$ is defined as $V_{\rho}(\bx):= \frac{2 r}{|\bx|} \sup_{|\by|=r}\Phi^{\rm HF}_r(\by)
 - \rho* \frac{1}{|\bx|}$.

{F}rom \eqref{rit0}, \eqref{neapol} and the definition of the Coulomb norm and scalar product (Definition~\ref{Cnorm}) we find
\begin{eqnarray} \notag
\mathcal{E}^A(\gamma^{\rm HF}_r) & \geq & \tfrac12 \alpha^{-1} \Tr[T(\bp) \gamma^{\rm HF}_{r}]+ \tfrac12 D(\rho^{\rm HF}_r)+\tfrac12 \|\rho_r^{\rm HF}
 -\rho \|_{C}^2\\
&&  - \tfrac12 D(\rho) + \tfrac12 \Tr[ (\alpha^{-1} T(\bp) -V_{\rho}) \gamma_r^{\rm HF} ]  \label{neapo}\\ \notag
&\geq & \tfrac12 \alpha^{-1}  \Tr[T(\bp) \gamma^{\rm HF}_{r}]+ \tfrac12 \sum_{i=1}^N (\theta_r u_i,(\alpha^{-1} T(\bp)-V_{\rho}) \theta_r u_i)
- \tfrac12 D(\rho),\end{eqnarray}
denoting by $u_i$ the HF-orbitals.

We now choose $\rho$ as the minimizer of the TF-energy functional of a neutral atom with Coulomb potential and nuclear charge
$2r \sup_{|\by|=r}\Phi^{\rm HF}_r(\by)$. Then $V_{\rho}$ is the corresponding TF-mean field potential and we see that the last
two terms on the right hand side of \eqref{neapo} are like the ones in the claim of Proposition~\ref{prthomas3}.
The only difference is due to the presence of the localization function $\theta_r$. We now prove that these terms give the
TF-energy modulo lower order terms. The method is the same as that of Proposition~\ref{prthomas3}.
We repeat the main steps since in this case the scaling depends on $r$. Notice that since $r> \beta_0 Z^{-\frac13}$
the contribution is coming only from the \lq\lq outer zone\rq\rq .

Let $g \in C^{\infty}_0(\R^3)$ be spherically symmetric, normalized in $L^2(\R^3)$ and with support in $B_1(0)$. Let us define $g_r(\bx):= r^{-3} g(\bx r^{-2}) \mbox{ and }\psi_r:=g_r^2$. Since $V_\rho$ is subharmonic on $|\bx|>0$, we see from the support properties of $\psi_r$ and $\theta_r$ that
\begin{equation*}
\sum_{i=1}^N (\theta_r u_i,(\alpha^{-1} T(\bp)-V_{\rho}) \theta_r u_i) \geq  \sum_{i=1}^N (\theta_r u_i, (\alpha^{-1} T(\bp)-V_{\rho}*\psi_r) \theta_r u_i )= \dots .
\end{equation*}
For $\bp, \bq \in \R^3$ we define the coherent states $g_r^{\bp,\bq}(\bx):= g_r(\bx-\bq)e^{i \bp \cdot \bx}$. By the formulas \eqref{coheformulas} and \eqref{Lcohe} with $L_{\bq}$ the operator defined in the equation below \eqref{Lcohe} we get
\begin{eqnarray}\notag
\dots & =&  \tfrac{1}{(2\pi)^3} \alpha^{-1} \int_{\R^3} \int_{\R^3} d\bp d\bq \;  (T(\bp)-\alpha V_{\rho}(\bq)) \sum_{i=1}^N \sum_{j=1}^q |(\theta_r u_i^j, g_r^{\bp,\bq})|^2 \\ \label{neap}
&& - \, \alpha^{-1} \sum_{i=1}^N  \int_{\R^3} \int_{\R^3} d\bx d\bq \,  \overline{(\theta_r u_i)}(\bx) (L_{\bq} \theta_r u_i)(\bx) \, ,
\end{eqnarray}
where $u_i^j$ denotes the $j$-th spin component of the orbital $u_i$. By the choice of the function $g_r$ and with the same arguments that led to \eqref{br2} in the appendix we find
\begin{eqnarray}\notag
&& \alpha^{-1} \sum_{i=1}^N \int_{\R^3} \int_{\R^3} d\bx d\bq \,  \overline{(\theta_r u_i)}(\bx) (L_{\bq} \theta_r u_i)(\bx) \\
& \leq &3 \sum_{i=1}^N \|\theta_r u_i\|_2^2 \| \nabla g_r \|_{\infty}^2 Vol(\supp(g_r)) \leq C r^{-4}\|\rho^{\rm HF}_r\|_1.  \label{adv1}
\end{eqnarray}

In the first term on the right hand side of \eqref{neap} the integrand is zero if $|\bq| < \frac14 r^{2}$ since in this case $\supp(\theta_r) \cap \supp(g_r^{\bq,\bp}) = \emptyset$ (by the choice $D<4/5$). To estimate it further from below we consider only the negative part of the integrand
\begin{eqnarray}\notag
&& \tfrac{1}{(2\pi)^3} \alpha^{-1} \int_{\R^3} \int_{\R^3} d\bp d\bq \;  (T(\bp)-\alpha V_{\rho}(\bq)) \sum_{i=1}^N \sum_{j=1}^q |(\theta_r u_i^j, g_r^{\bp,\bq})|^2\\
&\geq & \tfrac{q}{(2\pi)^3} \alpha^{-1} \iint_{\substack{|\bq|>\frac14 r^2\\T(\bp) \leq \alpha V_{\rho}(\bq)}} d\bp d\bq \;  (T(\bp)-\alpha V_{\rho}(\bq)) \, , \label{adv2}
\end{eqnarray}
where we have used that $0 \leq  \sum_{i=1}^N |(\theta_r u_i^j, g_r^{\bp,\bq})|^2 \leq 1$ (Bessel's inequality). We split the domain of integration in $\bp$ as follows\begin{equation*}
\{ \bp \in \R^3: T(\bp) \leq \alpha V_{\rho}(\bq)\} =  \Sigma_1 \cup \Sigma_2
\end{equation*}
with $\Sigma_1,\Sigma_2$ disjoint and $\Sigma_1=\{ \bp \in \R^3: \tfrac12 |\bp|^2 \leq V_{\rho}(\bq)\}$. We treat these two contributions separately. We have
\begin{equation*}
\alpha^{-1} \iint_{\substack{|\bq|>\frac14 r^2\\ \bp \in \Sigma_2}} d\bp d\bq \;  (T(\bp)-\alpha V_{\rho}(\bq)) \geq - \iint_{\substack{|\bq|>\frac14 r^2\\\bp \in \Sigma_2}} d\bp d\bq \;  [V_{\rho}(\bq)]_{+} = \dots
\end{equation*}
and computing the integral, using that $(1+x)^{\frac32} \leq 1+\frac32 x +\frac38 x^2$
\begin{equation}\label{rit1}
\dots \geq - C  \int_{|\bq|>\frac14 r^2} d\bq \;  (\alpha^2 [V_{\rho}(\bq)]^{\frac72}_{+}+\alpha^4 [V_{\rho}(\bq)]_{+}^{\frac92}) \geq  -C \alpha^2 r^{-\frac{23}{2}} -C \alpha^4 r^{-\frac{33}{2}}.
\end{equation}
In the last step we used that $[V_{\rho}(\bq)]_{+} \leq 2 \frac{r}{|\bq|} \sup_{|\bx|=r} \Phi^{\rm HF}_{r}(\bx)$ and that by the hypothesis and Corollary~\ref{TFscr}
\begin{equation}\label{r3}
r\sup_{|\bx|=r} \Phi^{\rm HF}_{r}(\bx) \leq C r^{-3},
\end{equation}
choosing $D$ such that $\sigma r^{\varepsilon'}\leq 1$.

Since $T(\bp) \geq \frac12 \alpha |\bp|^2 -\frac18 \alpha^3 |\bp|^4$ we find
\begin{eqnarray}\notag
&& \alpha^{-1} \iint_{\substack{|\bq|>\frac14 r^2\\ \bp \in \Sigma_1}} d\bp d\bq \;  (T(\bp)-\alpha V_{\rho}(\bq))\\
& \geq & \iint_{\substack{|\bq|>\frac14 r^2\\ \tfrac{1}{2}|\bp|^2 \leq V_{\rho}(\bq)}} d\bp d\bq \;  (\frac{1}{2}|\bp|^2-V_{\rho}(\bq)) -\tfrac18 \alpha^2 \iint_{\substack{|\bq|>\frac14 r^2\\ \frac{1}{2}|\bp|^2 \leq V_{\rho}(\bq)}} d\bp d\bq \;  |\bp|^4 . \label{rit2}
\end{eqnarray}
Computing the last integral we find
\begin{equation}\label{rit1a}
\alpha^2 \iint_{\substack{|\bq|>\frac14 r^2\\ \frac{1}{2}|\bp|^2 \leq V_{\rho}(\bq)}} d\bp d\bq \;  |\bp|^4 \leq C \alpha^2  r^{-1} (2r\sup_{|\bx|=r} \Phi^{\rm HF}_{r}(\bx))^{\frac72} \leq C \alpha^2 r^{-\frac{23}{2}}.
\end{equation}
While for the first term on the right hand side of \eqref{rit2}, computing the integral with respect to $\bp$, we get
\begin{eqnarray*}
\iint_{\substack{|\bq|>\frac14 r^2\\ \frac{1}{2}|\bp|^2 \leq V_{\rho}(\bq)}} d\bp d\bq \;  (\tfrac{1}{2}|\bp|^2-V_{\rho}(\bq)) = - 4\pi \tfrac{2^{\frac52}}{15} \int_{|\bq|>\frac14 r^2} d\bq \; [V_{\rho}(\bq)]_{+}^{\frac52}.
\end{eqnarray*}
Hence collecting together \eqref{neap}, \eqref{adv1}, \eqref{adv2} \eqref{rit1}, \eqref{rit1a} and the inequality above we find
\begin{equation*}
\Tr[ (\alpha^{-1} T(\bp)-V_{\rho}) \gamma_r^{\rm HF} ]\geq - \tfrac{2^{\frac32}q}{15 \pi^2 } \int_{\R^3} d\bx \; [V_{\rho}(\bx)]_{+}^{\frac52} - C r^{-4}\|\rho^{\rm HF}_r\|_1  -C r^{-\frac{11}{2}}=\dots.
\end{equation*}
since $\beta_0 Z^{-\frac13}\leq r$ implies $\beta_0 \alpha^{\frac13} \leq \kappa^{\frac13} r $. {F}rom the TF-equation that $\rho$ satisfies it follows that
\begin{eqnarray*}
\dots & = & \tfrac{3}{10}(\tfrac{6 \pi^2}{q})^{\frac23} \int_{\R^3} d\bx \; \rho(\bx)^{\frac53} -\int_{\R^3} \rho(\bx) V_{\rho}(\bx) \; d\bx - C r^{-4}\|\rho^{\rm HF}_r\|_1  -C r^{-\frac{11}{2}}\\
& = & \mathcal{E}^{\rm TF}(\rho) +D(\rho) - C r^{-4}\|\rho^{\rm HF}_r\|_1  -C r^{-\frac{11}{2}}.
\end{eqnarray*}

Hence from \eqref{neapo} and the inequality above we get using \eqref{Tfe} and  \eqref{r3}
\begin{eqnarray*}
\mathcal{E}^A(\gamma^{\rm HF}_r) & \geq &  \tfrac12 \alpha^{-1} \Tr[T(\bp)\gamma^{\rm HF}] - C r^{-7} - C r^{-4}\|\rho^{\rm HF}_r\|_1.
\end{eqnarray*}
The claim follows since $\mathcal{E}^{A}(\gamma^{\rm HF}_r) \leq \mathcal{R}$ by the result of Theorem~\ref{EA} considering as a trial density matrix $\gamma \equiv 0$.
\end{proof}

\begin{lemma}\label{auxi}
Let $N' \in \mathbb{N}$ and $Z \alpha= \kappa $ be fixed, $0 \leq \kappa < 2/\pi$ and $Z \geq 1$. Let $e_{j}$ be the first $N'$ negative eigenvalues of the operator $\alpha^{-1}T(\bp)-\varphi_{r}^{\rm OTF}$ acting on functions with support on $\{\bx \in \R^3 : |\bx| \geq r\}$.

Given constants $\varepsilon', \sigma >0$ there exists $D< 4/5$ such that for all $r$ with $\beta_0 Z^{-\frac13} \leq r \leq D$ for which \eqref{inter} holds for $|\bx|\leq r$, for all $\mu \in (0,1)$ and $s<r$ we have
\begin{eqnarray*}
\sum_{j=1}^{N'} e_{j} & \geq & -(\tfrac{2}{1-\mu})^{\frac32} \tfrac{1}{15 \pi^2} \int_{|\bq|>r} [\varphi_{r}^{\rm OTF}(\bq)]_{+}^{\frac52} \; d\bq - C r^{-8} s \mu^{-\frac32}- C \mu^{-3} r^{-5} s\\
&&-C (1-\mu)^{-\frac72}r^{-5} -C (1-\mu) s^{-2}N',
\end{eqnarray*}
with $C$ a positive constant.
\end{lemma}

\begin{proof}
Let $f_{j}$ be the eigenfunctions (normalized in $L^2(\R^3,\C^q)$) corresponding to the eigenvalues $e_{j}$, $j=1, .., N'$. Let $g \in C_{0}^{\infty}(\R^3)$ with support in $B_{1}(0)$ and define $g_{s}(\bx)=s^{-\frac32} g(\bx/s)$ for a positive parameter $s$, $s<r$. We then write for $\mu \in (0,1)$
\begin{equation*}
\sum_{j=1}^{N'} e_{j}= \sum_{j=1}^{N'} (f_j, (\alpha^{-1} T(\bp) - \varphi_{r}^{\rm OTF}) f_{j}) = \mathcal{B}_1+ \mathcal{B}_2,
\end{equation*}
where
\begin{eqnarray*}
\mathcal{B}_1 & = & \sum_{j=1}^{N'} (f_j, ((1-\mu)\alpha^{-1} T(\bp) - \varphi_{r}^{\rm OTF}*g_{s}^2) f_{j}) ,\\
\mathcal{B}_2 & = & \sum_{j=1}^{N'} (f_j, (\mu \alpha^{-1} T(\bp) - \varphi_{r}^{\rm OTF}+\varphi_{r}^{\rm OTF}*g_{s}^2) f_{j}) .
\end{eqnarray*}
We estimate these two terms separately. Considering for $\bp, \bq \in \R^3$ the coherent states $g_{s}^{\bp,\bq}(\bx):= e^{i\bp.\bx} g_{s}(\bx-\bq)$ using \eqref{coheformulas} and \eqref{Lcohe}, we find
\begin{eqnarray}\notag
\mathcal{B}_1 & = & \tfrac{1}{(2\pi)^3} \iint ((1-\mu)\alpha^{-1}T(\bp)-\varphi_{r}^{\rm OTF}(\bq)) \sum_{j=1}^N |(f_j, g_s^{\bp,\bq})|^2\; d\bq d\bp \\
&& - \,  (1-\mu) \alpha^{-1} \sum_{j=1}^{N'} \int_{\R^3} \int_{\R^3} d\bx  d\bq  \overline{f_{j}(\bx)} (L_{\bq}f_{j})(\bx) \, .\label{bici5}
\end{eqnarray}
Estimating the error term as done in \eqref{ce} and previous inequalities we get
\begin{equation*}
(1-\mu) \alpha^{-1} \sum_{j=1}^{N'} \int_{\R^3} \int_{\R^3} d\bx  d\bq  \overline{f_{j}(\bx)} (L_{\bq}f_{j})(\bx) \leq C (1-\mu) s^{-2} N'.
\end{equation*}
Since we are interested in an estimate from below and $\varphi_{r}^{\rm OTF}(\bq) \leq 0 $ for $|\bq| < r$, from \eqref{bici5} we find
\begin{eqnarray}\notag
\mathcal{B}_1 & \geq & \tfrac{1}{(2\pi)^3} \iint_{|\bq|>r} ((1-\mu)\alpha^{-1}T(\bp)-\varphi_{r}^{\rm OTF}(\bq)) \sum_{j=1}^N |(f_j, g_s^{\bp,\bq})|^2 \; d\bq d\bp\\
&&  -C (1-\mu) s^{-2} N' . \label{bici6}
\end{eqnarray}
We estimate now the first term on the right hand side of \eqref{bici6}. Considering only the negative part of the integrand and since $\sum_{j=1}^{N'} |(f_j, g_s^{\bp,\bq})|^2 \leq 1$ we get
\begin{eqnarray*}
&& \tfrac{1}{(2\pi)^3} \iint_{|\bq|>r} ((1-\mu)\alpha^{-1}T(\bp)-\varphi_{r}^{\rm OTF}(\bq)) \sum_{j=1}^{N'} |(f_j, g_s^{\bp,\bq})|\; d\bq d\bp\\
& \geq & \tfrac{1}{(2\pi)^3} \iint_{\substack{|\bq|>r,\\ (1-\mu)\alpha^{-1}T(\bp)\leq \varphi^{\rm OTF}_{r}(\bq)}} ((1-\mu)\alpha^{-1}T(\bp)-\varphi_{r}^{\rm OTF}(\bq)) \; d\bp d\bq.
\end{eqnarray*}
Now we split the domain of integration in $\bp$ as follows
\begin{equation*}
\{ \bp \in \R^3: \alpha^{-1}(1-\mu) T(\bp) \leq \varphi_{r}^{\rm OTF}(\bq)\} =  \Sigma_1 \cup \Sigma_2 ,
\end{equation*}
with $\Sigma_1,\Sigma_2$ disjoint and $\Sigma_1=\{ \bp \in \R^3: (1-\mu) |\bp|^2/2 \leq \varphi_{r}^{\rm OTF}(\bq)\}$. We treat these two contributions separately. Then
\begin{eqnarray*}
&& \tfrac{1}{(2\pi)^3} \iint_{\substack{|\bq|>r,\\ \bp \in \Sigma_{2}}} ((1-\mu)\alpha^{-1} T(\bp)-\varphi_{r}^{\rm OTF}(\bq)) d\bp d\bq \\
& \geq & - \tfrac{1}{(2\pi)^3} \iint_{\substack{|\bq|>r,\\ \bp \in \Sigma_{2}}} [\varphi_{r}^{\rm OTF}(\bq)]_{+} d\bp d\bq = \dots
\end{eqnarray*}
and since in the domain of integration
$$\tfrac{2}{1-\mu} [\varphi_{r}^{\rm OTF}(\bq)]_{+} \leq |\bp|^2 \leq \tfrac{2}{1-\mu} [\varphi_{r}^{\rm OTF}(\bq)]_{+}( 1+\tfrac{1}{2(1-\mu)}\alpha^2 [\varphi_{r}^{\rm OTF}(\bq)]_{+})\, $$
we get
\begin{eqnarray}\notag
\dots & \geq & -  \tfrac{C}{(1-\mu)^{\frac52}} \alpha^2  \int_{|\bq|>r} d\bq \;  ([\varphi^{\rm OTF}_{r}(\bq)]^{\frac72}_{+}+\tfrac{\alpha^2}{8(1-\mu)} [\varphi^{\rm OTF}_{r}(\bq)]_{+}^{\frac92})\\ \label{bici7}
& \geq & - \tfrac{C}{(1-\mu)^{\frac52}} \alpha^2 (r^{-11}+\tfrac{\alpha^2}{1-\mu} r^{-15}),
\end{eqnarray}
using Lemma~\ref{mu=0} in the last step.

Since $\sqrt{1+t^2}\geq 1+(1/2)t^2-(1/8)t^4$, we get
\begin{eqnarray*}
&&  \tfrac{1}{(2\pi)^3} \iint_{\substack{|\bq|>r,\\ \bp \in \Sigma_{1}}} ((1-\mu)\alpha^{-1}T(\bp)-\varphi_{r}^{\rm OTF}(\bq)) d\bp d\bq \\
& \geq &\tfrac{1}{(2\pi)^3} \iint_{\substack{|\bq|>r,\\ \bp \in \Sigma_{1}}} ((1-\mu)\tfrac12 |\bp|^2-\varphi_{r}^{\rm OTF}(\bq) -\tfrac18 (1-\mu) \alpha^2 |\bp|^4) d\bp d\bq .
\end{eqnarray*}
The last term gives by Lemma~\ref{mu=0}
\begin{equation}\label{bc8}
\alpha^2 \iint_{\substack{|\bq|>r\\  \bp \in \Sigma_1}} d\bp d\bq \;  |\bp|^4 =\alpha^{2} \tfrac{4\pi}{7} \int_{|\bq|>r} d\bq \; (\tfrac{2}{1-\mu})^{\frac72} [\varphi_{r}^{\rm OTF}(\bq)]_{+}^{\frac72} \leq C \alpha^2 (\tfrac{2}{1-\mu})^{\frac72} r^{-11}.
\end{equation}
While for the other terms computing the integral with respect to $\bp$, we get
\begin{equation}\label{bc9}
\tfrac{1}{(2\pi)^3} \iint_{\substack{|\bq|>r,\\ \bp \in \Sigma_{1}}} ((1-\mu)\tfrac12 |\bp|^2-\varphi_{r}^{\rm OTF}(\bq)) d\bp d\bq = -(\tfrac{2}{1-\mu})^{\frac32} \tfrac{1}{15 \pi^2} \int_{|\bq|>r} d\bq \; [\varphi_r^{\rm OTF}(\bq)]^{\frac52}_{+}.
\end{equation}

For the term $\mathcal{B}_2$ using Theorem~\ref{DaubeV} and Remark~\ref{DaubeR} we find
\begin{equation*}
\mathcal{B}_2 \geq -C q ( \mu^{-\frac32} \| [\varphi_{r}^{\rm OTF} -\varphi_{r}^{\rm OTF}*g_s^2]_{+}\|_{\frac52}^{\frac52} + \alpha^3 \mu^{-3} \| [\varphi_{r}^{\rm OTF} -\varphi_{r}^{\rm OTF}*g_s^2]_{+}\|_{4}^{4} ).
\end{equation*}
{F}rom the choice of $g_s$ it follows that $\varphi_{r}^{\rm OTF} -\varphi_{r}^{\rm OTF}*g_s^2 \leq V_{r} -V_{r}*g_s^2 $ and the term $V_{r} -V_{r}*g_s^2$ is non-zero only for $r-s \leq |\bx| \leq r+s$. Hence by Lemma~\ref{est0ext} and since $s<r$\begin{equation} \label{bc10}
\| [\varphi_{r}^{\rm OTF} -\varphi_{r}^{\rm OTF}*g_s^2]_{+}\|_{\frac52}^{\frac52} \leq \int_{r-s \leq |\bx| \leq r+s} [V_r(\bx) -V_r* g^2 (\bx)]_{+}^{\frac52} d \bx \leq C r^{-8}s,
\end{equation}
and similarly $ \| [\varphi_{r}^{\rm OTF} -\varphi_{r}^{\rm OTF}*g_s^2]_{+}\|_{4}^{4} \leq C r^{-14} s$. The claim follows from \eqref{bici6},  \eqref{bici7}, \eqref{bc8}, \eqref{bc9} and \eqref{bc10} using that $\beta_0 \alpha^{\frac13} \leq \kappa^{\frac13} r$.
\end{proof}

\begin{lemma}\label{extrCou}
Let $G_{\alpha}$ be the function defined in Theorem~\ref{Daube} and $\rho_r^{\rm HF}(\bx)$ the one-particle density of the density matrix $\gamma_r^{\rm HF}$ defined in \eqref{exteriordensity}. Let $Z \alpha= \kappa $ be fixed, $0 \leq \kappa <2/\pi$ and $Z \geq 1$.

There exists $\alpha_0>0$ such that given $\varepsilon', \sigma >0$ there exists $D< 1/4$ such that for all $\alpha \leq \alpha_0$ and $r$ with $\beta_0 Z^{-\frac13} \leq r \leq D$ for which \eqref{inter} holds for $|\bx| \leq r$, we have
\begin{eqnarray} \label{rrcx}
\begin{array}{ll}
\displaystyle{\| \chi_r^{+} \rho^{\rm HF} - \rho^{\rm OTF}_{r}\|_{C}  \leq C r^{-\frac72+\frac16}} &  \mbox{ and } \vspace{.1cm}\\
\displaystyle{\alpha^{-1} \int_{\R^3} G_{\alpha}(\chi_{r}^{+} \rho^{\rm HF}(\bx)) d \bx \leq C r^{-7}}, &
\displaystyle{\alpha^{-1} \Tr[T(\bp) \gamma^{\rm HF}_{r}] \leq C r^{-7}},
\end{array}
\end{eqnarray}
with $C$ a universal positive constant.
\end{lemma}

\begin{proof}
The idea of the proof is the same as that of Lemma~\ref{rrcou}. In this case we are interested only in the exterior part of the minimizer. Hence, instead of considering the HF-energy functional we estimate from above and below the auxiliary one $\mathcal{E}^{A}$, defined in \eqref{ea}, applied on the \lq\lq exterior part of the minimizer\rq\rq  $\gamma_r^{\rm HF}$.

\textit{Step I. Estimate from above on $\mathcal{E}^{A}(\gamma^{\rm HF}_r)$.} Let us consider $\gamma$ the density matrix that acts identically on the spin components and on each as \begin{equation*}
\gamma^{j}=\tfrac{1}{(2 \pi)^3} \iint_{\tfrac12 |\bp|^2\leq\varphi_{r}^{\rm OTF}(\bq)} \; \Pi_{\bp,\bq} \; d\bp d\bq ,
\end{equation*}
where $j \in \{1, \dots,q\}$ is the spin index, $\Pi_{\bp,\bq}$ is the projection onto the space spanned by $h_s^{\bp,\bq}(\bx)=h_s(\bx-\bq) e^{i \bp .\bx }$ where $h_s$ is the ground state for the Dirichlet Laplacian on the ball of radius $s$ for $0<s<r$. By the OTF-equation \eqref{OTFequation} and since $\mu_{r}^{\rm OTF}=0$ (see Lemma~\ref{mu=0}) we see that $\rho_{\gamma}(\bx)=\rho^{\rm OTF}_{r}*|h_{s}|^2(\bx)$. Moreover, by Lemma~\ref{mu=0}
\begin{equation}\label{bici2}
\Tr[-\tfrac12 \Delta \gamma ]=\tfrac{3}{10} (\tfrac{6\pi^2}{q})^{\frac23} \int_{\R^3} (\rho_{r}^{\rm OTF}(\bx))^{\frac53} \; d\bx + C s^{-2} r^{-3}.
\end{equation}
Since $[\Phi_{r}^{\rm HF}]_{+} \in L^{\frac52}_{loc}(\R^3)$, by \cite[Lemma 8.5]{Sol} for $\lambda' \in (0,1)$ we may find $\tilde{\gamma}$ such that $\supp{(\rho_{\tilde{\gamma}})} \subset \{ \bx: |\bx| \geq r\}$, $\rho_{\tilde{\gamma}}(\bx)\leq \rho_{\gamma}(\bx)$ for $\bx \in \R^3$ and
\begin{eqnarray}\notag
\Tr[ (-\tfrac12 \Delta -\Phi_{r}^{\rm HF}) \tilde{\gamma}] & \leq & \Tr[ (-\tfrac12 \Delta -\chi_r^{+}\Phi_{r}^{\rm HF}) \gamma] + L_{1} \int_{|\bx|\leq \frac{r}{1-\lambda'}} [V_{r}(\bx)]_{+}^{\frac52} \; d\bx\\ \label{bike1}
&& + \tfrac12 (\tfrac{\pi}{2\lambda'r})^2 \int_{|\bx| \leq \frac{r}{1-\lambda'}} \rho_{\gamma}(\bx) \; d\bx .
\end{eqnarray}
Since $\int \rho_{\tilde{\gamma}}\leq \int \rho_{\gamma} = \int \rho_{r}^{\rm OTF} \leq \int \chi_{r}^{+} \rho^{\rm HF}$ we may choose $\tilde{\gamma}$ as a trial density matrix in Theorem~\ref{EA} and we find for $\lambda, \nu $ to be chosen
\begin{equation*}
\mathcal{E}^{A}(\gamma_{r}^{\rm HF}) \leq \mathcal{E}^{A}(\tilde{\gamma})+\mathcal{R} \leq \Tr[(-\tfrac12 \Delta -\Phi_{r}^{\rm HF}) \tilde{\gamma}] +\mathcal{R} + D(\rho_{\tilde{\gamma}}) ,
\end{equation*}
since $\alpha^{-1} T(\bp) \leq \frac12 |\bp|^2$. Notice that $\mathcal{R}$ depends on $\lambda$ and $\nu$. {F}rom \eqref{bike1} it follows that
\begin{eqnarray}\notag
\mathcal{E}^{A}(\gamma_{r}^{\rm HF})  & \leq & \Tr[ (-\tfrac12 \Delta -\chi_r^{+}\Phi_{r}^{\rm HF}) \gamma] + L_{1} \int_{|\bx|\leq \frac{r}{1-\lambda'}} [V_{r}(\bx)]_{+}^{\frac52} \; d\bx\\ \label{Kop1}
&& + \tfrac12 (\tfrac{\pi}{2\lambda'r})^2 \int_{|\bx| \leq \frac{r}{1-\lambda'}} \rho_{\gamma}(\bx) \; d\bx +\mathcal{R} + D(\rho_{\tilde{\gamma}}).
\end{eqnarray}
{F}rom the OTF-equation \eqref{OTFequation} and Lemma~\ref{mu=0} we get
\begin{equation*}
\int_{|\bx| \leq \frac{r}{1-\lambda'}} \rho_{\gamma}(\bx) \; d\bx \leq \int_{|\bx| \leq \frac{2-\lambda'}{1-\lambda'}r} \rho_{r}^{\rm OTF} (\bx) \; d\bx \leq C r^{-3}.
\end{equation*}
While since $V_{r}(\by) \leq Cr^{-4}$ (Lemma~\ref{est0ext}) and is non-zero only for $|\by|>r$
\begin{equation*}
 \int_{|\bx|\leq \frac{r}{1-\lambda'}} [V_{r}(\bx)]_{+}^{\frac52} \; d\bx \leq C r^{-7} \tfrac{\lambda'}{(1-\lambda')^3}.
\end{equation*}
Hence, from \eqref{bici2} and \eqref{Kop1} and the inequalities above we find choosing $\lambda' = r^{\frac23}$
\begin{eqnarray*}
\mathcal{E}^{A}(\gamma_{r}^{\rm HF}) &\leq& \tfrac{3}{10} (\tfrac{6\pi^2}{q})^{\frac23} \int_{\R^3} (\rho_{r}^{\rm OTF}(\bx))^{\frac53} \; d\bx -\int_{\R^3} V_r(\bx) \rho_{\gamma}(\bx) \; d \bx  + C s^{-2} r^{-3} \\
&& + C r^{-7+\frac23} +\mathcal{R} + D(\rho_{\tilde{\gamma}}) = \dots .
\end{eqnarray*}
Here we used that $\lambda' \leq 1/2$ which follows by the bound on $D$. Since $\rho_{\tilde{\gamma}} \leq \rho_{\gamma}$, $D(\rho_{\tilde{\gamma}}) \leq D(\rho_{\gamma})$. Moreover by Newton's Theorem $D(\rho_{\gamma}) \leq D(\rho_{r}^{\rm OTF})$. Hence we get
\begin{eqnarray}\notag
\dots &\leq& \mathcal{E}^{\rm OTF}(\rho^{\rm OTF}_{r}) +\int_{\R^3} V_r(\bx) (\rho^{\rm OTF}_{r}(\bx)- \rho_{\gamma}(\bx)) \; d \bx  + C s^{-2} r^{-3} \\ \label{bici3}
&& + C r^{-7+\frac23} +\mathcal{R} .
\end{eqnarray}
We study now the second term on the right hand side of \eqref{bici3}. Since $\rho_{\gamma}=\rho^{\rm OTF}*|h_s|^2$, rewriting
\begin{equation*}
\int_{\R^3} V_r(\bx) (\rho^{\rm OTF}_{r}(\bx)- \rho_{\gamma}(\bx)) \; d \bx = \int_{\R^3} \rho^{\rm OTF}_{r}(\bx) (V_r(\bx)-V_{r}*|h_s|^2(\bx)) \; d \bx .
\end{equation*}
Since $s<r$, $V_r$ is harmonic on $|\bx|>r$ and $\rho_{r}^{\rm OTF}$  vanishes for $|\bx|<r$ one sees that the integrand on the right hand side of the equation above is non-zero only for $r < |\bx| <r+s$. Hence by Lemma~\ref{est0ext}\begin{equation*}
\int_{\R^3} V_r(\bx) (\rho^{\rm OTF}_{r}(\bx)- \rho_{\gamma}(\bx)) \; d \bx \leq \int_{r<|\bx|<r+s} \rho^{\rm OTF}_{r}(\bx) V_r(\bx)\; d \bx \leq C r^{-8}s .
\end{equation*}
Choosing $s=r^{\frac53}$ we find from \eqref{bici3} that
\begin{equation}\label{fi2}
\mathcal{E}^{A}(\gamma^{\rm HF}_{r}) \leq \mathcal{E}^{\rm OTF}(\rho^{\rm OTF}_{r}) + C r^{-7+\frac23} +\mathcal{R} .
\end{equation}
It remains to estimate $\mathcal{R}$. {F}rom Lemma~\ref{extde}, choosing $\lambda, \nu \leq 1/2$ and $D$ such that $\sigma r^{\varepsilon'}\leq 1$ we find
\begin{equation*}
(\tfrac{\pi}{2\lambda r} +\tfrac{C}{\lambda^2r^2}) \int_{|\bx| \geq r(1-\lambda)(1-\nu)} \rho^{\rm HF}(\bx) \; d\bx \leq C r^{-5} \lambda^{-2}.
\end{equation*}
By Lemma~\ref{est0ext}, \eqref{Vr} and since $\lambda \leq 1/2$ we get
\begin{equation*}
\int_{r(1-\lambda)\leq |\bx| \leq\frac{r}{1-\lambda}} (\Phi_{r(1-\lambda)}^{\rm HF}(\bx))^{\frac52} \; d\bx \leq C r^{-7} \lambda,
\end{equation*}
and similarly
\begin{equation*}
\alpha^3 \int_{r(1-\lambda)\leq |\bx| \leq\frac{r}{1-\lambda}} (\Phi_{r(1-\lambda)}^{\rm HF}(\bx))^{4} \; d\bx \leq C r^{-4} \lambda,
\end{equation*}
since $r \geq \beta_0 Z^{-\frac13}$ implies $\alpha r^{-3} \leq \beta_{0}^{-3} \kappa$. Hence from the expression of $\mathcal{R}$ and the boundness of $t^p e^{-t}$ for $t>0$, we find
\begin{equation}\label{bici4}
\mathcal{R} \leq \mathcal{E}x(\gamma^{\rm HF}_r) + C r^{-5} \lambda^{-2} + C r^{-7}\lambda.
\end{equation}
We estimate now the exchange term. By the exchange inequality (\cite{LO} or \cite[Th.6.4]{Sol}) and proceeding as in \eqref{pri4} we find by Lemma~\ref{extde} and Lemma~\ref{extGest}
\begin{eqnarray*}
\mathcal{E}x(\gamma^{\rm HF}_{r}) & \leq & C \int_{\R^3} G_{\alpha}(\rho^{\rm HF}_r(\bx)) d\bx + Cr^{-\frac32} \Big(\alpha^{-1} \int_{\R^3} G_{\alpha}(\rho_{r}^{\rm HF}(\bx)) d\bx \Big)^{\frac12} \\
& \leq &  C \alpha \mathcal{R} + C \alpha r^{-7} + Cr^{-\frac32} ( \mathcal{R}+r^{-7})^{\frac12} .
\end{eqnarray*}
Hence choosing $\alpha_0$ such that $1-C\alpha \geq 1/2$ for all $\alpha \leq \alpha_0$ we get from the inequality above and \eqref{bici4}
\begin{equation*}
\tfrac12 \mathcal{R}\leq Cr^{-\frac32} ( \mathcal{R}+r^{-7})^{\frac12} + C r^{-5} \lambda^{-2} + C r^{-7}\lambda \, ,
\end{equation*}
that gives
\begin{equation}\label{fi3}
\mathcal{R} \leq C(r^{-5} \lambda^{-2} + \lambda r^{-7}) \, .
\end{equation}
The second two inequalities in \eqref{rrcx} follow from the estimate above and  lemmas~\ref{extde} and \ref{extGest} choosing $\lambda =1/2$ and replacing $r$ with $r/2$.

\textit{Step II. Estimate from below on $\mathcal{E}^{A}(\gamma^{\rm HF}_{r})$.} Adding and subtracting $D(\rho^{\rm OTF}_{r})$ and $\Tr[\rho^{\rm OTF}_{r}*\frac{1}{|\cdot|} \gamma_{r}^{\rm HF}]$ we write
\begin{equation}\label{fi1}
\mathcal{E}^{A}(\gamma_{r}^{\rm HF})=\Tr[(\alpha^{-1}T(\bp)-\varphi_{r}^{\rm OTF})\gamma^{\rm HF}_{r}]+\| \rho^{\rm OTF}_{r} -\rho^{\rm HF}_{r}\|_{C}^2 -D(\rho_{r}^{\rm OTF}),
\end{equation}
using that $V_r = \Phi^{\rm HF}_{r}$ on the support of $\rho_{r}^{\rm HF}$. The first term on the right hand side of \eqref{fi1} is estimated from below by the sum of the first $N'$ eigenvalues of the operator $\alpha^{-1} T(\bp) -\varphi_{r}^{\rm OTF}$ acting on the functions with support on $\{ \bx : |\bx| \geq r\}$. Here $N'$ denotes the smallest integer bigger than $\Tr[\gamma_{r}^{\rm HF}]$. Hence by Lemma~\ref{auxi} we find for $\mu \in (0,1)$ and $s <r$
\begin{eqnarray*}
\mathcal{E}^{A}(\gamma_{r}^{\rm HF})&  \geq & -(\tfrac{2}{1-\mu})^{\frac32} \tfrac{q}{15 \pi^2} \int_{\R^3} [\varphi_{r}^{\rm OTF}(\bq)]_{+}^{\frac52} \; d\bq - C r^{-8} s \mu^{-\frac32}- C \mu^{-3} r^{-5} s\\
&&-C (1-\mu)^{-\frac72}r^{-5} -C (1-\mu) s^{-2} \Big(\int_{\R^3} \rho_{r}^{\rm HF}(\bx) \; d\bx +1 \Big)  \\
&& +\| \rho^{\rm OTF}_{r} -\rho^{\rm HF}_{r}\|_{C}^2 -D(\rho_{r}^{\rm OTF}) =\dots ,
\end{eqnarray*}
Notice the factor $q$ due to spin. Choosing $D$ such that $\sigma r^{\varepsilon'} \leq 1$, by lemmas~\ref{extde} and \ref{mu=0} we find
\begin{equation*}
\int_{\R^3} \rho_{r}^{\rm HF}(\bx) \; d\bx \leq C r^{-3} \mbox{ and }  \int_{\R^3} [\varphi_{r}^{\rm OTF}(\bq)]_{+}^{\frac52} \; d\bq \leq C r^{-7}.
\end{equation*}
Hence considering $\mu \leq 1/2$
\begin{eqnarray*}
\dots &  \geq & -2^{\frac32} \tfrac{q}{15 \pi^2} \int_{\R^3} [\varphi_{r}^{\rm OTF}(\bq)]_{+}^{\frac52} \; d\bq -C r^{-7} - C r^{-8} s \mu^{-\frac32}- C \mu^{-3} r^{-5} s\\
&& -C s^{-2} r^{-3}+\| \rho^{\rm OTF}_{r} -\rho^{\rm HF}_{r}\|_{C}^2 -D(\rho_{r}^{\rm OTF}) =\dots .
\end{eqnarray*}
By the OTF-equation \eqref{OTFequation} and since $\rho_{r}^{\rm OTF}$ has support where $\varphi_{r}^{\rm OTF}\geq 0$ we find
\begin{equation*}
\dots = \mathcal{E}^{\rm OTF}(\rho_{r}^{\rm OTF}) - C r^{-7+\frac13}+\| \rho^{\rm OTF}_{r} -\rho^{\rm HF}_{r}\|_{C}^2 ,
\end{equation*}
choosing $\mu=\frac12 r^{-\frac25}s^{\frac25}$ and $s=r^{\frac{11}{6}}$.

Hence combining the inequality above with \eqref{fi2} and \eqref{fi3} we find
\begin{equation}\label{fi4}
\| \rho_{r}^{\rm OTF}-\rho_{r}^{\rm HF}\|_{C}^2 \leq C r^{-7+\frac13} + C (r^{-5} \lambda^{-2} +\lambda r^{-7}).
\end{equation}
We study now $\| \chi_{r}^{+} \rho^{\rm HF} - \rho_{r}^{\rm HF}\|_{C}$. By Hardy-Littlewood-Sobolev inequality we find
\begin{equation}\label{fi5}
\| \chi_{r}^{+} \rho^{\rm HF} - \rho_{r}^{\rm HF}\|_{C} \leq  C \| \chi_{r}^{+} \rho^{\rm HF} - \rho_{r}^{\rm HF}\|_{\frac65} \leq C \Big( \int_{r \leq |\bx| \leq \frac{r}{1-\lambda}} \rho^{\rm HF}(\bx)^{\frac65} \; d\bx  \Big)^{\frac56}.
\end{equation}
To estimate the last term in \eqref{fi5} we are going to use the second estimate in \eqref{rrcx} that we have just proved. With $\Sigma$ defined as in \eqref{Sigma} we find by H\"older's inequality
\begin{eqnarray*}
\int_{r \leq |\bx| \leq \frac{r}{1-\lambda}} \rho^{\rm HF}(\bx)^{\frac65} \; d\bx & \leq & \Big( \int_{\substack{r \leq |\bx|,\\ \bx \in \Sigma}} \rho^{\rm HF}(\bx)^{\frac43} \; d\bx  \Big)^{\frac{9}{10}} \Big( \int_{r \leq |\bx| \leq \frac{r}{1-\lambda}} 1 \; d\bx  \Big)^{\frac{1}{10}}\\
&&+\Big( \int_{\substack{r \leq |\bx|,\\ \bx \in \R^3 \setminus \Sigma}} \rho^{\rm HF}(\bx)^{\frac53} \; d\bx  \Big)^{\frac{18}{25}}\Big( \int_{r \leq |\bx| \leq \frac{r}{1-\lambda}} 1 \; d\bx  \Big)^{\frac{7}{25}}\\
& \leq & C r^{-\frac{33}{10}} \lambda^{\frac{1}{10}} + C r^{-\frac{21}{5}} \lambda^{\frac{7}{25}}.
\end{eqnarray*}
{F}rom the estimate above, \eqref{fi4} and \eqref{fi5} it then follows
\begin{eqnarray*}
\| \chi_{r}^{+} \rho^{\rm HF} - \rho_{r}^{\rm OTF}\|_{C} & \leq & \| \chi_{r}^{+} \rho^{\rm HF} - \rho_{r}^{\rm HF}\|_{C} + \| \rho_{r}^{\rm HF} -\rho_{r}^{\rm OTF}\|_{C} \\
&\hspace{-.5cm} \leq & \hspace{-.5cm}  C r^{-\frac72+\frac16} + C (r^{-5} \lambda^{-2} +\lambda r^{-7})^{\frac12} + C( r^{-\frac{11}{4}} \lambda^{\frac{1}{12}} + r^{-\frac72} \lambda^{\frac{7}{30}}),
\end{eqnarray*}
that gives the claim choosing $\lambda = r^{\frac57}$
\end{proof}

\subsubsection{Estimate on $\mathcal{A}_3$}

\begin{lemma}\label{Gx}
Let $G_{\alpha}$ be the function defined in Theorem~\ref{Daube}. Let $Z \alpha= \kappa$ fixed, $0 \leq \kappa < 2/\pi$ and $Z \geq 1$.

There exists $\alpha_0>0$ such that given $\varepsilon', \sigma >0$ there exists a constant $D< 1/4$ depending only on $\varepsilon'$ and $\sigma$ such that if \eqref{inter} holds  for all $|\bx| \leq D$, then for all $\alpha \leq \alpha_0$
\begin{equation*}
\alpha^{-1} \int_{|\by| \geq |\bx|} G_{\alpha}(\rho^{\rm HF}(\by)) d \by \leq C |\bx|^{-7} \mbox{ for all }|\bx| \leq D,
\end{equation*}
with $C$ a universal positive constant.
\end{lemma}

\begin{proof}
If $|\bx| < \beta_{0} Z^{-\frac13}$ we find by Lemma~\ref{rrcou}
\begin{equation*}
\alpha^{-1} \int_{|\by| > |\bx|} G_{\alpha}(\rho^{\rm HF}(\by)) d \by \leq \alpha^{-1} \int_{\R^3} G_{\alpha}(\rho^{\rm HF}(\by)) d \by \leq C Z^{\frac73} \leq C |\bx|^{-7}.
\end{equation*}
While if $D \geq |\bx| \geq \beta_0 Z^{-\frac13}$ the claim follows from the second estimate in \eqref{rrcx}.
\end{proof}

\begin{lemma}\label{A3}
Let $Z \alpha= \kappa$ fixed, $0 \leq \kappa < 2/\pi$, $Z \geq 1$ and $0< \mu < \frac{1}{109}$.

 There exists $\alpha_0$ such that given $\varepsilon', \sigma >0$ there exists a constant $D< 1/4$ depending only on $\varepsilon'$ and $\sigma$ such that  for all $\alpha \leq \alpha_0$ and for all $r$ with $\beta_0 Z^{-\frac{1-\mu}{3}} \leq r \leq D$ for which \eqref{inter} holds for $|\bx|\leq r$, then for all $\bx$ with $|\bx| \geq r$
\begin{equation*}
| \mathcal{A}_{3}(r, \bx)| \leq C  \Big(\frac{|\bx|}{r}\Big)^{\frac{1}{12}} r^{-4+\frac{3 \mu}{1-\mu}} ,
\end{equation*}
with $C>0$ a universal constant.
\end{lemma}

\begin{proof}
We proceed similarly as in Theorem~\ref{Lmnu}. By the formula for $\mathcal{A}_{3}$, Proposition~\ref{Counorm} and Lemma~\ref{extrCou} we get
\begin{equation}\label{fr1}
| \mathcal{A}_{3}(r, \bx)| \leq \int_{A(|\bx|,k)} \chi_{r}^{+}(\by) \frac{|\rho_{r}^{\rm OTF}(\by)-\rho^{\rm HF}(\by)|}{|\bx -\by|} \; d\by + C k^{-1} |\bx|^{-\frac12} r^{-\frac72+\frac16} .
\end{equation}
By H\"older's inequality, Lemma~\ref{mu=0}, the OTF-equation \eqref{OTFequation} and \eqref{52ann} we find
\begin{equation}\label{fr2}
 \int_{A(|\bx|,k)} \frac{\rho_{r}^{\rm OTF}(\by)}{|\bx -\by|} \; d\by \leq C r^{-\frac{21}{5}} |\bx|^{\frac15} k^{\frac15}.
\end{equation}
Once again, to estimate $\int_{A(|\bx|,k)} \frac{\chi_{r}^{+}(\by)\rho^{\rm HF}(\by)}{|\bx -\by|} \; d\by $ we have to proceed differently than in \cite[Lem.12.7]{Sol} since $\rho^{\rm HF}$ is not in $L^{\frac53}( \R^3)$. We consider the following splitting
\begin{equation}\label{sol1}
\int_{A(|\bx|,k)} \chi_{r}^{+}(\by) \frac{\rho^{\rm HF}(\by)}{|\bx -\by|} \; d\by =\int_{\substack{A(|\bx|,k)\\|\bx-\by|>R,|\by|>r}} \frac{\rho^{\rm HF}(\by)}{|\bx -\by|} \; d\by + \int_{\substack{|\by|>r,\\|\bx-\by|<R}} \frac{\rho^{\rm HF}(\by)}{|\bx -\by|} \; d\by,
\end{equation}
for $R>0$ to be chosen. By H\"older's inequality, Theorem~\ref{Daube}, Remark~\ref{Daure}, \eqref{52ann} and Lemma~\ref{extrCou} we get
\begin{equation}\label{hint}
\int_{\substack{A(|\bx|,k)\\|\bx-\by|>R,|\by|>r}} \frac{\rho^{\rm HF}(\by)}{|\bx -\by|} \; d\by \leq C \alpha^{\frac34} R^{-\frac38} |\bx|^{\frac18} k^{\frac18} r^{-\frac{21}{4}} + C r^{-\frac{21}{5}} |\bx|^{\frac15} k^{\frac15}.
\end{equation}
It remains to study the second term on the right hand side of \eqref{sol1}. Let $\nu \in \R^{+}$ be such that $\nu \alpha \leq 2/\pi$. We consider the density matrix $\gamma^{\rm HF}_{r/2}$ defined in \eqref{exteriordensity} with $\lambda = 1/2$. {F}rom Theorem~\ref{dayau} it follows that for $\bx$ such that $|\bx|\geq r$
\begin{equation*}
\Tr[(\alpha^{-1} T(\bp) -\frac{\nu}{|\cdot-\bx|} \chi_{B_{R}(\bx)}(\cdot)) \gamma^{\rm HF}_{r/2}] \geq - C (\nu^{\frac52} R^{\frac12} + \nu^{4} \alpha^{2}).
\end{equation*}
Hence we find
\begin{eqnarray*}
\nu \int_{|\by -\bx|<R} \chi_{r}^{+}(\by) \frac{\rho^{\rm HF}(\by)}{|\bx -\by|} d \by & \leq & \nu \int_{|\by -\bx|<R} \frac{\rho_{r/2}^{\rm HF}(\by)}{|\bx -\by|} d \by \\
& \leq & \Tr[\alpha^{-1} T(\bp) \gamma^{\rm HF}_{r/2}] + C (\nu^{\frac52} R^{\frac12} + \nu^{4} \alpha^{2})
\end{eqnarray*}
 and by Lemma~\ref{extrCou}
\begin{equation}\label{hfex}
\int_{|\by -\bx|<R} \chi_{r}^{+}(\by) \frac{\rho^{\rm HF}(\by)}{|\bx -\by|} d \by \leq C \nu^{-1} r^{-7}+ C(\nu^{\frac32} R^{\frac12}+ \nu^{3} \alpha^{2}).
\end{equation}
Hence from \eqref{fr1}, \eqref{fr2}, \eqref{hint} and \eqref{hfex} it follows that
\begin{eqnarray*}
| \mathcal{A}_{3}(r, \bx)| & \leq & C \nu^{-1} r^{-7}+ C(\nu^{\frac32} R^{\frac12}+ \nu^{3} \alpha^{2})+ C \alpha^{\frac34} R^{-\frac38} |\bx|^{\frac18} k^{\frac18} r^{-\frac{21}{4}}\\
&& + C r^{-\frac{21}{5}} |\bx|^{\frac15} k^{\frac15} + C k^{-1} |\bx|^{-\frac12} r^{-\frac72+\frac16} .
\end{eqnarray*}
So choosing $\nu = 1/2 (\beta_{0}r^{-1})^{\frac{3}{1-\mu}}$ (that gives $\nu \alpha <2/\pi$), $k$ such that $r^{-\frac{21}{5}} |\bx|^{\frac15} k^{\frac15}= k^{-1} |\bx|^{-\frac12} r^{-\frac72+\frac16}$, i.e. $k=|\bx|^{-\frac{7}{12}} r^{\frac{13}{18}}$ and $R$ such that $\alpha^{\frac34} R^{-\frac38} |\bx|^{\frac18 \frac{5}{12}} r^{-\frac{21}{4}+\frac18\frac{13}{18}} = r^{-4-\frac{1}{18}} |\bx|^{\frac{1}{12}}$, i.e. $R=\alpha^2 |\bx|^{-\frac{1}{12}}r^{-\frac{5}{18}}$
\begin{equation*}
| \mathcal{A}_{3}(r, \bx)| \leq C (r^{-4+\frac{3 \mu}{1-\mu}}+|\bx|^{-\frac{1}{24}} r^{-\frac{5}{36}-\frac{9}{2(1-\mu)}} \alpha+ r^{-\frac{9}{1-\mu}} \alpha^{2}+ |\bx|^{\frac{1}{12}} r^{-4-\frac{1}{18}}) .
\end{equation*}
Finally since $r^{-1}\alpha^{\frac{1-\mu}{3}} \leq \beta_{0}^{-1}\kappa^{\frac{1-\mu}{3}}$, the claim follows for $|\bx| \geq r$ and $\mu <1/(109)$.
\end{proof}

\subsection{The intermediate region} Here we prove the main estimate in Theorem~\ref{mainesti} up to a fixed distance independent of $Z$.

\begin{lemma}[Iterative step]\label{Iterstep}

Let $Z \alpha =\kappa$ fixed with $0 \leq \kappa <2/\pi$. Consider $\mu=\frac{1}{11}\frac{1}{49}$ and assume $N \geq Z \geq 1$.

Then there exists $\alpha_{0}>0$ such that for all $\delta, \varepsilon', \sigma>0$ with $\delta < \delta_{0}$, where $\delta_{0}$ is some universal constant, there exists constants $\varepsilon_{2}, C_{\phi}'>0$ depending only on $\delta$ and a constant $D=D(\varepsilon', \sigma)>0$ depending only on $\varepsilon', \sigma$ with the following property. For all $\alpha \leq \alpha_{0}$ and $R_{0}<D$ satisfying that $\beta_{0}Z^{-\frac{1-\mu}{3}} \leq R_{0}^{1+\delta}$ and that \eqref{inter}
holds for all $|\bx| \leq R_{0}$, there exists $R_{0}'>R_{0}$ such that
\begin{equation*}
|\Phi^{\rm HF}_{|\bx|}(\bx)-\Phi^{\rm TF}_{|\bx|}(\bx)| \leq C_{\Phi}' |\bx|^{-4+\varepsilon_{2}}
\end{equation*}
for all $\bx$ with $R_{0} < |\bx| < R_{0}'$.
\end{lemma}
\begin{proof}
Let $D>0$ depending on $\sigma, \varepsilon'$ be the smaller of the values of $D$ occurring in Lemma~\ref{A1A2} and Lemma~\ref{A3}.
Given $\delta>0$. We consider $R_{0}<D$ satisfying $\beta_{0}Z^{-\frac{1-\mu}{3}}\leq R_{0}^{1+\delta}$ and such that \eqref{inter} holds for all $|\bx| \leq R_{0}$.

Set $R_{0}'=R_{0}^{1-\delta}$ and $r=R_{0}^{1+\delta}$. Then we have $\beta_{0}Z^{-\frac13} \leq \beta_{0}Z^{-\frac{1-\mu}{3}} \leq r \leq R_{0}<D$ we can therefore apply Lemma~\ref{A1A2} and Lemma~\ref{A3}. {F}rom \eqref{A123} we obtain that for all $|\bx| \geq r$ and all $\alpha \leq \alpha_{0}$
\begin{equation*}
|\Phi^{\rm HF}_{|\bx|}(\bx)-\Phi^{\rm TF}_{|\bx|}(\bx)| \leq  C |\bx|^{-4-\zeta} r^{\zeta} +  C  \Big(\frac{|\bx|}{r}\Big)^{\frac{1}{12}} r^{-4+\frac{3 \mu}{1-\mu}}.
\end{equation*}
Since for $R_{0}<|\bx|<R_{0}'$ we have
\begin{equation*}
|\bx|^{\frac{2\delta}{1-\delta}}\leq \frac{r}{|\bx|} \leq |\bx|^{\delta}
\end{equation*}
and thus
\begin{equation*}
|\Phi^{\rm HF}_{|\bx|}(\bx)-\Phi^{\rm TF}_{|\bx|}(\bx)| \leq  C |\bx|^{-4+\delta \zeta} + C  |\bx|^{-4+3\frac{ \mu}{1-\mu}} |\bx|^{-\frac{\delta}{1-\delta}(8+\frac16-\frac{6 \mu}{1-\mu})}.
\end{equation*}
Hence choosing $\delta_0$ sufficiently small there are $C_{\Phi}'$ and $\varepsilon_{2}$ such that the claim holds.
\end{proof}

\begin{lemma}\label{<D}
Let $Z \alpha =\kappa$ fixed with $0 \leq \kappa< 2/\pi$. Assume $N \geq Z \geq 1$.

Then there exist universal constants $\alpha_{0}$, $\varepsilon \in (0,4)$ and $D, C_{\Phi}>0$, $D <1/4$, such that for all $\alpha \leq \alpha_{0}$ and $\bx$ with $|\bx| \leq D$ we have\begin{equation*}
|\Phi^{\rm HF}_{|\bx|}(\bx)-\Phi^{\rm TF}_{|\bx|}(\bx)| \leq  C_{\Phi} |\bx|^{-4+\varepsilon}.
\end{equation*}
\end{lemma}
\begin{proof}
We fix $\mu=\frac{1}{11} \frac{1}{49}$ as in Lemma~\ref{Iterstep}. Since $\mu < \frac{2}{11}\frac{1}{49}$, by Theorem~\ref{Lmnu} we know that there exists constants $a,b,c >0$ such that for all $|\bx| \leq \beta Z^{-\frac{1-\mu}{3}}$
\begin{equation}\label{u1}
|\Phi^{\rm HF}_{|\bx|}(\bx)-\Phi^{\rm TF}_{|\bx|}(\bx)| \leq  C (1+\beta^{2}+\beta^{5/2}+\beta^{b} |\bx|^{c}) \beta^{2-a} |\bx|^{-4+a}.
\end{equation}
We first show that we may choose $\delta$ small enough such that if we choose $\tilde{R}^{1+\delta}=\beta_{0} Z^{-\frac{1-\mu}{3}}$ we have for all $|\bx|<\tilde{R}$ that
\begin{equation}\label{b3}
|\Phi^{\rm HF}_{|\bx|}(\bx)-\Phi^{\rm TF}_{|\bx|}(\bx)| \leq  C_{\Phi}'' |\bx|^{-4+\frac{a}{2}}.
\end{equation}
Let $\beta>0$ be such that $(\beta Z^{-\frac{1-\mu}{3}})^{1+\delta} =\beta_{0}Z^{-\frac{1-\mu}{3}}$, i.e. $\beta^{1+\delta}=\beta_{0} Z^{\delta\frac{1-\mu}{3}}$. Hence from \eqref{u1} we find for all $|\bx| \leq \beta Z^{-\frac{1-\mu}{3}}$
\begin{equation*}
|\Phi^{\rm HF}_{|\bx|}(\bx)-\Phi^{\rm TF}_{|\bx|}(\bx)| \leq C (1+\beta^2+\beta^{5/2}+\beta^{b} |\bx|^{c}) \beta^{2-\frac{a}{2}} Z^{-\frac{a}{2}\frac{1-\mu}{3}} |\bx|^{-4+\frac{a}{2}} \, ,
\end{equation*}
and by the choice of $\beta$ (and $\beta_{0} < 1$)
\begin{eqnarray*}
|\Phi^{\rm HF}_{|\bx|}(\bx)-\Phi^{\rm TF}_{|\bx|}(\bx)| &  \leq & C (1+Z^{2\frac{\delta}{1+\delta}\frac{1-\mu}{3}} + Z^{\frac52\frac{\delta}{1+\delta}\frac{1-\mu}{3}}+ Z^{\frac{\delta}{1+\delta}\frac{1-\mu}{3}(b+c)} Z^{-c\frac{1-\mu}{3}}) \\
&& \hspace{2cm} Z^{(2-\frac{a}{2})\frac{1-\mu}{3}\frac{\delta}{1+\delta}} Z^{-\frac{a}{2}\frac{1-\mu}{3}} |\bx|^{-4+\frac{a}{2}}.
\end{eqnarray*}
Hence if $\delta$ is small enough we may choose a universal constant $C_{\Phi}''$ such that \eqref{b3} holds.

Let now $\delta$ be small enough so that we may apply Lemma~\ref{Iterstep}. This give constant $\varepsilon_{2}$ and $C_{\Phi}'$ (depending only on $\delta$) and for all $\sigma, \varepsilon'>0$ a constant $D < 1/4$. Now choose $\sigma= \max\{C_{\Phi}', C_{\Phi}''\}$ and $\varepsilon'=\min\{a/2, \varepsilon_{2}\}$. Now $\sigma, \varepsilon'$ and $D$ are universal constants. To prove the claim we shall prove that for all $|\bx| \leq D$
\begin{equation}\label{b4}
|\Phi^{\rm HF}_{|\bx|}(\bx)-\Phi^{\rm TF}_{|\bx|}(\bx)| \leq \sigma |\bx|^{-4+\varepsilon'}.
\end{equation}
We have to prove that $D$ belongs to the set
\begin{equation*}
\mathcal{M}= \{ 0 <R \leq 1/4: \mbox{ Inequality }\eqref{b4} \mbox{ holds for all }|\bx | \leq R\}.
\end{equation*}
We reason by contradiction. If this was not true then $D>R_{0} =\sup \mathcal{M}$ and in particular $R_{0}< 1/4$. {F}rom \eqref{b3} and the choice of $\sigma$ and $\varepsilon'$ it follows that either $\tilde{R}>1/4$ or $\tilde{R} \in \mathcal{M}$. In the first case then $R_{0}=\sup \mathcal{M}=1/4>D$ that contradicts our hypothesis. On the other hand if $\tilde{R} \in \mathcal{M}$, then $R_{0}^{1+\delta}\geq \tilde{R}^{1+\delta}= \beta_{0}Z^{-\frac{1-\mu}{3}}$. It then follows from Lemma~\ref{Iterstep} that there exists $R_{0}' \in \mathcal{M}$ with $R_{0}'>R_{0}$. This contradicts also our hypothesis.
\end{proof}

\subsection{The outer zone and proof of Theorem~\ref{mainesti}}

The proof of Theorem~\ref{mainesti} follows directly from Lemma~\ref{<D} and the following result.

\begin{lemma}\label{>D}
Let $Z \alpha =\kappa$, $0 \leq \kappa <2/\pi$. Assume $N \geq Z \geq 1$. Let $D,\varepsilon$ and $C_{\Phi}$ be the constants introduced in Lemma~\ref{<D}.

Then there exist $\alpha_{0}>0$ and a universal constant $C_{M}>0$ such that for all $\alpha \leq \alpha_{0}$ and $\bx$ with $|\bx| \geq D$ we have
\begin{equation*}
|\Phi^{\rm HF}_{|\bx|}(\bx)-\Phi^{\rm TF}_{|\bx|}(\bx)| \leq C_{M}.
\end{equation*}
\end{lemma}
\begin{proof} Here $C_{i}$, $i=1, \dots, 6$ denote positive universal constants. We write
\begin{equation}\label{or1}
|\Phi^{\rm HF}_{|\bx|}(\bx)-\Phi^{\rm TF}_{|\bx|}(\bx)| \leq |\Phi^{\rm HF}_{D}(\bx)-\Phi^{\rm TF}_{D}(\bx)|+ \int_{D<|\by|<|\bx|} \frac{\rho^{\rm TF}(\by)+\rho^{\rm HF}(\by)}{|\bx-\by|} \, d\by .
\end{equation}
Since $\Phi_{D}^{\rm HF}(\bx)-\Phi^{\rm TF}_{D}(\bx)$ is harmonic for $|\bx|>D$ and tends to zero at infinity we have by Lemma~\ref{<D}
\begin{equation}\label{or2}
 |\Phi^{\rm HF}_{D}(\bx)-\Phi^{\rm TF}_{D}(\bx)| \leq \sup_{|\bx|=D}  |\Phi^{\rm HF}_{D}(\bx)-\Phi^{\rm TF}_{D}(\bx)| \leq C_{\phi} D^{-4+\varepsilon}.
\end{equation}
For the second term on the right hand side of \eqref{or1} we write
\begin{eqnarray}\notag
&& \int_{D<|\by|<|\bx|} \frac{\rho^{\rm TF}(\by)+\rho^{\rm HF}(\by)}{|\bx-\by|} \, d\by\\ \label{o3}
& \leq &  \int_{\substack{|\bx-\by|<D/4\\|\by|>D}} \frac{\rho^{\rm TF}(\by)+\rho^{\rm HF}(\by)}{|\bx-\by|} \, d\by+ \frac{4}{D}\int_{D<|\by|} (\rho^{\rm TF}(\by)+\rho^{\rm HF}(\by)) \, d\by.
\end{eqnarray}
By Lemma~\ref{extde}, Lemma~\ref{<D}, estimate \eqref{Tfpou} and the TF-equation we find
\begin{equation}\label{or4}
\int_{D<|\by|} (\rho^{\rm TF}(\by)+\rho^{\rm HF}(\by)) \, d\by \leq C_{1}(1+C_{\Phi}D^{\varepsilon}) (1+D^{-3})+ C_{1} D^{-3}.
\end{equation}
It remains to estimate the first term on the right hand side of \eqref{o3}. By H\"older's inequality, estimate \eqref{Tfpou} and the TF-equation we get
\begin{equation}\label{or5}
\int_{\substack{|\bx-\by|<D/4\\|\by|>D}} \frac{\rho^{\rm TF}(\by)}{|\bx-\by|} \, d\by \leq C_{2} \Big(\int_{|\by|>D} (\rho^{\rm TF}(\by))^{\frac53} \, d\by \Big)^{\frac35} D^{\frac15} \leq C_{3} D^{-4}.
\end{equation}
To estimate the term with the HF-density we use Theorem~\ref{dayau}. Let $\gamma^{\rm HF}_{D}$ be the exterior HF-density matrix as defined in \eqref{exteriordensity} with $r=D/2$ and $\lambda=1/2$. Then by Theorem~\ref{dayau} with $\nu=\beta_{0}^3 D^{-3}$
\begin{equation*}
\alpha^{-1}\Tr[(T(\bp)-\frac{\nu \alpha}{|\bx -\cdot|} \chi_{B_{\frac{D}{4}}(\bx)} (\cdot)) \gamma^{\rm HF}_{D/2}] \geq -C_{4} (D^{\frac12} \nu^{\frac52}+\nu^{4} \alpha^2),
\end{equation*}
and thus
\begin{equation*}
\int_{|\bx-\by|<D/4} \frac{\rho_{D/2}^{\rm HF}(\by)}{|\bx-\by|} \, d\by \leq  C_{5} D^{3} \alpha^{-1} \Tr[T(\bp) \gamma^{\rm HF}_{D/2}] +C_{6} D^{-4},
\end{equation*}
Here we use that $D>2\beta_{0} Z^{-\frac13}$ (for $\alpha \leq \alpha_0$) and $D<1/4$. By Lemma~\ref{extrCou} we conclude
\begin{equation}\label{or6}
\int_{|\bx-\by|<D/4} \chi_{D}^{+}(\by) \frac{\rho^{\rm HF}(\by)}{|\bx-\by|} \, d\by \leq \int_{|\bx-\by|<D/4} \frac{\rho_{D/2}^{\rm HF}(\by)}{|\bx-\by|} \, d\by \leq C_{7}D^{-4}.
\end{equation}
The claim follows collecting together formula \eqref{or1} to formula \eqref{or6}.
\end{proof}

\section{Proofs of Theorems~\ref{charge}, \ref{radius}, \ref{Ionenergy} and \ref{thpot}}\label{secproofs}

In this section we always assume the following: $Z \alpha = \kappa$ with $0 \leq \kappa <2/\pi$ and $N \geq Z \geq 1$.

\begin{proof}[Proof of Theorem~\ref{charge}]
Assume that a HF-minimizer exists with $\int \rho^{\rm HF}=N$. Let $\rho^{\rm TF}$ be the minimizer of the TF-energy functional of the neutral atom with nuclear charge $Z$. Then for $R>0$ to be chosen
\begin{equation}\label{ref0}
N = \int_{|\bx|<R} \rho^{\rm TF}(\bx) \; d\bx + \int_{|\bx|<R} (\rho^{\rm HF}(\bx) -\rho^{\rm TF}(\bx)) \; d\bx + \int_{|\bx|>R} \rho^{\rm HF}(\bx) \; d\bx .
\end{equation}
By Theorem~\ref{mainesti} we know that there exist universal positive constants $\varepsilon, \alpha_{0}, C_{M}$ and $C_{\Phi}$ such that for all $\alpha \leq \alpha_{0}$ and $\bx \in \R^3$
\begin{equation}\label{ref}
|\Phi^{\rm HF}_{|\bx|}(\bx)-\Phi^{\rm TF}_{|\bx|}(\bx)| \leq C_{\Phi} |\bx|^{-4+\varepsilon} + C_{M} .
\end{equation}
Let $Z_{0}$ be such that $Z_{0}\alpha_{0}=\kappa$. Then $\alpha \leq \alpha_{0}$ corresponds to $Z \geq Z_{0}$.  Let us choose $R$ such that $C_{\Phi} R^{-4+\varepsilon} = C_{M}$. Then from \eqref{ref0}, \eqref{ref} and Lemma~\ref{extde} for all $Z \geq Z_{0}$ we find
\begin{equation*}
N \leq \int_{|\bx|<R} \rho^{\rm TF}(\bx) \; d\bx + 2 C_{\Phi} R^{-3+\varepsilon}+C  (1+C_{\Phi} R^{\varepsilon}) (R^{-3}+1) < Z+ \tilde{Q}.
\end{equation*}
The claim follows choosing $Q = \max \{\tilde{Q},Z_{0}+1\}$.
\end{proof}

\begin{proof}[Proof of Theorem~\ref{radius}]
Let $\rho^{\rm HF}$ be the density of the HF-minimizer in the neutral case $N=Z$. We have
\begin{eqnarray*}
\Big| \int_{|\bx|>R} (\rho^{\rm HF}(\bx)-\rho^{\rm TF}(\bx)) d\bx\Big| & = & \Big| \int_{|\bx|<R} (\rho^{\rm HF}(\bx)-\rho^{\rm TF}(\bx)) d\bx\Big|\\
& = & \Big| \frac{R}{4\pi} \int_{S^2} d \omega (\Phi^{\rm HF}_{R}(R \omega)-\Phi_{R}^{\rm TF}(R \omega)) \Big|\\
& \leq & C_{\Phi} R^{-3+\varepsilon}+C_{M} R ,
\end{eqnarray*}
where in the last step we have used Theorem~\ref{mainesti}. Notice that for $Z$ sufficiently big $\alpha \leq \alpha_0$ where $\alpha_0$ is the constant given in Theorem~\ref{mainesti}. By the TF-equation, Theorem~\ref{Slb} we then find
\begin{equation*}
3^4\frac{2 \pi^2}{q^2} R^{-3}- C_{\Phi}R^{-3+\varepsilon}-C_{M}R \leq \int_{|\bx|>R} \rho^{\rm HF}(\bx) d\bx \leq 3^4\frac{2 \pi^2}{q^2} R^{-3}+ C_{\Phi}R^{-3+\varepsilon}+C_{M}R ,
\end{equation*}
from which the claim follows directly by the definition of HF-radius.
\end{proof}

\begin{proof}[Proof of Theorem~\ref{Ionenergy}] Since $E^{\rm HF}(Z-1,Z) \geq E^{\rm HF}(Z,Z)$ the ionization energy is bounded from below by zero. If $Z$ is smaller than a universal constant then we can also bound the ionization energy with a universal constant using Theorem~\ref{HFen}.

It remains to estimate from above the ionization energy when $Z$ is larger than a universal constant. We first construct a density matrix $\gamma$ such that $\Tr[\gamma] \leq Z-1$. Let $\theta_{-}:= (1-\theta_{r(1-\lambda)}^2)^{\frac12}$ for $r, \lambda$ positive parameters and $\theta_{r}$ defined in Definition~\ref{defloc}. We consider the density matrix $\gamma^{\rm HF}_{-}:= \theta_{-} \gamma^{\rm HF} \theta_{-}$ where $\gamma^{\rm HF}$ is the HF-minimizer in the neutral case. By an opportune choice of $r$ we will then have $\Tr[\gamma^{\rm HF}_{-}] \leq Z-1$. Indeed,
\begin{equation*}
\Tr[\gamma^{\rm HF}_{-}]= \int_{\R^3} \rho^{\rm HF}(\bx) \, d\bx -\int_{\R^3} \theta_{r(1-\lambda)}^2(\bx) \rho^{\rm HF}(\bx) \, d\bx \leq Z -\int_{|\bx|>r} \rho^{\rm HF}(\bx) \, d\bx .
\end{equation*}
We now choose $\lambda=\frac12$. Let $R>0$ be such that $C_{M}=C_{\Phi} R^{-4+\varepsilon}$ where $C_{M}, C_{\Phi}, \varepsilon$ are the constants in Theorem~\ref{mainesti}. Then $R$ is a universal constant. We consider $Z$ large enough so that $\beta_{0} Z^{-\frac13}<R$ where $\beta_{0}$ is the constant in Theorem~\ref{Slb}. This gives that $Z$ has to be larger than some universal constant. For $r$ such that $\beta_{0} Z^{-\frac13}<r<R$ by Theorem~\ref{mainesti} we find
\begin{equation*}
|\Phi^{\rm HF}_{|\bx|}(\bx)-\Phi^{\rm TF}_{|\bx|}(\bx)| \leq 2 C_{\Phi} |\bx|^{-4+\varepsilon} \mbox{ for all }|\bx| \leq r.
\end{equation*}
Since $\int \rho^{\rm TF}= \int \rho^{\rm HF}$, by the choice of $r$ and Lemma~\ref{extde} we get
\begin{eqnarray}\notag
\int_{|\bx|>r} \rho^{\rm HF}(\bx) \, d\bx & = & \int_{|\bx|>r} \rho^{\rm TF}(\bx) \, d\bx + \int_{|\bx|<r} (\rho^{\rm TF}(\bx)-\rho^{\rm HF}(\bx)) \, d\bx \\  \label{Cop1}
&\geq & \int_{|\bx|>r} \rho^{\rm TF}(\bx) \, d\bx -2 C_{\Phi} r^{-3+\varepsilon} \geq C r^{-3}-2 C_{\Phi} r^{-3+\varepsilon} .
\end{eqnarray}
In the last step we used the TF-equation, Corollary~\ref{TFsotto} and that $r> \beta_{0} Z^{-\frac13}$. Finally, it follows from \eqref{Cop1} by choosing $r$ sufficiently small that $\int_{|\bx|>r} \rho^{\rm HF} >1$ and hence that $\Tr[\gamma^{\rm HF}_{-}] \leq Z-1$. We may choose $r$ sufficiently small by taking $Z$ large enough. Notice that $r$ can be choosen universally and so $Z$ has to be larger than some universal constant.

By the last estimate in the proof of Theorem~\ref{EA} we find
\begin{equation*}
\mathcal{E}^{\rm HF}(\gamma^{\rm HF}_{-}) \leq \mathcal{E}^{\rm HF}(\gamma^{\rm HF}) -\mathcal{E}^{A}(\gamma^{\rm HF}_{r}) + \mathcal{R},
\end{equation*}
with $\mathcal{R}$ and $\gamma^{\rm HF}_{r}$ as defined in the statement of Theorem~\ref{EA}. Since $\mathcal{E}^{\rm HF}(\gamma^{\rm HF}_{-}) \geq E^{\rm HF}(Z-1,Z)$ and $\mathcal{E}^{\rm HF}(\gamma^{\rm HF}) = E^{\rm HF}(Z,Z)$ it remains to prove that $-\mathcal{E}^{A}(\gamma^{\rm HF}_{r}) + \mathcal{R}$ is bounded from above by some universal constant. Here we use repeteadly that $r$ is a universal constant. By estimate \eqref{fi3} we see that $\mathcal{R} \leq C r^{-7}$ a universal constant. To estimate from below $\mathcal{E}^{A}(\gamma^{\rm HF}_{r})$ we first leave out the kinetic energy term and the direct term since these are positive. Moreover, since $\Phi^{\rm HF}_{r}$ is harmonic for $|\bx|>r$ and tends to zero at infinity we see that
\begin{equation*}
\Phi^{\rm HF}_{r}(\bx) \leq \frac{r}{|\bx|} \sup_{|\by|=r} \Phi^{\rm HF}_{r}(\by) \leq \frac{r}{|\bx|} \sup_{|\by|=r} \Phi^{\rm TF}_{r}(\by) + \frac{r}{|\bx|} \sup_{|\by|=r} | \Phi^{\rm TF}_{r}(\by)- \Phi^{\rm HF}_{r}(\by) |,
\end{equation*}
which is bounded by $C'/|\bx|$, $C'$ a universal constant, by Theorem~\ref{mainesti} and Corollary~\ref{TFscr}. It then follows that
\begin{equation*}
\mathcal{E}^{A}(\gamma^{\rm HF}_{r}) \geq - \Tr[\frac{C'}{|\cdot|}\gamma^{\rm HF}_{r}]  \geq - \frac{C'}{r}\int_{|\bx|>r} \rho^{\rm HF}(\bx) \, d\bx ,
\end{equation*}
that is bounded from below by a universal constant using Lemma~\ref{extde}.
\end{proof}

\begin{proof}[Proof of Theorem~\ref{thpot}]
Let $\alpha_0$ be the constant appearing in Theorem~\ref{mainesti} and $Z_0$ be such that $\alpha_0 Z_0=\kappa$.
The claim follows directly for $Z \leq Z_0$ since both functions are bounded for $|\bx|$ large, while for $|\bx|$ small the functions are bounded by a constant times $|\bx|^{-1}$.

The case $Z > Z_0$ corresponds to $\alpha <\alpha_0$ and for such values of $\alpha$ we can use the result in Theorem~\ref{mainesti}. We separate the case small $\bx$, intermediate $\bx$ and large $\bx$. Once again, comparing with the proof in the non-relativistic case (\cite{Sol}) we have to do an extra splitting for small $\bx$.

By the definition of the mean field potential and Proposition~\ref{Counorm} we find
\begin{equation*}
|\varphi^{\rm TF}(\bx)-\varphi^{\rm HF}(\bx)| \leq \int_{|\bx -\by|<s} (\rho^{\rm TF}(\by)+\rho^{\rm HF}(\by)) \Big(\frac{1}{|\bx -\by|}-\frac{1}{s}\Big) + \frac{\sqrt{2}}{s^{\frac12}} \| \rho^{\rm TF}-\rho^{\rm HF}\|_{C}.
\end{equation*}
Since $\rho^{\rm TF}$ is bounded in $L^{\frac53}$-norm, we find using H\"older's inequality, Corollary~\ref{rTF53} and Lemma~\ref{rrcou} that
\begin{equation}\label{m0}
|\varphi^{\rm TF}(\bx)-\varphi^{\rm HF}(\bx)| \leq \int_{|\bx -\by|<s} \rho^{\rm HF}(\by) \Big(\frac{1}{|\bx -\by|}-\frac{1}{s}\Big) + C ( s^{\frac15} Z^{\frac75} + s^{-\frac12} Z^{1+\frac{3}{22}}).
\end{equation}
For the integral with the HF-density we need to split the region where the HF-density is bounded in $L^{\frac43}$-norm from the one where it is bounded in $L^{\frac53}$-norm. Proceeding as in the proof of Lemma~\ref{smallx} (from \eqref{f2} to \eqref{2a} replacing the integrals on $A(|\bx|,k)$ with integrals on $|\bx -\by|<s$) using the results of  Lemma~\ref{rrcou} we get with $R \in (0,s)$ to be chosen
\begin{equation}\label{m2}
\int_{|\bx -\by|<s} \rho^{\rm HF}(\by) \Big(\frac{1}{|\bx -\by|}-\frac{1}{s}\Big) \leq C ( Z^{\frac75} s^{\frac15}+ R^{-\frac14} (\alpha Z^{\frac73})^{\frac34}  + Z^{\frac43}+R^{\frac12} Z^{\frac32}).
\end{equation}
Recall that $Z\alpha=\kappa$ is fixed. Choosing $s$ such that $Z^{\frac75} s^{\frac15}=Z^{\frac43}$ (i.e. $s=Z^{-\frac13}$) and $R$ such that $R^{-\frac14}Z= R^{\frac12} Z^{\frac{3}{2}}$ (i.e $R=Z^{-\frac23}$; notice that $R<s$) we get from \eqref{m0} and \eqref{m2}
\begin{equation*}
|\varphi^{\rm TF}(\bx)-\varphi^{\rm HF}(\bx)| \leq C (Z^{\frac43}+ Z^{\frac{7}{6}}).
\end{equation*}
The claim follows from this inequality for $\bx \in \R^3$ such that $|\bx|\leq \beta_{0} Z^{-\frac{1+\gamma}{3}}$ for $\gamma>0$. We consider $\gamma < \frac{1}{263}$.

If $|\bx| \geq \beta_{0} Z^{-\frac{1+\gamma}{3}}$ then proceeding as for very small $\bx$ and as in the proof of Theorem~\ref{Lmnu} up to inequality \eqref{R2} we get for $t \in (\frac{1+\gamma}{3},\frac35)$, $l>t$ and $R<\beta_{0} Z^{-l}$
\begin{equation*}
|\varphi^{\rm TF}(\bx)-\varphi^{\rm HF}(\bx)| \leq C(s^{\frac15} Z^{\frac75} + s^{-\frac12} Z^{1+\frac{3}{22}} + R^{-\frac38} s^{\frac18} Z  + Z^{\frac12(3-t)} ).
\end{equation*}
Here we have also used that $Z \alpha$ is a constant. So choosing $s$ such that $s^{\frac15} Z^{\frac75}= Z^{\frac12(3-t)}$ (i.e. $s= Z^{\frac12-\frac52t}$), $R$ such that $R^{-\frac38} Z^{1+\frac{1}{16}-\frac{5}{16}t}=Z^{\frac12(3-t)}$ (i.e. $R= Z^{-\frac76+\frac12 t}$) and optimizing in $t$ (i.e. $t=\frac13+\frac43\frac{1}{77}$) we obtain
\begin{equation}\label{m3}
|\varphi^{\rm TF}(\bx)-\varphi^{\rm HF}(\bx)| \leq C Z^{\frac43-\frac23\frac{1}{77}}.
\end{equation}
Notice that $t>\frac{1+\gamma}{3}$, $R<s$ by the choice of $t$ and that $R$ satisfies the condition $R<\beta_{0} Z^{-l}$, $l>t$, for $Z$ sufficiently big. The claim then follows from \eqref{m3} for $\bx \in \R^3$ such that $|\bx|^{1+\delta} \leq \beta_{0} Z^{-\frac13}$ for $\delta < \frac{1}{153}$. We fix $\delta = \frac12 \frac{1}{153}$.

We turn now to study intermediate $\bx$. Let $D \leq 1$ be such that $C_{M} \leq C_{\Phi}D^{-4+\varepsilon}$ with $C_{M}, C_{\Phi}, \varepsilon$ the constants in Theorem~\ref{mainesti}. Then for all $\bx$ such that $|\bx| \leq D$
\begin{equation*}
|\Phi^{\rm HF}_{|\bx|}(\bx)-\Phi^{\rm TF}_{|\bx|}(\bx)| \leq 2 C_{\Phi} |\bx|^{-4+\varepsilon}.
\end{equation*}
Moreover we choose $D$ such that Lemma~\ref{A1A2} holds. Let $\bx$ be such that $\beta_{0} Z^{-\frac13} \leq |\bx|^{1+\delta} \leq D^{\frac{1+\delta}{1+\mu}}$ with $0< \mu \leq \delta$. We set $r=|\bx|^{1+\mu}$. Then $\beta_{0} Z^{-\frac13} \leq r \leq D$. We write $\varphi^{\rm TF}(\bx)-\varphi^{\rm HF}(\bx)=\varphi^{\rm TF}(\bx)-\varphi_{r}^{\rm OTF}(\bx)+\varphi_{r}^{\rm OTF}(\bx)-\varphi^{\rm HF}(\bx)$ with $\varphi_{r}^{\rm OTF}$ the mean field potential of the OTF-problem defined in Subsection~\ref{Sotf}. By  the choice of $r$ and $D$ and Lemma~\ref{A1A2} we get since $|\bx|\geq r=|\bx|^{1+\mu}$
\begin{equation}\label{pi1}
|\varphi^{\rm TF}(\bx)-\varphi_{r}^{\rm OTF}(\bx)| \leq C |\bx|^{-4-\zeta} r^{\zeta} \, ,
\end{equation}
for $|\bx| \geq r$ with $\zeta = (7+\sqrt{73})/2$. For the other two terms we see
\begin{equation*}
\varphi^{\rm HF}(\bx)-\varphi^{\rm OTF}_{r}(\bx)= \int \frac{\rho^{\rm OTF}_{r}(\by)-\chi_{r}^{+}(\by)\rho^{\rm HF}(\by)}{|\bx -\by|} \, d\by ,
\end{equation*}
and proceeding as for small $\bx$ with the Coulomb-norm estimate Proposition \ref{Counorm}, by Lemma~\ref{extrCou} and inequality \eqref{hfex}
\begin{equation*}
|\varphi^{\rm HF}(\bx)-\varphi^{\rm OTF}_{r}(\bx)| \leq C \Big(\frac{s^{\frac15}}{r^{\frac{21}{5}}} + \frac{r^{-\frac72+\frac16}}{s^{\frac12}} + R^{-\frac14} (\alpha r^{-7})^{\frac34}  +\nu^{-1} r^{-7} +\nu^{\frac32} R^{\frac12} +\nu^3 \alpha^2 \Big).
\end{equation*}
Choosing $\nu = \beta_{0}^3 r^{-3\frac{1+\delta}{1+\mu}}$, so that $\nu \alpha \leq \kappa < 2/\pi$, $s$ such that $s^{\frac15}r^{-\frac{21}{5}} = r^{-\frac72+\frac16}s^{-\frac12}$ (i.e. $s=r^{1+\frac{5}{21}}$), and choosing $R$ such that the two terms where it appears are equal (i.e. $R=r^{2+9\frac{\delta-\mu}{1+\mu}}$; notice that $R<s$) we get
\begin{equation*}
|\varphi^{\rm HF}(\bx)-\varphi^{\rm OTF}_{r}(\bx)| \leq C ( r^{-4+\frac{1}{21}}+r^{-4+3\frac{\delta-\mu}{1+\mu}}),
\end{equation*}
since $\alpha r^{-3\frac{1+\delta}{1+\mu}}$ is bounded and $r \leq 1$. Collecting together the inequality above and \eqref{pi1} and using that $r=|\bx|^{1+\mu}$ the claim follows for $ \beta_{0} Z^{-\frac13} \leq|\bx|^{1+\delta} \leq D^{\frac{1+\delta}{1+\mu}}$. We fix $\mu =\delta/2$.

It remains to study the case of large $\bx$, i.e. $|\bx| \geq D^{\frac{1+\delta}{1+\mu}}$ with $D, \delta, \mu$ universal constants. For simplicity of notation we fix the universal constant $A:=  D^{\frac{1+\delta}{1+\mu}}$. We first notice that
\begin{equation*}
\varphi^{\rm HF}(\bx) -\varphi^{\rm TF}(\bx)= \Phi^{\rm HF}_{|\bx|}(\bx) - \Phi^{\rm TF}_{|\bx|}(\bx) + \int_{|\by| > |\bx|} \frac{\rho^{\rm TF}(\by)-\rho^{\rm HF}(\by)}{|\bx -\by|}  d\by .
\end{equation*}
The difference of the first two terms is bounded by a universal constant for $|\bx| \geq A$ by the result in Theorem~\ref{mainesti}. To estimate the last integral we split it as follows
\begin{eqnarray*}
\int_{|\by| > |\bx|} \frac{|\rho^{\rm TF}(\by)-\rho^{\rm HF}(\by)|}{|\bx -\by|} d\by & \leq & \int_{\substack{|\by| > |\bx|\\|\bx -\by|<1}} \frac{\rho^{\rm TF}(\by)}{|\bx -\by|} d\by +\int_{\substack{|\by| > |\bx|\\|\bx -\by|<1}} \frac{\rho^{\rm HF}(\by)}{|\bx -\by|} d\by\\
&& +\int_{|\by| > |\bx|} (\rho^{\rm TF}(\by) +\rho^{\rm HF}(\by)) \,  d\by .
\end{eqnarray*}
Since $|\bx| \geq A$ the third term on the right hand side is bounded by a universal constant by Lemma~\ref{extde} (for $\rho^{\rm HF}$) and Corollary~\ref{TFsotto} (for $\rho^{\rm TF}$). We estimate the first term by H\"older's inequality and Corollary~\ref{rTF53}. We get a bound on the second term proceeding as in \eqref{hfex} (using Theorem~\ref{dayau}) and choosing $\nu=\frac12$ and $R=1$. We obtain
\begin{equation*}
\int_{\substack{|\by| > |\bx|\\|\bx -\by|<1}} \frac{\rho^{\rm TF}(\by)+\rho^{\rm HF}(\by)}{|\bx -\by|} d\by \leq C ( A^{-\frac{21}{5}} + A^{-7} +\alpha^2) .
\end{equation*}
Then there exists a universal contant $A'$ such that $|\varphi^{\rm HF}(\bx) -\varphi^{\rm TF}(\bx)|  \leq A'$ for $|\bx| \geq A$.
\end{proof}


\renewcommand{\thesection}{}

\appendix\renewcommand{\thesection}{\Alph{section}}
\setcounter{lemma}{0}
\renewcommand{\thetheorem}{\Alph{section}.\arabic{theorem}}
\setcounter{equation}{0}
\renewcommand{\theequation}{\Alph{section}\arabic{equation}}

\section{Technical lemmas}\label{tech}

{\bf Proof of \eqref{BehG}}
By the definition of the function $G_{\alpha}$ the inequalities in \eqref{BehG} are equivalent to the following ones
\begin{equation}\label{natale}
\tfrac35 t^4 \min\{\tfrac25 t,1\} \leq g(t)-\tfrac83 t^3 \leq 2t^4 \min \{ \tfrac25 t ,1\} \mbox{ for }t \geq 0 .
\end{equation}
As before we use the substitution $t=\alpha (\rho/C)^{\frac13}$.

The estimates in \eqref{natale} follow directly from the study of the function $g$ separating the cases $t<\frac52$ and $ t \geq \frac52$.

{\bf Proof of Remark~\ref{K2est}}
Using the estimate on $K_2$ given in \eqref{estk2} we find
\begin{eqnarray*}
&& \iint_{\begin{array}{l}\bx \in \Sigma_r(\beta_1,\beta_2)\\
\by \in \Sigma_r(\beta_3,\beta_4) \end{array}} K_{2}( \alpha ^{-1}\vert \bx-\by\vert)^{2} \, d\bx d\by \\
& \leq &(16)^2\alpha ^{4} \iint_{\begin{array}{l}\bx \in \Sigma_r(\beta_1,\beta_2)\\ \by \in \Sigma_r(\beta_3,\beta_4) \end{array}} \frac{ e^{-\alpha^{-1}\vert \bx-\by\vert}}{|\bx -\by|^4} \;  d\bx d\by \\
 &\leq &(16)^2 \alpha ^{4} e^{-\alpha ^{-1}r(\beta _{3}-\beta_{2})} 4 \pi \int_{r(\beta_3 -\beta_2)}^{\infty} \rho^{-2}d\rho \; \int_{\Sigma_{r}(\beta_1, \beta_2)} d \bx ,
\end{eqnarray*}
since $|\bx -\by| \geq (\beta_3-\beta_2)r$. The claim follows computing the two integrals.

\subsection{Fourier transform\label{AppFou}}

In the present sub-section we present our notation for the Fourier transform (as in \cite{RS2}). Given $f\in L^{2}( \R^{3})$ we denote its Fourier transform by
\begin{equation*}
\hat{f}( \bp) =\mathcal{F}( f) (\bp) :=\tfrac{1}{( 2\pi)^{\frac{3}{2}}}\int_{\R^{3}}e^{i\bp\cdot x} f( \bx) d\bx.
\end{equation*}
Let $f,g\in L^{2}( \R^{3}) $. The following formulas hold:
\begin{enumerate}
\item $\mathcal{F}( f\ast g) ( \bp) =(2\pi) ^{\frac{3}{2}}\hat{f}( \bp) \hat{g}(\bp) ;$
\item $\mathcal{F}( fg) ( \bp ) =( 2\pi)^{-\frac{3}{2}}( \hat{f}\ast \hat{g}) ( \bp ) ;$
\item if $g( \bx ) =e^{-\lambda \vert \bx\vert ^{2}}$ then $\hat{g}( \bp ) =( 2\lambda ) ^{-\frac{3}{2}}e^{-\vert \bp \vert ^{2}/(4\lambda ) };$
\item $\vert \bx\vert ^{-\alpha }=\pi ^{\frac{\alpha }{2}}( \Gamma ( \frac{\alpha }{2}) )^{-1}\int_{0}^{+\infty }e^{-\pi\vert \bx\vert ^{2}\lambda}\lambda ^{\frac{\alpha }{2}-1}d\lambda $ for $0< \alpha <n$ (see \cite[page 130]{LL}).
\end{enumerate}
Moreover,
\begin{equation*}
\mathcal{F}\Big(\frac{f(\bx)}{|\bx|}\Big) ( \mathbf{k}) =\tfrac{1}{ 2\pi^{2}}\int_{\R^3}\frac{\hat{f}(\bp) }{|\mathbf{k}-\bp|^{2}}d\bp.
\end{equation*}

\section{Large $Z$-behavior of the energy}\label{ZZ}

In \cite{TS} the author studies the large $Z$-behavior of the ground state energy for problem \eqref{Hamiltonian}. In this work we are going to use the
same construction in several points (Lemmas~\ref{rrcou}, \ref{extGest}, Theorem~\ref{Lmnu}, ....) and with, in certain cases, a slightly different
Hamiltonian. For convenience we repeat here the main ideas of the proof. We do it as it is needed in the proof of Theorem~\ref{Lmnu} since in this case the proof is more involved. We remark that in our proof we use a localisation less than in \cite{TS}. Thanks to Theorem \ref{dayau} and \cite[Theorem 2.8]{SSS} it is sufficient to consider the region near the nuclei and the one far away from the nuclei. There is no need for an intermediate region.

\begin{proposition}\label{prthomas}
Let $Z\alpha=\kappa$ be fixed with $0\leq \kappa<2/\pi$ and $Z \geq
1$. Let us consider $\bP \in \mathbb{R}^3$, with $|\bP| \geq \beta
Z^{-\frac{1+\mu}{3}}$ for $\beta >0$ and $\mu \in (0,4/5)$. Let $Z
\geq \nu>0$ and $R>0$ be such that $R<\beta Z^{-l}/4$ for some $
\frac{1+\mu}{3}<l$. Moreover, let $\rho^{\rm TF}$ denote the
minimizer of the TF-energy functional of a neutral atom with nucleus
of charge $Z$. Consider the Hamiltonian
\begin{equation}\label{HamiltonianP}
H_{\bP}:= \sum_{i=1}^N \big(\alpha^{-1} T(\bp_i) -\frac{Z}{|\bx_i|} -\frac{\nu}{|\bx_i -\bP|} \chi_{B_R(\bP)}(\bx_i)\big) +\sum_{i<j} \frac{1}{|\bx_i -\bx_j|},
\end{equation}
acting on $\wedge_{i=1}^N L^2(\R^3; \C^q)$.

Then for all $t \in (\frac{1+\mu}{3},\min\{l,\frac{3}{5}\})$ and $\psi \in \wedge_{i=1}^N L^2(\R^3)$, with $\| \psi \|_2 =1$,
\begin{equation*}
\langle \psi, H_{\bp} \psi \rangle \geq \mathcal{E}^{\rm
TF}(\rho^{\rm TF}) - C (\beta^{\frac12}+\beta^{-2})
Z^{\frac52-\frac12 t},
\end{equation*}
with $C$ depending only on $q$ and $\kappa$.
\end{proposition}

\begin{proof}
Since $\mathcal{E}^{\rm TF}(\rho^{\rm TF})=-e_0 Z^{\frac73}$ (see
\eqref{Tfe}) to prove the claim it is sufficient to show that the
TF-energy gives a lower bound to the quantum energy modulo lower
order terms. In the proof we first reduce to a one-particle
operator. Then we localize the energy separating the contribution
from the regions near the nuclei from the contribution from the
region far away from them. Finally we study the contribution of each
of these terms. The main contribution to the energy is given by the
region far away from the nuclei. This region will give the
TF-energy.

In the following, $s=(3-t)/4$  ($t<s< 2/3$).

In the proof $C$ denotes a generic positive constant depending only on $q$ and $\kappa$.

{\it Reduction to a one-particle problem.} We are going to estimate from below $H_{\bP}$ by a one-particle operator. This allows us to consider only
Slater determinants when minimizing the energy.

Let $g \in C^{\infty}_0(\R^3)$, $g \geq 0$ be spherically symmetric
with $\supp(g)\subset B_1(0)$ and such that $\| g\|_{2}=1$. Starting
from these $g$ we define $\Phi_{s}(\bx):= (\beta/(8Z^{s}))^{-3}g^2(8
Z ^{s}\bx/ \beta)$. Then by Newton's theorem
\begin{eqnarray*}
&& \sum_{i<j} \frac{1}{|\bx_i -\bx_j|} \geq \sum_{i<j} \iint \frac{\Phi_{s}(\bx_i -\bx) \Phi_s(\bx_j -\by)}{|\bx-\by|} \; d\bx d\by =\\
& =&   \tfrac12  \sum_{i,j=1}^N \iint \frac{\Phi_{s}(\bx_i -\bx) \Phi_s(\bx_j -\by)}{|\bx-\by|} \; d\bx d\by -\tfrac{N}{2} \iint \frac{\Phi_{s}(\bx)
\Phi_s(\by)}{|\bx-\by|} \; d\bx d\by = \dots
\end{eqnarray*}
and introducing $\rho \in L^1(\R^3) \cap L^{\frac53}(\R^3)$, $\rho \geq 0$, to be chosen
\begin{eqnarray}\notag
\dots & = & \tfrac12 \int_{\R^3} \int_{\R^3} \frac{(\sum_{i=1}^N\Phi_{s}(\bx_i -\bx)-\rho(\bx)) (\sum_{j=1}^N \Phi_s(\bx_j -\by) -\rho(\by))}{|\bx-\by|} \; d\bx d\by  \\ \notag
&+& \sum_{i=1}^N \int_{\R^3} \int_{\R^3} \frac{\Phi_{s}(\bx_i -\bx)\rho(\by)}{|\bx-\by|} \; d\bx d\by - D(\rho) -\tfrac{N}{2} \int_{\R^3} \int_{\R^3} \frac{\Phi_{s}(\bx) \Phi_s(\by)}{|\bx-\by|} \; d\bx d\by \\
& \geq & \sum_{i=1}^N \rho * \Phi_{s} * \frac{1}{|\bx_i|} - D(\rho)
- C \|g^2\|^2_{\frac65} N \beta^{-1} Z^{s}. \label{m1}
\end{eqnarray}
In the last inequality we use that the first term on the left hand side of \eqref{m1} is non-negative and that
\begin{eqnarray*}
 \int_{\R^3} \int_{\R^3} \frac{\Phi_{s}(\bx) \Phi_s(\by)}{|\bx-\by|} \; d\bx d\by  & = &  C \beta^{-1} Z^{s} \int_{\R^3} \int_{\R^3} \frac{g^2(\bx) g^2(\by)}{|\bx-\by|} \; d\bx d\by \\
&\leq & C \beta^{-1} Z^{s} \|g^2\|^2_{6/5},
\end{eqnarray*}
by definition of $\Phi_s$ and Hardy-Littlewood-Sobolev's inequality. Hence
\begin{eqnarray}\notag
H_{\bP} & \geq & \sum_{i=1}^N \big(\alpha^{-1} T(\bp_i) -\frac{Z}{|\bx_i|} -\frac{\nu}{|\bx_i -\bP|} \chi_{B_R(\bP)}(\bx_i) + \rho * \Phi_{s} * \frac{1}{|\bx_i|}\big)\\
&&  - D(\rho) - C \|g^2\|^2_{\frac65} N \beta^{-1} Z^{s} .
\label{1p}
\end{eqnarray}

{\it Choice of the localization.} The localization will be given by the following functions $\chi_1, \chi_2\in C^{\infty}_0(\R^3)$:
\begin{eqnarray}\label{12345}
\chi_1(\bx):= \left\{ \begin{array}{ll}
1 &\mbox{if } |\bx|< \frac14  \beta Z^{-t}, \\
0 &\mbox{if } |\bx|> \frac12 \beta Z^{-t},  \rule{0cm}{.4cm}
\end{array} \right. &\hspace{-.5cm} & \chi_2(\bx):= \left\{ \begin{array}{ll}
1 &\mbox{if } |\bx-\bP|< \frac14 \beta Z^{-t},\\
0 &\mbox{if } |\bx-\bP|>\frac12 \beta Z^{-t}  \rule{0cm}{.4cm}
\end{array} \right.
\end{eqnarray}
and $\chi_3 \in C^{\infty}(\R^3)$ such that $\sum_{i=1}^3\chi_i^2(\bx)=1$ for all $\bx \in \R^3$. Moreover we ask that
\begin{equation}\label{grad}
\| \nabla \chi_1\|_{\infty},\| \nabla \chi_2\|_{\infty},\|\nabla
\chi_3\|_{\infty}\leq 2^5 \beta^{-1} Z^{t}.
\end{equation}

Here $t$ is the parameter given in the statement of the proposition. Notice that by the assumptions on $R$ and $\bP$ the functions defined above give
a well defined partition of unity of $\R^3$. Moreover, $B_{R}(\mathbf{P})$ is a subset of $\{ \bx \in \R^3: \chi_{2}(\bx)=1\}$.

{\it The localization in the energy expectation.} We insert now the
localization in the energy expectation. As already observed, since
we reduced the operator to a one-particle operator in the energy
expectation it is sufficient to consider Slater determinants: i.e.
$\psi= u_1 \wedge \dots \wedge u_N$ with $\{u_i\}_{i=1}^N$
orthonormal functions in $L^2(\R^3,\C^q)$. We may assume that $u_i
\in H^{\frac12}(\R^3,\C^q)$ for $i=1, \dots, N$.

{F}rom \eqref{1p} and Theorem~\ref{IMSrel} we find with $\psi= u_1 \wedge \dots \wedge u_N$
\begin{eqnarray}\notag
\langle \psi, H_{\bP}  \psi \rangle & \geq & \sum_{i=1}^N \sum_{j=1}^3 (\chi_j u_i, h \chi_j u_i )- D(\rho)- C \|g^2\|^2_{\frac65} N \beta^{-1}Z^{s}\\
&& -\alpha^{-1} \sum_{i=1}^N \sum_{j=1}^3(u_i, L_j  u_i),\label{p2}
\end{eqnarray}
with
\begin{equation*}
h:=\alpha^{-1} T(\bp) -\frac{Z}{|\cdot |} -\frac{\nu \;  \chi_{B_R(\bP)}(\cdot )}{|\cdot -\bP|} + \rho * \Phi_{s} * \frac{1}{|\cdot|},
\end{equation*}
and $L_{j}$ is the operator (defined in Theorem~\ref{IMSrel}) that gives the error due to the localization in the kinetic energy. We first estimate
this error term. Using the definition of $L_{j}$ we find for all $j \in \{1,2,3\}$, $i \in\{1, \dots,N\}$
\begin{equation*}
(u_i, L_{j} u_i) \leq \frac{\alpha^{-2}}{4 \pi^2} \| \nabla \chi_{j}\|_{\infty}^2 \iint K_2(\alpha^{-1}|\bx -\by|) |u_i(\by)||u_i(\bx)| \; d\bx d\by.
\end{equation*}
We then obtain by using Schwarz's inequality
\begin{equation}\label{aa}
\alpha^{-1} \sum_{i=1}^N \sum_{j=1}^3(u_i, L_j  u_i)\leq
\frac{\alpha^{-3}}{4 \pi^2} \sum_{j=1}^3\|\nabla
\chi_{j}\|_{\infty}^2 \sum_{i=1}^N \int K_2(\alpha^{-1}|\bz|) d\bz
\leq C N \beta^{-2} Z^{2t} ,
\end{equation}
since from \eqref{estk2}
\begin{equation}\label{k2a}
\int_{\R^3} K_2(\alpha^{-1}|\bz|) \;d\bz =\alpha^3 \int_{\R^3} K_2(|\bz|) \;d\bz = 4 \pi \alpha^3 \int_0^{\infty} t^2 K_2(t) \;dt= 6\pi^2 \alpha^3.
\end{equation}

Collecting together \eqref{p2} and \eqref{aa} we get
\begin{equation}
\langle \psi, H_{\bP}  \psi \rangle \geq \sum_{i=1}^N \sum_{j=1}^3
(\chi_j u_i, h \chi_j u_i ) - D(\rho)- C \beta^{-2} Z^{1+2t} -C
\beta^{-1} Z^{7/4-t/4} .\label{p3}
\end{equation}
Here we used that $N\leq 2Z+1$, the choice of $s$ and that we may
choose $g$ such that $\| \nabla g\|^2_2 \leq 2 \pi$.

{\it Near the nuclei.} When $j=1$ in the summation in the first term on the right hand side of \eqref{p3} we find
\begin{equation*}
\sum_{i=1}^N (\chi_1 u_i, h \chi_1 u_i ) \geq \sum_{i=1}^N (\chi_1 u_i, (\alpha^{-1} T(\bp) - \frac{Z}{|\cdot|}) \chi_1 u_i),
\end{equation*}
since $\chi_{B_R(\bP)} \chi_1 \equiv 0$ by the choice of $\chi_1$, and the term $\Phi_s * \rho * \frac{1}{|\cdot|}$ is non-negative. Then by Theorem~\ref{dayau} we find
\begin{eqnarray}\notag
\sum_{i=1}^N (\chi_1 u_i, h \chi_1 u_i ) & \geq & \Tr[\alpha^{-1}
T(\bp) - \frac{Z}{|\cdot|} \chi_{|\bx|<\frac12 \beta Z^{-t}}]_{-} \\
\label{br1} & \geq & -C \beta^{1/2} Z^{5/2-t/2} -C \kappa^{2} Z^{2}.
\end{eqnarray}

To estimate from below the term corresponding to $j=2$ in the sum on the right hand side of \eqref{p3} we use \cite[Theorem 2.8]{SSS}. Here we need the result in \cite{SSS} (instead of Theorem \ref{dayau}) because of the presence of the two nuclei. Notice that Theorem \ref{dayau} can be extended to include also different nuclei. We have
\begin{eqnarray*}
\sum_{i=1}^N (\chi_2 u_i, h \chi_2 u_i ) & \geq & \sum_{i=1}^N (\chi_2 u_i, (\alpha^{-1} T(\bp) -\frac{Z}{|\bx|}-\frac{\nu}{|\bx-\bP|} \chi_{B_{R}(\bP)}) \chi_2 u_i ) \\
& \geq & \Tr [\alpha^{-1} T(\bp)
-\frac{Z}{|\bx|}\chi_{|\bx-\bP|<\frac12 \beta
Z^{-t}}-\frac{\nu}{|\bx-\bP|} \chi_{B_{R}(\bP)}]_{-},
\end{eqnarray*}
and by  \cite[Theorem 2.8]{SSS} we get
\begin{eqnarray}\notag
\sum_{i=1}^N (\chi_2 u_i, h \chi_2 u_i ) & \geq & -C Z^{5/2}
\alpha^{1/2} -C\int_{\frac12 \beta Z^{-t}> |\bx -\bP|>\alpha}
\left(\frac{Z^{5/2}}{|\bx|^{5/2}}+\alpha^3
\frac{Z^4}{|\bx|^4}\right) \; d\bx \\ \notag
& &  -C \int_{R> |\bx -\bP|>\alpha} \left(\frac{\nu^{5/2}}{|\bx-\bP|^{5/2}} +\alpha^3\frac{\nu^4}{|\bx-\bP|^4} \right) \; d\bx\\
& \geq & -C \kappa^{1/2} Z^{2}-C  \beta^{1/2} Z^{5/2-t/2} -C
\kappa^{2} Z^{2}.\label{N2}
\end{eqnarray}
Here we used that $t <l$ and $Z \alpha =\kappa$.

{\it The outer zone.} This region gives the main contribution to the energy. The term in \eqref{p3} that we still have to study is
\begin{equation}\label{2010}
\sum_{i=1}^N  (\chi_3 u_i, h \chi_3 u_i ) - D(\rho)
\end{equation}
We start by estimating the first term in \eqref{2010} using coherent states.

We consider again the function $g \in \mathcal{C}^{\infty}_{0}(\mathbb{R}^3)$ introduced at the beginning of the proof and we define the function
\begin{equation}\label{gspr3}
g_{s}(\bx):= (\beta/(8Z^{s}))^{-\frac{3}{2}}g(8 Z^{s}\bx /\beta)=
\Phi_s^{\frac12}(\bx),
\end{equation}
with $s$ the same parameter as before. For simplicity of notation we write $\tilde{V}:= Z/|\bx|-\rho*1/|\bx|$. Then
\begin{equation*}
\frac{Z}{|\bx|} -\rho* \Phi_s *\frac{1}{|\bx|}= \tilde{V}*\Phi_s -Z \Phi_s* \frac{1}{|\bx|}+\frac{Z}{|\bx|}.
\end{equation*}
Since $\supp(g_{s}) \cap \supp(\chi_3)=\emptyset$ by Newton's Theorem we find
\begin{eqnarray}\label{bb}
 \sum_{i=1}^N  (\chi_3 u_i, h \chi_3 u_i )=  \sum_{i=1}^N  (\chi_3 u_i, (\alpha^{-1} T(\bp)-\tilde{V}*\Phi_s) \chi_3 u_i ).
\end{eqnarray}
We consider the coherent states $g_{s}^{\bp,\bq}$ defined for $\bp,\bq \in \R^3$ by
$$g_{s}^{\bp,\bq}(\bx)=g_{s}(\bx-\bq) e^{-i\bp .\bx}. $$
The following formulas hold for $f \in H^{\frac{1}{2}}(\R^3,\C)$
\begin{eqnarray}\notag
(f,f)&=&\tfrac{1}{(2 \pi)^3} \int_{\R^3} d\bp \int_{\R^3} d\bq \, (f, g_{s}^{\bp,\bq}) \,  (g_{s}^{\bp,\bq},f),\\ \label{coheformulas}
(f,V \ast g^2_{s} f)&=& \tfrac{1}{(2 \pi)^3} \int_{\R^3} d\bp \int_{\R^3} d\bq \, V(\bq) \, (f, g_{s}^{\bp,\bq}) \, (g_{s}^{\bp,\bq},f)
\end{eqnarray}
and
\begin{eqnarray}\notag
(f, T(\bp) f) & = & \tfrac{1}{(2 \pi)^3} \int_{\R^3} d\bp \int_{\R^3} d\bq \; T(\bp) \, (f, g_{s}^{\bp,\bq}) \, (g_{s}^{\bp,\bq},f)\\ \label{Lcohe}
&& - \,  \int_{\R^3} d\bx \int_{\R^3} d\bq \overline{f(\bx)} (L_{q} f)(\bx),
\end{eqnarray}
where $L_{q}$ has integral kernel
\begin{equation*}
L_{q}(\bx,\by)=\frac{\alpha^{-2}}{4 \pi^2}  |g_{s}(\bx-\bq)-g_{s}(\by-\bq)|^2 \frac{K_2(\alpha^{-1}|\bx-\by|)}{|\bx-\by|^2}.
\end{equation*}
Using these formulas we can rewrite \eqref{bb} as follows
\begin{eqnarray}\notag
&& \sum_{i=1}^N ( \chi _{3}u_{i},(\alpha^{-1} T(\bp)-\tilde{V}*\Phi_{s}) \chi _{3}u_i) \\ \notag
& = &  \tfrac{1}{(2 \pi)^3} \alpha^{-1} \int_{\R^3} d\bp \int_{\R^3} d\bq (T(\bp)-\alpha \tilde{V}(\bq))
\, \sum_{j=1}^q \sum_{i=1}^N |(\chi_3 u_i^j, g_{s}^{\bp,\bq})|^2  \\
&& - \, \alpha^{-1} \sum_{i=1}^N \int_{\R^3} d\bx \int_{\R^3} d\bq \;  \overline{\chi_3 u_i(\bx)} (L_{\bq} \chi_3 u_i)(\bx), \label{a1}
\end{eqnarray}
Here $u_{i}^{j}$ is the $j$-th spin component of $u_{i}$. We start by estimating the error term, the last term on the right hand side of \eqref{a1}.
{F}rom the definition of $L_{\bq}$ it follows
\begin{equation*}
L_{q}(\bx,\by)\leq \frac{\alpha^{-2}}{4 \pi^2} \| \nabla g_{s}\|_{\infty}^2 K_2(\alpha^{-1}|\bx-\by|)
(\chi_{\supp(g_s)}(\bx-\bq)+ \chi_{\supp(g_s)}(\by-\bq)),
\end{equation*}
and by the definition of the function $g_s$
\begin{equation*}
\int_{\R^3}L_{q}(\bx,\by) \; d\bq \leq C \| \nabla g\|^2_{\infty}
\alpha^{-2} \beta^{-2} Z^{2s} K_2(\alpha^{-1}|\bx-\by|).
\end{equation*}
By the estimate above, Schwarz's inequality, \eqref{k2a} and the choice of $s$ we find
\begin{equation}\label{br2}
\alpha^{-1} \sum_{i=1}^N \int_{\R^3} d\bx \int_{\R^3} d\bq \;
\overline{\chi_3 u_i(\bx)} (L_{\bq} \chi_3 u_i)(\bx) \leq  C \|
\nabla g\|_{\infty}^2 \beta^{-2} Z^{3/2-t/2} N.
\end{equation}

It remains to study the first term on the right hand side of
\eqref{a1}. In order to get an estimate from below we consider only
the negative part of the integrand. Moreover, since if $|\bq|<\beta
Z^{-t}/8$ then $\supp(\chi_3 g_{s}^{\bp,\bq}) = \emptyset $ (because
$Z^{-t} > Z^{-s}$ since $s>t$) we find
\begin{eqnarray}\notag
&& \tfrac{1}{(2\pi)^{3}} \alpha^{-1}\int_{\R^3} d\bp\int_{\R^3} d\bq
\; (T(\bp)-\alpha\tilde{V}(\bq)) \sum_{j=1}^q\sum_{i=1}^{N} |
(\chi_{3}u^j_{i},g_{s}^{\bp,\bq})|^{2} \\ \label{bie1pr3} & \geq &
\tfrac{q}{( 2\pi)^{3}} \alpha^{-1} \int_{|\bq|\geq \frac18\beta
Z^{-t}} d\bq\int_{T(\bp)-\alpha\tilde{V}(\bq)\leq 0} d\bp \; (
T(\bp)-\alpha\tilde{V}(\bq)) =\dots ,
\end{eqnarray}
where we also use that $\sum_{i=1}^{N} | (\chi_{3}u^j_{i},g_{s
}^{\textbf{p},\textbf{q}})|^{2} \leq 1$ (Bessel's inequality). We
split now the integral as a sum of two terms
\begin{eqnarray}\notag
\dots & = & \tfrac{q}{\left( 2\pi \right) ^{3}} \alpha^{-1}
\iint_{\substack{\frac12|\mathbf{p}|^2-\tilde{V}(\bq)\leq 0\\
|\bq|\geq \frac18\beta Z^{-t}}} d\bq d\bp\;
(T(\bp)-\alpha\tilde{V}(\bq))\\ \label{cioco} &+&
\tfrac{q}{(2\pi)^{3}} \alpha^{-1}
\iint_{\substack{\frac{\alpha}{2}|\mathbf{p}|^2\geq \alpha
\tilde{V}(\bq)\geq T(\bp)\\ |\bq|\geq \frac18\beta Z^{-t}}} d\bq
d\bp\; (T(\bp)-\alpha\tilde{V}(\bq)).
\end{eqnarray}
We consider these two terms separately. The second term in \eqref{cioco} gives a lower order contribution. Indeed
\begin{eqnarray*}
&& \tfrac{q}{(2\pi)^{3}} \alpha^{-1} \iint_{\substack{\frac{\alpha}{2}|\bp|^2\geq \alpha \tilde{V}(\bq)\geq T(\bp)\\ |\bq|\geq \frac18\beta Z^{-t}}}
d\bq d\bp\; (T(\bp)-\alpha\tilde{V}(\bq)) \\
& \geq & - \tfrac{q}{(2\pi)^{3}} \iint_{\substack{(\alpha^2
[\tilde{V}(\bq)]^2_{+}+2[\tilde{V}(\bq)]_{+})^{\frac12} \geq
|\bp|\geq (2[\tilde{V}(\bq)]_{+})^{\frac12}\\ |\bq|\geq \frac18\beta
Z^{-t}}} d\bq d\bp\; [\tilde{V}(\bq)]_{+}=\dots ,
\end{eqnarray*}
and computing the $\bp$-integral
\begin{equation*}
\dots = - C \int_{|\bq|\geq \frac18 \beta Z^{-t}} d\bq \;
[\tilde{V}(\bq)]_{+}^{\frac52}
((1+\frac{\alpha^2}{2}[\tilde{V}(\bq)]_{+} )^{\frac32}-1)=\dots.
\end{equation*}
Using $(1+x)^{\frac32}\leq 1 + \frac32 x +\frac38 x^2$ and that $[\tilde{V}(\bq)]_{+} \leq Z/|\bq|$ we get computing the integral
\begin{equation}\label{72pr2}
\left. \begin{array}{lcl}
\dots & = & - C\alpha^2 \int_{|\bq|\geq \frac18 \beta Z^{-t}} d\bq \; [\tilde{V}(\bq)]_{+}^{\frac72} (1+ \frac{\alpha^2}{8} [\tilde{V}(\bq)]_{+})\\
& \geq & -C \beta^{-\frac12} \kappa^2 Z^{3/2+t/2} - C \kappa^4
\beta^{-\frac32} Z^{1/2+3t/2}.\rule{0cm}{.6cm}
\end{array}\right.
\end{equation}
Here we use that $Z \alpha =\kappa$.

Since $\sqrt{1+x} \geq 1+x/2-x^3/8$ for all $x>0$, we have
\begin{equation*}
T(\bp)\geq \alpha \tfrac{1}{2}|\bp|^2-\alpha^3 \tfrac18 |\bp|^4,
\end{equation*}
and, for the first term on the right hand side of \eqref{cioco}, we obtain
\begin{eqnarray*}
&& \tfrac{q}{\left( 2\pi \right) ^{3}} \alpha^{-1} \iint_{\substack{\frac12|\mathbf{p}|^2-\tilde{V}(\bq)\leq 0\\ |\bq|\geq \frac18 \beta Z^{-t}}}
d\bq d\bp\; ( T(\bp)-\alpha \tilde{V}(\bq)) \geq \\
& \geq &  \tfrac{q}{(2\pi)^{3}}
\iint_{\substack{\frac12|\mathbf{p}|^2-\tilde{V}(\bq)\leq 0\\
|\bq|\geq \frac18 \beta Z^{-t}}} d\bq d\bp\; ( \tfrac12 |\bp|^2 -
\tfrac18 \alpha^2 |\bp|^4 -\tilde{V}(\bq)) = \dots .
\end{eqnarray*}
Computing now the integral with respect to $\bp$, we find
\begin{equation}\label{33opr2}
\dots = -\tfrac{2^{\frac32}q}{15 \pi^2} \int_{|\bq|>\frac18 \beta
Z^{-t}} [\tilde{V}(\bq)]_{+}^{\frac{5}{2}} \, d\bq- C \alpha^2
\int_{|\bq|>\frac18 \beta Z^{-t}} [\tilde{V}(\bq)]_{+}^{\frac{7}{2}}
\, d\bq.
\end{equation}
We see that the second term on the right hand side of \eqref{33opr2} gives a lower order contribution since it is of the same order
as the one in \eqref{72pr2}.

Collecting together \eqref{p3}, \eqref{br1}, \eqref{N2}, \eqref{bb}, \eqref{a1}, \eqref{br2}, \eqref{72pr2} and \eqref{33opr2}
\begin{eqnarray}\label{p4pr2}
\langle \psi, H_{\bP} \psi \rangle & \geq & - C(\beta^{\frac12}
+\beta^{-2}) Z^{5/2-t/2} -\tfrac{2^{\frac32}q}{15 \pi^2} \int_{\R^3}
[\tilde{V}(\bq)]_{+}^{\frac{5}{2}} \, d\bq - D(\rho) \, .
\end{eqnarray}
Here we used also that $N <2Z +1$, the choice of $s$ and that $t
\leq 3/5$.

Now we choose $\rho=\rho^{\rm TF}$ the minimizer of the TF-energy functional of a neutral atom with Coulomb potential
and nuclear charge $Z$. Hence $\rho^{\rm TF}$ satisfies the TF-equation
\begin{equation*}
\tfrac{1}{2} \big(\tfrac{6 \pi^2}{q}\big)^{\frac23} \rho^{\rm TF}(\bx)^{\frac23}=[\tilde{V}(\bx)]_{+},
\end{equation*}
since $\tilde{V}$ is the TF-mean field potential. Notice that here we use that the chemical potential of a neutral atom
is zero. By the choice of $\rho$ from the TF-equation it follows from \eqref{p4pr2} that
\begin{eqnarray*}
\langle \psi, H_{\bP} \psi \rangle  & \geq  &- C (\beta^{\frac12} +\beta^{-2}) Z^{5/2-t/2} + \tfrac{3}{10} \big(\tfrac{6 \pi^2}{q}\big)^{\frac23} \int_{\R^3} d\bx \; \rho^{\rm TF}(\bx)^{\frac53}\\
& & - Z \int_{\R^3}\frac{\rho^{\rm TF}(\bx)}{|\bx|} \; d\bx +D(\rho^{\rm TF})\\
& = & \mathcal{E}^{\rm TF}(\rho^{\rm TF}) -C (\beta^{\frac12}
+\beta^{-2}) Z^{5/2-t/2} \, .
\end{eqnarray*}
The claim follows.
\end{proof}

\begin{proposition}\label{prthomas3}
Let $\rho^{\rm TF}$ be the minimizer of the TF-energy functional of
a neutral atom with nuclear charge $Z$. Let $Z \alpha=\kappa$ be
fixed with $0 \leq \kappa < 2/\pi$ and $Z \geq 1$.

Then there is a constant depending only on $\kappa$ and $q$ such that for all $\{u_i\}_{i=1}^N \subset H^{\frac12}(\R^3; \C^{q})$ orthonormal in $L^2(\R^3)$ we have
\begin{equation*}
\sum_{i=1}^N (u_i, (\alpha^{-1}T(\bp) -\varphi^{\rm
TF})u_i)-D(\rho^{\rm TF}) \geq \mathcal{E}^{\rm TF}(\rho^{\rm TF}) -
C Z^{2+\frac15}\, ,
\end{equation*}
with $D(\cdot)=D(\cdot, \cdot)$ the Coulomb scalar product.
\end{proposition}

\begin{proof}
Since $\mathcal{E}^{\rm TF}(\rho^{\rm TF})=-e_0 Z^{\frac73}$ (see
\eqref{Tfe}) to prove the claim it is sufficient to show that the
TF-energy gives a lower bound to the quantum energy modulo lower
order terms. In the proof we localize the energy separating the
contribution from the region near the nucleus to the one far away.
The region far away from the nuclei will give the TF-energy.

In the proof $C$ denotes a generic universal positive constant.

{\it Choice of the localization.} The localization will be given by the functions $\chi_1 \in C^{\infty}_0(\R^3)$ and $\chi_2 \in C^{\infty}(\R^3)$ such that: $0\leq \chi_1,\chi_2 \leq 1$, $\chi_1^2+\chi_2^2=1$ in $\mathbb{R}^3$,
\begin{eqnarray}\label{12}
\chi_1(\bx):= \left\{ \begin{array}{ll}
1 &\mbox{if } |\bx|< 2 Z^{-3/5}, \\
0 &\mbox{if } |\bx|>3 Z^{-3/5}.
\end{array}\right.
\end{eqnarray}
Moreover we ask that
\begin{equation}\label{gradp3}
\left.\begin{array}{l} \| \nabla \chi_1\|_{\infty},\| \nabla
\chi_2\|_{\infty} \leq 2^2 Z^{3/5}.
\end{array}\right.
\end{equation}

{\it The localization in the energy expectation.} We insert now the localization in the energy expectation. {F}rom Theorem~\ref{IMSrel} we find
\begin{eqnarray}\label{locp2}
&& \sum_{i=1}^N (u_i, (\alpha^{-1}T(\bp) -\varphi^{\rm TF})u_i)-D(\rho^{\rm TF}) \\
& \geq  & \sum_{i=1}^N \sum_{j=1}^2(\chi_j u_i, (\alpha^{-1}T(\bp) -\varphi^{\rm TF})\chi_j u_i)-D(\rho^{\rm TF})
-\alpha^{-1} \sum_{i=1}^N \sum_{j=1}^2(u_i, L_j  u_i), \notag
\end{eqnarray}
with $L_{j}$ is the operator (defined in Theorem~\ref{IMSrel}) that
gives the error due to the localization in the kinetic energy. We
first estimate this error term. Since $N \leq 2Z +1$ we find as in
\eqref{aa} that
\begin{equation}\label{locp3a}
\alpha^{-1} \sum_{i=1}^N \sum_{j=1}^2(u_i, L_{j} u_i) \leq C Z^{6/5}
N  \leq C Z^{2+1/5} \, .
\end{equation}

{\it Near the nucleus.} Since
\begin{equation*}
\sum_{i=1}^N (\chi_1 u_i, (\alpha^{-1}T(\bp) -\varphi^{\rm
TF})\chi_1 u_i) \geq \Tr[\alpha^{-1} T(\bp) -\varphi^{\rm TF}
\chi_{|\bx|<3 Z^{-3/5}}]_{-},
\end{equation*}
by Theorem~\ref{dayau} with $R=3 Z^{-3/5}$ we find
\begin{equation}\label{npr3}
\sum_{i=1}^N (\chi_1 u_i, (\alpha^{-1}T(\bp) -\varphi^{\rm
TF})\chi_1 u_i) \geq -C Z^{2+1/5} -C \kappa^{2} Z^{2}.
\end{equation}
Here we use that $Z \alpha =\kappa$.

{\it The outer zone.} This region gives the main contribution to the energy.

Let $g \in C^{\infty}_0(\R^3)$, $g \geq 0$ be spherically symmetric
with $\supp(g)\subset B_1(0)$ and such that $\| g\|_{2}=1$. Starting
from these $g$ we define $\Phi_{Z}(\bx):= ( Z^{-3/5})^{-3}g^2(\bx
Z^{3/5})$ and
\begin{equation*}
g_{Z}(\bx):= (Z^{-3/5})^{-\frac{3}{2}}g(\bx Z^{3/5})=
\Phi_{Z}^{\frac12}(\bx).
\end{equation*}
Since $\supp(g_{Z}) \cap \supp(\chi_2)=\emptyset$ by Newton's
Theorem we find
\begin{eqnarray}\label{bbpr2}
 \sum_{i=1}^N  (\chi_2 u_i, (\alpha^{-1}T(\bp) -\varphi^{\rm TF})\chi_2 u_i )=
 \sum_{i=1}^N  (\chi_2 u_i, (\alpha^{-1}T(\bp) -\varphi^{\rm TF} * \Phi_{Z})\chi_2 u_i ).
\end{eqnarray}
We consider the coherent states $g_{Z}^{\bp,\bq}$ defined for
$\bp,\bq \in \R^3$ by
$$g_{Z}^{\bp,\bq}(\bx)=g_{Z}(\bx-\bq) e^{-i\bp .\bx}. $$
Using formulas \eqref{coheformulas} and \eqref{Lcohe} we can rewrite \eqref{bbpr2} as follows
\begin{eqnarray}\notag
&& \sum_{i=1}^N ( \chi _{2}u_{i},(\alpha^{-1} T(\bp)-\varphi^{\rm
TF}*g^2_{Z}) \chi _{2}u_i) \\ \notag
& = &  \tfrac{1}{(2 \pi)^3} \alpha^{-1} \int_{\R^3} d\bp \int_{\R^3} d\bq (T(\bp)-\alpha \varphi^{\rm TF}(\bq))
\, \sum_{j=1}^q \sum_{i=1}^N |(\chi_2 u_i^j, g_{Z}^{\bp,\bq})|^2  \\
&& - \, \alpha^{-1} \sum_{i=1}^N \int_{\R^3} d\bx \int_{\R^3} d\bq \;  \overline{\chi_2 u_i(\bx)} (L_{\bq} \chi_2 u_i)(\bx), \label{a1pr2}
\end{eqnarray}
Here $u_{i}^{j}$ is the $j$-th spin component of $u_{i}$. We start by estimating the error term, the last term on the right hand side of \eqref{a1pr2}. We find as in \eqref{br2} that
\begin{equation}\label{ce}
\alpha^{-1} \sum_{i=1}^N \int_{\R^3} d\bx \int_{\R^3} d\bq \;
\overline{\chi_2 u_i(\bx)} (L_{\bq} \chi_2 u_i)(\bx) \leq  C \|
\nabla g\|_{\infty}^2 Z^{6/5} N.
\end{equation}

It remains to study the first term on the right hand side of
\eqref{a1pr2}. In order to get an estimate from below we consider
only the negative part of the integrand. Moreover, since if $|\bq|<
Z^{-3/5}$ then $\supp(\chi_2 g_{Z}^{\bp,\bq}) = \emptyset $ we find
\begin{eqnarray}\notag
&& \tfrac{1}{(2\pi)^{3}} \alpha^{-1}\int_{\R^3} d\bp\int_{\R^3} d\bq
\; (T(\bp)-\alpha\varphi^{\rm TF}(\bq)) \sum_{j=1}^q\sum_{i=1}^{N} |
(\chi_{2}u^j_{i},g_{Z}^{\bp,\bq})|^{2} \\ \label{bie1} & \geq &
\tfrac{q}{( 2\pi)^{3}} \alpha^{-1} \int_{|\bq|\geq Z^{-3/5}}
d\bq\int_{T(\bp)-\alpha\varphi^{\rm TF}(\bq)\leq 0} d\bp \; (
T(\bp)-\alpha\varphi^{\rm TF}(\bq)) =\dots ,
\end{eqnarray}
where we also use that $\sum_{i=1}^{N} |
(\chi_{3}u^j_{i},g_{Z}^{\textbf{p},\textbf{q}})|^{2} \leq 1$
(Bessel's inequality). We split now the integral as a sum of two
terms
\begin{eqnarray}\notag
\dots & = & \tfrac{q}{\left( 2\pi \right) ^{3}} \alpha^{-1} \iint_{\substack{\frac12|\mathbf{p}|^2-\varphi^{\rm TF}(\bq)\leq 0\\
|\bq|\geq Z^{-3/5}}} d\bq d\bp\; (T(\bp)-\alpha\varphi^{\rm
TF}(\bq))\\ \label{ciocopr3} &+& \tfrac{q}{(2\pi)^{3}} \alpha^{-1}
\iint_{\substack{\frac{\alpha}{2}|\mathbf{p}|^2\geq \alpha
\varphi^{\rm TF}(\bq)\geq T(\bp)\\ |\bq|\geq Z^{-3/5}}} d\bq d\bp\;
(T(\bp)-\alpha\varphi^{\rm TF}(\bq)).
\end{eqnarray}
We consider these two terms separately. The second term in \eqref{ciocopr3} gives a lower order contribution. Indeed
\begin{eqnarray*}
&& \tfrac{q}{(2\pi)^{3}} \alpha^{-1} \iint_{\substack{\frac{\alpha}{2}|\bp|^2\geq \alpha \varphi^{\rm TF}(\bq)\geq T(\bp)\\
|\bq|\geq Z^{-3/5}}} d\bq d\bp\; (T(\bp)-\alpha\varphi^{\rm TF}(\bq)) \\
& \geq & - \tfrac{q}{(2\pi)^{3}} \iint_{\substack{(\alpha^2
[\varphi^{\rm TF}]^2_{+}+2[\varphi^{\rm TF}]_{+})^{\frac12} \geq
|\bp|\geq (2[\varphi^{\rm TF}(\bq)]_{+})^{\frac12}\\ |\bq|\geq
Z^{-3/5}}} d\bq d\bp\; [\varphi^{\rm TF}(\bq)]_{+}=\dots ,
\end{eqnarray*}
and computing the integral in $\bp$
\begin{equation*}
\dots = - C \int_{|\bq|\geq Z^{-3/5}} d\bq \; [\varphi^{\rm
TF}(\bq)]_{+}^{\frac52} ((1+\frac{\alpha^2}{2}[\varphi^{\rm
TF}(\bq)]_{+} )^{\frac32}-1)=\dots.
\end{equation*}
Using $(1+x)^{\frac32}\leq 1 + \frac32 x +\frac38 x^2$ and that $[\varphi^{\rm TF}(\bq)]_{+} \leq Z/|\bq|$ we get computing the integral
\begin{equation}\label{72}
\left. \begin{array}{lcl}
\dots & = & - C \alpha^2 \displaystyle{\int}_{|\bq|\geq Z^{-3/5}} d\bq \;
[\varphi^{\rm TF}]_{+}^{\frac72} (1+ \frac{\alpha^2}{8} [\varphi^{\rm TF}(\bq)]_{+})\\
& \geq & -C \kappa^{2} Z^{2-\frac15} - C \kappa^{4}
Z^{\frac75}.\rule{0cm}{.6cm}
\end{array}\right.
\end{equation}
Since $\sqrt{1+x} \geq 1+x/2 -x^3/8$ for all $x \geq 0$, we have
\begin{equation*}
T(\bp)\geq \alpha \tfrac{1}{2}|\bp|^2-\alpha^3 \tfrac18 |\bp|^4,
\end{equation*}
and, for the first term on the right hand side of \eqref{ciocopr3}, we obtain
\begin{eqnarray*}
&& \tfrac{q}{\left( 2\pi \right) ^{3}} \alpha^{-1} \iint_{\substack{\frac12|\mathbf{p}|^2-\varphi^{\rm TF}(\bq)\leq 0\\ |\bq|\geq Z^{-3/5}}}
d\bq d\bp\; ( T(\bp)-\alpha \varphi^{\rm TF}(\bq)) \geq \\
& \geq &  \tfrac{q}{(2\pi)^{3}}
\iint_{\substack{\frac12|\mathbf{p}|^2-\varphi^{\rm TF}(\bq)\leq 0\\
|\bq|\geq Z^{-3/5}}} d\bq d\bp\; ( \tfrac12 |\bp|^2 - \tfrac18
\alpha^2 |\bp|^4 -\varphi^{\rm TF}(\bq)) = \dots .
\end{eqnarray*}
Computing now the integral with respect to $\bp$, we find
\begin{equation}\label{33o}
\dots = -\tfrac{2^{\frac32}q}{15 \pi^2} \int_{|\bq|> Z^{-3/5}}
[\varphi^{\rm TF}(\bq)]_{+}^{\frac{5}{2}} \, d\bq- C \alpha^2
\int_{|\bq|>Z^{-3/5}} [\varphi^{\rm TF}(\bq)]_{+}^{\frac{7}{2}} \,
d\bq.
\end{equation}
We see that the second term on the right hand side of \eqref{33o} gives a lower order contribution since it is of the same order as the one in \eqref{72}.

Starting from \eqref{locp2}, by \eqref{locp3a}, \eqref{npr3},
\eqref{ce}, \eqref{72} and \eqref{33o} we find
\begin{eqnarray}\label{p4}
&&  \sum_{i=1}^N (u_i, (\alpha^{-1}T(\bp) -\varphi^{\rm
TF})u_i)-D(\rho^{\rm TF})\\ \notag & \geq & - C (Z^{2+1/5}+ Z^{2} +
Z^{2-1/5}+Z^{7/5})-\tfrac{2^{\frac32}q}{15 \pi^2} \int_{\R^3}
[\varphi^{\rm TF}(\bq)]_{+}^{\frac{5}{2}} \, d\bq - D(\rho^{\rm
TF}).
\end{eqnarray}
The result follows from the TF-equation.
\end{proof}

\vspace{.5cm}

\begin{acknowledgement}
The authors wish to thank Heinz Siedentop for suggesting the problem.
Support from the EU IHP network
{\it Postdoctoral Training Program in Mathematical Analysis of
Large Quantum Systems},
contract no.\
HPRN-CT-2002-00277
is gratefully acknowledged.
\end{acknowledgement}

\end{document}